\documentclass{conm-p-l}
\usepackage{amsmath}
\usepackage{amsthm}
\usepackage{amscd}
\usepackage{tikz}
\usetikzlibrary{matrix,arrows}
\usepackage{amssymb}
\usepackage{dsfont}
\usepackage[T1]{fontenc}
\usepackage{url}
\usepackage[shortlabels]{enumitem}

\newtheoremstyle{standard}
  {}
  {}
  {\rmfamily\mdseries\itshape}
  {\parindent}
  {\rmfamily\mdseries\scshape}
  {.}
  {.7em}
  {\thmnumber{#2. }\thmname{#1}\thmnote{ \textmd{(#3)}}
}

\newtheoremstyle{alt}
  {}
  {}
  {\rmfamily\mdseries\upshape}
  {\parindent}
  {\rmfamily\mdseries\scshape}
  {.}
  {.7em}
  {\thmnumber{#2. }\thmname{#1}\thmnote{ \textmd{(#3)}}
}

\newcounter{entry}[section]

\renewcommand{\theentry}{\thesection.\arabic{entry}}
\newcommand{\entry}[1][\hspace{-.5em}]{%
 \vspace{.5\baselineskip}\par%
 \refstepcounter{entry}%
 {\rmfamily\mdseries\upshape\theentry.\hspace{.5em}\itshape #1\hspace{.5em}---\hspace{.6em}}%
}

\newcommand{\mysubsection}[1]{%
 \vspace{.5\baselineskip}\par%
 {\rmfamily\bfseries\upshape #1.}%
 \addcontentsline{toc}{section}{$\qquad\;$\emph{#1}}
 \nopagebreak
 \par
}

\theoremstyle{standard}
\newtheorem{theorem}[entry]{Theorem}
\newtheorem{lemma}[entry]{Lemma}
\newtheorem{proposition}[entry]{Proposition}
\newtheorem{corollary}[entry]{Corollary}
\newtheorem{definition}[entry]{Definition}
\newtheorem{lem-def}[entry]{Lemma-definition}
\newtheorem{prop-def}[entry]{Proposition-definition}
\newtheorem*{theorem*}{Theorem}
\newtheorem*{lemma*}{Lemma}
\newtheorem*{proposition*}{Proposition}
\newtheorem*{corollary*}{Corollary}
\newtheorem*{definition*}{Definition}
\newtheorem*{lem-def*}{Lemma-definition}
\newtheorem*{prop-def*}{Proposition-definition}

\theoremstyle{alt}
\newtheorem{example}[entry]{Example}
\newtheorem{remark}[entry]{Remark}
\newtheorem*{example*}{Example}
\newtheorem*{remark*}{Remark}

\numberwithin{equation}{entry}
\renewcommand{\theequation}{\theentry.\arabic{equation}}
\renewcommand{\theenumi}{\roman{enumi}}
\renewcommand{\labelenumi}{(\theenumi)}

\newcommand{\beq}{\begin{equation*}}
\newcommand{\eeq}{\end{equation*}}

\DeclareMathOperator{\im}{Im}
\DeclareMathOperator{\tr}{tr}
\DeclareMathOperator{\rk}{rk}
\DeclareMathOperator{\car}{char}
\DeclareMathOperator{\dist}{dist}
\DeclareMathOperator{\dmin}{d_{min}}
\DeclareMathOperator{\ddual}{d^\perp}
\DeclareMathOperator{\Supp}{Supp}
\DeclareMathOperator{\Spec}{Spec}
\DeclareMathOperator{\Proj}{Proj}
\DeclareMathOperator{\Sym}{Sym}
\DeclareMathOperator{\Aut}{Aut}

\DeclareMathOperator{\End}{End}

\newcommand{\Fq}{{\mathbb{F}_q}}
\newcommand{\N}{\mathbb{N}}
\newcommand{\Z}{\mathbb{Z}}
\newcommand{\Q}{\mathbb{Q}}

\newcommand{\F}{\mathbb{F}}
\newcommand{\K}{\mathbb{K}}
\newcommand{\PP}{\mathbf{P}}
\newcommand{\fJ}{\mathfrak{J}}
\newcommand{\fS}{\mathfrak{S}}
\newcommand{\cA}{\mathcal{A}}
\newcommand{\cB}{\mathcal{B}}
\newcommand{\cC}{\mathcal{C}}
\newcommand{\cG}{\mathcal{G}}
\newcommand{\cL}{\mathcal{L}}
\newcommand{\cO}{\mathcal{O}}
\newcommand{\cP}{\mathcal{P}}
\newcommand{\cQ}{\mathcal{Q}}
\newcommand{\cR}{\mathcal{R}}
\newcommand{\cS}{\mathcal{S}}
\newcommand{\cT}{\mathcal{T}}
\newcommand{\cU}{\mathcal{U}}

\newcommand{\ie}{\emph{i.e. }}
\newcommand{\eg}{\emph{e.g. }}
\newcommand{\longto}{\longrightarrow}
\newcommand{\tens}{\otimes}
\newcommand{\moins}{\setminus}
\newcommand{\implique}{\Longrightarrow}
\newcommand{\equivaut}{\Longleftrightarrow}
\newcommand{\deux}[1][2]{^{\langle #1\rangle}}
\newcommand{\abs}[1]{\lvert #1\rvert}
\newcommand{\un}{\mathds{1}}
\newcommand{\inj}{\hookrightarrow}
\newcommand{\surj}{\twoheadrightarrow}
\newcommand{\SFrob}[1][\cdot]{S^{\,#1}_{\!Frob}}
\newcommand{\itemref}{\ref}

\let\Eulerphi\phi
\let\phi\varphi
\let\epsilon\varepsilon
\let\subset\subseteq

\title[On products and powers of linear codes]{On products and powers of linear codes under componentwise multiplication}
\author{Hugues Randriambololona}
\date{}

\begin{document}

\begin{abstract}
In this text we develop the formalism of products and
powers of linear codes under componentwise multiplication.
As an expanded version of the author's talk at AGCT-14,
focus is put mostly on basic properties and descriptive statements
that could otherwise probably not fit in a regular research paper.
On the other hand, more advanced results and applications are only quickly mentioned with references to the literature.
We also point out a few open problems.

Our presentation alternates between two points of view, which the theory intertwines
in an essential way:
that of combinatorial coding, and that of algebraic geometry. 

In appendices that can be read independently, we investigate topics in multilinear algebra over finite fields,
notably we establish a criterion for a symmetric multilinear map to admit a symmetric algorithm, or equivalently,
for a symmetric tensor to decompose as a sum of elementary symmetric tensors.
\end{abstract}

\maketitle

\vspace{-\baselineskip}

\tableofcontents

\vspace{-\baselineskip}

\textbf{Notations and conventions.}
In the first three sections of this text we will be working over an arbitrary field $\F$,
although we will keep in mind the case where $\F=\Fq$ is the finite field with $q$ elements.

If $V$ is a vector space over $\F$, we denote by $V^\vee$ its dual, that is, the vector space of all linear forms $V\longto\F$.
If $X\subset V$ is an arbitrary subset, we denote by $\langle X\rangle$ its linear span in $V$.
We let $S^\cdot V=\bigoplus_{t\geq0}S^tV$ be the symmetric algebra of $V$, which is the
the largest commutative graded quotient algebra of the tensor algebra of $V$.
In particular, the $t$-th symmetric power $S^tV$ is a \emph{quotient} of $V^{\tens t}$,
and should not be confused with $\Sym^tV$, the space of symmetric tensors of order $t$, 
which is a \emph{subspace} of $V^{\tens t}$. 
If $W$ is another vector space, we also let $\Sym^t(V;W)$ be the space of symmetric $t$-multilinear
maps from $V^t$ to $W$.
All these objects are related by natural identifications such as
\beq
\Sym^t(V^\vee)=\Sym^t(V;\F)=(S^tV)^\vee.
\eeq
Here it was always understood that we were working over $\F$, but in case of ambiguity we will use more precise
notations such as $S^\cdot_{\F}V$, $\Sym^t_{\F}(V;W)$, etc.

By $[n]$ we denote the standard set with $n$ elements,
the precise definition of which will depend on the context:
when doing combinatorics, it will be $[n]=\{1,2,\dots,n\}$,
and when doing algebraic geometry over $\F$, it will
be $[n]=\Spec\F^n$.
We also let $\fS_n$ be the symmetric group on $n$ elements, which acts naturally on $[n]$.

By a linear code of length $n$ over $\F$ we mean a linear subspace $C\subset\F^n$; moreover if $\dim(C)=k$ we say $C$ is a $[n,k]$ code.
Given a word $x\in\F^n$, its support $\Supp(x)\subset[n]$ is the set of indices over which $x$ is nonzero, and its Hamming weight is $w(x)=\abs{\Supp(x)}$.
If $S\subset[n]$ is a subset we let $1_S\in\F^n$ be the characteristic vector of $S$,
that is, the vector with coordinates $1$ over $S$ and $0$ over $[n]\moins S$.
We also let $\pi_S:\F^n\surj\F^S$ be the natural linear projection
and $\iota_S:\F^S\inj\F^n$ be the natural linear inclusion of vector spaces.

The dual code $C^\perp\subset\F^n$ is defined as the orthogonal of $C$ with respect to the standard scalar product in $\F^n$.
One should be careful not to confuse this notion with that of the dual vector space $C^\vee$.



\section{Introduction}
\label{sect_intro}

\mysubsection{Basic definitions}
\entry
Given a field $\F$ and an integer $n\geq1$, we let $*$ denote componentwise
multiplication in $\F^n$, so
\beq
(x_1,\dots,x_n)*(y_1,\dots,y_n)=(x_1y_1,\dots,x_ny_n)
\eeq
for $x_i,y_j\in\F$.

This makes $\F^n$ a commutative $\F$-algebra,
with unit element the all-$1$ vector $1_{[n]}=(1,\dots,1)$.
Its group of invertible elements is $(\F^n)^\times=(\F^\times)^n$.

This algebra $\F^n$ can also be identified with the algebra of diagonal
matrices of size $n$ over $\F$
(see~\ref{discussion_trace} and~\ref{monomial_tsfo} for more).

\entry
\label{def_dot*}
If $S,S'\subset\F^n$ are two subsets, that is, two (nonlinear) codes of the same length~$n$,
we let
\beq
S\,\dot{*}\,S'=\{ c*c'\,;\;(c,c')\in S\times S'\}\;\subset\,\F^n
\eeq
be the set of componentwise products of their elements.

This operation $\dot{*}$ is easily seen to be commutative and associative,
and distributive with respect to the union of subsets:
$S\,\dot{*}\,(S'\cup S'')=(S\,\dot{*}\,S')\cup(S\,\dot{*}\,S'')$
for all $S,S',S''\subset\F^n$. 
This means that the set of (nonlinear) codes over $\F$ of some given length~$n$
becomes a commutative \emph{semiring} under these laws $\cup,\dot{*}$,
with zero element the empty code $\emptyset$ and with unit element the singleton $\{1_{[n]}\}$.

This semiring is in fact an \emph{ordered} semiring (under the inclusion relation $\subset$), since these laws
are obviously compatible with $\subset$.

\entry
\label{def_*}
If moreover $C,C'\subset\F^n$ are two linear subspaces, that is,
two linear codes of the same length~$n$, we let
\beq
C*C'=\langle C\,\dot{*}\,C'\rangle\;\subset\,\F^n
\eeq
be the linear span of $C\,\dot{*}\,C'$.
(The reader should be careful of the shift in notations from \cite{agis}.)

For $x\in\F^n$, we also write $x*C=\langle x\rangle*C=\{x*c\,;\;c\in C\}$.

Some authors also call our $C*C'$ the Schur product of $C$ and $C'$.
It is easily seen that this operation $*$, defined on pairs
of linear codes of a given length $n$, is commutative and associative.
Given linear codes $C_1,\dots,C_t\subset\F^n$, their product
$C_1*\cdots*C_t\subset\F^n$ is then the linear code spanned by componentwise
products $c_1*\cdots*c_t$ for $c_i\in C_i$.

Also, given three linear codes $C,C',C''\subset\F^n$, we have the distributivity
relation
\beq
C*(C'+C'')\,=\,C*C'\,+\,C*C''\;\subset\,\F^n
\eeq
where $+$ is the usual sum of subspaces in $\F^n$.

This means that the set of linear codes over $\F$ of some given length~$n$
becomes a commutative \emph{semiring} under these operations $+,*$.
The zero element of this semiring is the zero subspace, and its
unit element is the one-dimensional repetition code $\un$
(generated by the all-$1$ vector $1_{[n]}$).
And as above, this semiring is easily seen to be an \emph{ordered} semiring (under inclusion).

The $t$-th power of an element $C$ in
this semiring will be denoted $C\deux[t]$.
For instance, $C\deux[0]=\un$, $\,$ $C\deux[1]=C$, $\,$ $C\deux=C*C$,
and we have the usual relations $C\deux[t]*C\deux[t']=C\deux[t+t']$,
$\,$ $(C\deux[t])\deux[t']=C\deux[tt']$.

\begin{definition}
\label{def_sequences}
Let $C\subset\F^n$ be a linear code.
The sequence of integers 
\beq
\dim(C\deux[i]),\quad i\geq0
\eeq
is called the
\emph{dimension sequence}, or the \emph{Hilbert sequence},
of $C$.

The sequence of integers 
\beq
\dmin(C\deux[i]),\quad i\geq0
\eeq
is called the
\emph{distance sequence} of $C$.

Occasionally we will also consider the \emph{dual distance sequence}, $\ddual(C\deux[i])=\dmin((C\deux[i])^\perp)$, $i\geq0$.
\end{definition}

Probably the more self-describing term ``dimension sequence'' should be preferred over ``Hilbert sequence,'' although two very good reasons for using the latter would be: 
\begin{enumerate}[(i)]
\item emphasis on the geometric analogy that will be given in Proposition~\ref{eq_def}
\item pedantry.
\end{enumerate}

One can show (\cite[Prop.~11]{agis}, or Theorem~\ref{monotone} below) that the dimension sequence of a nonzero linear code is non-decreasing,
so it becomes ultimately constant. This allows the following:

\begin{definition}
\label{def_CMreg}
The (Castelnuovo-Mumford) \emph{regularity} of a nonzero linear code $C\subset\F^n$
is the smallest integer $r=r(C)\geq0$ such that
\beq
\dim(C\deux[r])=\dim(C\deux[r+i]),\quad \forall i\geq0.
\eeq
\end{definition}

\begin{example}
\label{exRS}
For $2\leq k\leq n\leq\abs{\F}$, the $[n,k]$ Reed-Solomon code $C$ obtained by evaluating polynomials of degree up to $k-1$ at $n$ distinct elements of $\F$, has dimension sequence
\beq
1,\,k,\,2k-1,\,\dots,\,\left\lfloor\frac{n-1}{k-1}\right\rfloor\!(k-1)+1,\,n,\,n,\,n,\,\dots
\eeq
and distance sequence
\beq
n,\,n-k+1,\,n-2k+2,\,\dots,\,n-\left\lfloor\frac{n-1}{k-1}\right\rfloor\!(k-1),\,1,\,1,\,1,\,\dots
\eeq
and regularity
\beq
r(C)=\left\lceil\frac{n-1}{k-1}\right\rceil.
\eeq
\end{example}

\begin{example}
\label{exAG}
More generally, one class of linear codes that behave particularly well with respect to the operation $*$,
is that of evaluation codes.
For example if $RM_\F(r,m)$ is the generalized Reed-Muller code, obtained
by evaluating polynomials in $m\geq2$ variables, of total degree up to $r\geq0$,
at all points of $\F^m$, then we have
\beq
RM_\F(r,m)*RM_\F(r',m)=RM_\F(r+r',m).
\eeq
Likewise for algebraic-geometry codes we have
\beq
C(D,G)*C(D',G)\subset C(D+D',G)
\eeq
where $C(D,G)$ is the code obtained by evaluating functions from
the Riemann-Roch space $L(D)$, associated with a divisor $D$,
at a set $G$ of $\F$-points out of the support of $D$, on an algebraic curve
over $\F$. Note that here one can construct examples in which the inclusion is strict (even when $D,D'$ are taken effective):
indeed $L(D+D')$ contains $L(D)L(D')$ but need not be spanned by it --- see \eg \cite{Mumford_quad}, \S\S1-2, for a study of questions of this sort.
\end{example}

\entry
There is a natural \emph{semiring morphism}, from the $\cup,\dot{*}\,$-semiring of (nonlinear)
codes of length $n$ over $\F$, to the $+,*\,$-semiring of linear codes of length $n$ over $\F$,
mapping a subset $S\subset\F^n$ to its linear span $\langle S \rangle$
(moreover this map respects the oredred semiring structures defined by inclusion).

Last, there is another semiring structure, defined
by the operations $\oplus,\tens$ on the set of \emph{all} linear
codes over $\F$ (not of fixed length; see \eg \cite{Slepian}).
One should not confuse these constructions, although we will see some
more relations between them in~\ref{tens_et_*} and~\ref{def_decomposable} below.

\entry
\label{intersecting_codes}
Also there are links between products of codes and the theory of intersecting codes~\cite{CL}\cite{21sep}, but one should be careful to avoid certain misconceptions.
Given two linear codes $C_1,C_2\subset\F^n$, we define their intersection number
\beq
i(C_1,C_2)=\min_{\substack{c_1\in C_1,\,c_2\in C_2\\ c_1,c_2\neq0}}w(c_1*c_2).
\eeq
We say the pair $(C_1,C_2)$ is intersecting if $i(C_1,C_2)>0$, that is, if any two nonzero codewords from $C_1$ and $C_2$ have intersecting supports.
And given an integer $s>0$, we say $(C_1,C_2)$ is $s$-intersecting if $i(C_1,C_2)\geq s$.
(We also say a linear code $C$ is ($s$-)intersecting if the pair $(C,C)$ is.)

Now the two quantities $\dmin(C_1*C_2)$ and $i(C_1,C_2)$ might seem related, but there are some subtleties:
\begin{itemize}
\item Let $s>0$ and suppose $(C_1,C_2)$ is $s$-intersecting. Then this does not necessarily imply $\dmin(C_1*C_2)\geq s$.
More precisely, what we have is that any codeword $z\in C_1*C_2$ of the specific form $z=c_1*c_2$ has weight $w(z)\geq s$.
But there are codewords in $C_1*C_2$ that are not of this form (because $C_1*C_2$ is defined as a linear span), which can make its minimum distance smaller:
see Example~\ref{dmin_non_produit}.
\item On the other hand, $\dmin(C_1*C_2)\geq s$ does not necessarily imply that $(C_1,C_2)$ is $s$-intersecting.
In fact it could well be that there are nonzero $c_1,c_2$ such that $c_1*c_2=0$, so $(C_1,C_2)$ is even not intersecting at all!
Indeed note that $c_1*c_2=0$ does not contribute to $\dmin(C_1*C_2)$.
The possibility of such ``unexpected zero codewords'' is one difficulty in the estimation of the minimum distance of products of codes;
see also the discussion in~\ref{distance_geom}.
\end{itemize}
However it remains that, if $(C_1,C_2)$ is intersecting, then it is at least $\dmin(C_1*C_2)$-intersecting.
Or equivalently, if $i(C_1,C_2)>0$, then $i(C_1,C_2)\geq\dmin(C_1*C_2)$.

\mysubsection{Link with tensor constructions}
Here we present algebraic constructions related to products and powers of codes.
Later (in~\ref{geom_schema}-\ref{Veronese}) they will be revisited from a geometric point of view.

\entry
\label{tens_et_*}
We identify the tensor product $\F^m\tens\F^n$ with the space $\F^{m\times n}$ of $m\times n$ matrices,
by identifying the elementary tensor $(x_1,\dots,x_m)\tens(y_1,\dots,y_n)$ 
with the matrix with entries $(x_iy_j)$, and extending by linearity.

Given two linear codes $C,C'\subset\F^n$ of the same length~$n$,
their tensor product $C\tens C'\subset\F^{n\times n}$ is then the linear code
spanned by elementary tensor codewords $c\tens c'$, for $c\in C,c'\in C'$.
It then follows from the definitions that $C*C'$ is the projection of $C\tens C'$
on the diagonal. In fact this could be taken as an alternative definition for $C*C'$.

So we have an exact sequence
\beq
0\longto I(C,C')\longto C\tens C'\overset{\pi_\Delta}{\longto} C*C'\longto 0
\eeq
where $\pi_\Delta$ is the projection on the diagonal, and $I(C,C')$ is its kernel in $C\tens C'$.
One might view $I(C,C')$ as the space of formal bilinear expressions in codewords of $C,C'$
that evaluate to zero.

Equivalently, if $c_1,\dots,c_k$ are the rows of a generator matrix of $C$, and $c'_1,\dots,c'_{k'}$
are the rows of a generator matrix of $C'$, then the products $c_i*c'_j$ generate $C*C'$, but need
not be linearly independent: the space of linear relations between these product generators is precisely $I(C,C')$.
In particular we have
\beq
\dim(C*C')\;=\;kk'-\dim(I(C,C'))\;\leq\;\min(n,kk').
\eeq

It is true, although not entirely obvious from its definition, that $C\tens C'$ can also be described
as the space of $n\times n$ matrices all of whose columns are in $C$ and all of whose rows are in $C'$.
Then $I(C,C')$ is the subspace made of such matrices that are zero on the diagonal.

\entry
\label{def_ItC}
Likewise, by the universal property of symmetric powers, there is a natural surjective
map $S^tC\surj C\deux[t]$, whose kernel $I^t(C)$ can be viewed as the space of formal
homogeneous polynomials of degree $t$ in codewords of $C$ that evaluate to zero.

Equivalently, if $c_1,\dots,c_k$ are the rows of a generator matrix of $C$, then the monomials
$(c_1)^{i_1}*\cdots*(c_k)^{i_k}$, for $i_1+\cdots+i_k=t$, generate $C\deux[t]$, but need
not be linearly independent: the space of linear relations between these monomial generators is precisely $I^t(C)$.
In particular we have
\beq
\dim(C\deux[t])\;=\;\binom{k\!+\!t\!-\!\!1}{t}-\dim(I^t(C))\;\leq\;\min\left(n,\binom{k\!+\!t\!-\!\!1}{t}\right).
\eeq

For example, let $5\leq n\leq\abs{\F}$, and $C$ be the $[n,3]$ Reed-Solomon code obtained by evaluating polynomials of degree up to $2$ in $\F[x]$ at $n$ points in $\F$.
Denote by $1,x,x^2$ the canonical basis of $C$ (obtained from the corresponding monomials in $\F[x]$).
Then
\beq
1\cdot x^2-x\cdot x
\eeq
is nonzero in $S^2C$ but it evaluates to zero in $C\deux$, that is, it defines a nonzero element in $I^2(C)$.
We have $\dim(S^2C)=6$ and $\dim(C\deux)=5$, so in fact $I^2(C)$ has dimension $1$ and admits $1\cdot x^2-x\cdot x$ as a generating element.

\entry
\label{C_algebra}
The direct sum
\beq
C\deux[.]=\bigoplus_{t\geq0}C\deux[t]
\eeq
admits a natural structure of
graded $\F$-algebra, which makes it a quotient of the symmetric algebra $S^\cdot C$ under
the maps described in the previous entry.

Equivalently, the direct sum
\beq
I^\cdot(C)=\bigoplus_{t\geq0}I^t(C)
\eeq
is a homogeneous ideal
in the symmetric algebra $S^\cdot C$, and we have a natural identification
\beq
C\deux[.]=S^\cdot C/I^\cdot(C)
\eeq
of graded $\F$-algebras.

\mysubsection{Rank functions}
%

\entry
\label{def_rank}
Let $C$ be a finite dimensional $\F$-vector space, and $S\subset C$
a generating set.
\begin{definition*}
The rank function associated with $S$ is the
function $\rk_S:C\longto\Z_{\geq0}$ defined as follows: the rank $\rk_S(c)$
of an element $c\in C$ is the smallest integer $r\geq0$ such that
there is a decomposition
\beq
c=\lambda_1s_1+\cdots+\lambda_rs_r,\qquad\lambda_i\in\F,s_i\in S
\eeq
of $c$ as a linear combination of $r$ generators from $S$.
\end{definition*}

Conversely we say that $\rk:C\longto\Z_{\geq0}$ is a rank function on $C$ if $\rk=\rk_S$
for some generating set $S$.
It is easily seen that a rank function is a norm on $C$
(relative to the trivial absolute value on $\F$),
in particular it satisfies
\beq
\rk(\lambda x)=\rk(x)
\eeq
and
\beq
\rk(x+y)\leq\rk(x)+\rk(y)
\eeq
for all $\lambda\in\F^\times$, $x,y\in C$.
In fact, a generating set $S$ defines a surjective linear map $\F^{(S)}\surj C$,
and $\rk_S$ is then the quotient norm on $C$ of the Hamming norm on $\F^{(S)}$.

\vspace{.2\baselineskip}

(In some instances we will also make the following abuse:
given $S\subset C$ that does not span the whole of $C$, we define $\rk_S$ on $\langle S\rangle$ as above,
and then we let $\rk_S(c)=\infty$ for $c\in C\moins\langle S\rangle$.)

\begin{example}
Suppose given a full rank $n\times(n-k)$ matrix $H$, and let $y\in\F^n$ be an arbitrary word (in row vector convention),
and $z=yH^T\in\F^{n-k}$ the corresponding ``syndrome''.
The set $S$ of columns of $H$, or equivalently, of rows of $H^T$, is a generating set in $\F^{n-k}$, and the rank
\beq
\rk_S(z)
\eeq
is then equal to the weight of a minimum error vector $e$ such that $y-e$ is in the code defined by the parity-check matrix $H$.
\end{example}

\entry
Given a rank function on $C$, we let
\beq
C_{(i)}=\{c\in C\,;\;\rk(c)=i\},\qquad C_{(\leq i)}=\{c\in C\,;\;\rk(c)\leq i\}
\eeq
for all $i\geq0$.
Obviously $\rk(c)\leq k=\dim(C)$ for all $c\in C$, so
\beq
0=C_{(\leq 0)}\subset C_{(\leq 1)}\subset\dots\subset C_{(\leq k)}=C.
\eeq

Two generating sets $S,S'\subset C$ define the same rank function if the sets of lines
$\{\F\cdot s\,;\,s\in S,s\neq0\}$ and $\{\F\cdot s'\,;\,s'\in S',s'\neq0\}$ are equal.
Given a rank function $\rk$ on $C$, there is a preferred generating set $S$ such
that $\rk=\rk_S$, namely it is $S=C_{(1)}$.

\begin{lemma}
\label{fonctorialite_rang}
Let $f:C\longto C'$ be a linear map between two finite dimensional $\F$-vector spaces.
Suppose $S\subset C$, $S'\subset C'$ are generating sets with $f(S)\subset S'$.
Then for all $c\in C$ we have
\beq
\rk_{S}(c)\geq\rk_{S'}(f(c)).
\eeq
\end{lemma}
\begin{proof}
Obvious from the definition.
\end{proof}

\entry
\label{notation_point}
Now we generalize constructions~\ref{def_dot*}-\ref{def_*} slightly.
Suppose given sets $V_1,\dots,V_t,W$ and a map
\beq
\Phi:V_1\times\cdots\times V_t\longto W.
\eeq
Then for subsets $S_1\subset V_1,\dots,S_t\subset V_t$ we define
\beq
\dot{\Phi}(S_1,\dots,S_t)=\{ \Phi(c_1,\dots,c_t)\,;\;c_1\in S_1,\dots,c_t\in S_t\}\;\subset\,W.
\eeq
If moreover $V_1,\dots,V_t,W$ are $\F$-vector spaces and $C_1\subset V_1,\dots,$ $C_t\subset V_t$ are $\F$-linear subspaces, we let
\beq
\Phi(C_1,\dots,C_t)=\langle\dot{\Phi}(C_1,\dots,C_t)\rangle\;\subset\,W.
\eeq
In this definition $\Phi$ could be an arbitrary map, although in most examples it will be $t$-multilinear.
In this case, the ($\F$-)linear span is easily seen to reduce to just an additive span.

Also, we will use analogous notations when $\Phi$ is written as a composition law,
for instance if $V,V'$ are $\F$-vector spaces and $S\subset V$, $S'\subset V'$ arbitrary subsets,
then $S\dot{\tens}S'=\{c\tens c'\,;\;(c,c')\in S\times S'\}\;\subset\,V\tens V'.$

\entry
Let $V,V',W$ be $\F$-vector spaces and $\Phi:V\times V'\longto W$ a bilinear map.
Let $C\subset V$ and $C'\subset V'$ be linear subspaces, and suppose $C,C'$ equipped
with rank functions $\rk,\rk'$ respectively.
Then $\dot{\Phi}(C_{(1)},{C'}_{(1)})\subset\Phi(C,C')$ is a
generating set, and we define $\rk^\Phi$ as the associated rank function on $\Phi(C,C')$, also called the $\Phi$-rank function deduced from $\rk$ and $\rk'$.

Alternatively,
the rank $\rk^\Phi(z)$ of an element $z\in\Phi(C,C')$ can be computed as the
smallest value of the sum $\sum_i\rk(c_i)\rk'(c'_i)$ over all possible decompositions
(of arbitrary length)
$z=\sum_i\Phi(c_i,c'_i)$, $c_i\in C,c'_i\in C'$.

Here we considered a bilinear $\Phi$, but these constructions easily generalize to the $t$-multilinear case, $t\geq3$.

When $\Phi=\tens$ is tensor product, or when $\Phi=*$ is componentwise multiplication in $\F^n$,
we will often make the following additional assumptions:

\entry
From now on, unless otherwise specified, when a linear code $\widetilde{C}$ is written as a tensor product
$\widetilde{C}=C_1\tens\dots\tens C_t$, it will be assumed that the $C_i$ are equipped with the
trivial rank function (such that $\rk(c_i)=1$ for all nonzero $c_i\in C_i$),
and that $\widetilde{C}$ is equipped with the $\tens$-rank function from these.
A nonzero $c\in\widetilde{C}$ is then of rank $1$ if it is an \emph{elementary tensor}
$c=c_1\tens\dots\tens c_t$, with $c_i\in C_i$.

For example, with these conventions,
the rank function on $\F^{m\times n}=\F^m\tens\F^n$ is the rank of matrices in the usual sense.
Given two linear codes $C\subset\F^m$, $C'\subset\F^n$,
we have an inclusion $C\tens C'\subset\F^m\tens\F^n$, and by Lemma~\ref{fonctorialite_rang}
the rank of a codeword $z\in C\tens C'$ is greater than or equal to its rank as a matrix.

\entry
\label{rang_puissance_cano}
Likewise, when a linear code $\widetilde{C}\subset\F^n$ is written as a product
$\widetilde{C}=C_1*\cdots*C_t$ for some $C_i\subset\F^n$, it will be assumed that these $C_i$ are equipped
with the trivial rank function, and that $\widetilde{C}$ is equipped with the $*$-rank function from these.
A nonzero $c\in \widetilde{C}$ is then of rank $1$ if it is an \emph{elementary product}
$c=c_1*\dots* c_t$, with $c_i\in C_i$.

In particular, given a code $C$ and its $t$-th power $C\deux[t]$, it will be assumed that $C$ is equipped with the trivial rank function, and $C\deux[t]$ with its $t$-th power.

As in~\ref{tens_et_*} there is a natural projection
$\pi_\Delta:C_1\tens\cdots\tens C_t\longto C_1*\cdots*C_t$
(or $C^{\tens t}\longto C\deux[t]$), and by Lemma~\ref{fonctorialite_rang}
this map can only make the rank decrease.

\begin{definition}
\label{dist_rang}
Let $C$ be a nonzero linear code equipped with a rank function $\rk$. 
For any $i\geq1$
we then define ${\dmin}_{,i}(C)$ as the minimum weight of a nonzero element in $C_{(\leq i)}$.
\end{definition}
If $C$ has dimension $k$, then obviously
\beq
{\dmin}_{,1}(C)\geq {\dmin}_{,2}(C) \geq\dots\geq {\dmin}_{,k}(C)=\dmin(C).
\eeq

\begin{example}
\label{dmin_non_produit}
Given two linear codes $C\subset\F^m$ and $C'\subset\F^n$, it follows from the
description of $C\tens C'$ as the space of $m\times n$ matrices with columns in $C$ and
rows in $C'$, that $\dmin(C\tens C')=\dmin(C)\dmin(C')$, and that moreover this value
is attained by an elementary tensor codeword, that is
\beq
{\dmin}_{,1}(C\tens C')=\dmin(C\tens C')=\dmin(C)\dmin(C').
\eeq

On the other hand, let $C,C'\subset(\F_2)^7$ be the linear codes with generator matrices
\beq
G=\left(
\begin{array}{ccccccc}
1&0&0&1&1&1&1\\
0&1&1&1&1&0&0
\end{array}\right),\qquad
G'=\left(
\begin{array}{ccccccc}
1&0&0&1&1&1&1\\
0&1&1&0&0&1&1
\end{array}\right)
\eeq
respectively, so the nonzero codewords of $C$ are
$c_1=(1001111)$, $c_2=(0111100)$, $c_1+c_2=(1110011)$
and the nonzero codewords of $C'$ are $c'_1=c_1=(1001111)$,
$c'_2=(0110011)$, $c'_1+c'_2=(1111100)$.
We can then check that $c*c'$ has weight at least $2$ for all nonzero $c\in C$, $c'\in C'$,
while $C*C'$ also contains $c_1*c'_1+c_1*c'_2+c_2*c'_1=(1000000)$.
Hence 
\beq
{\dmin}_{,1}(C*C')=2>\dmin(C*C')=1,
\eeq
and in this example, $\dmin(C*C')$ cannot be attained by an elementary product codeword.
\end{example}

\mysubsection{Geometric aspects}
Part of the discussion here is aimed at readers with a certain working knowledge of algebraic geometry. Other readers can still read \ref{longueurs}-\ref{geom_naive}, the first halves of \ref{Segre} and \ref{Veronese}, and also \ref{distance_geom}, which remain elementary, and then skip to the next section with no harm.

\entry
\label{longueurs}
Let $C\subset\F^n$ be a linear code and $C^{\perp}\subset\F^n$ its dual.
For each integer $i\geq0$ we have $C=C^{\perp\perp}\subset\langle x\in C^{\perp}\,;\;w(x)\leq i\rangle^{\perp}$,
and we denote the dimension of the latter as
\beq
n_i=\dim \langle x\in C^{\perp}\,;\;w(x)\leq i\rangle^{\perp}.
\eeq
Obviously $n_i\geq n_{i+1}$, and the first values are easily computed
(see also \ref{support_et_dual} and~\ref{repete_et_dual}):
\begin{itemize}
\item $n_0=n$ is the length of $C$
\item $n_1=\abs{\Supp(C)}$ is the support length of $C$, that is, the number of nonzero columns
of a generator matrix of $C$
\item $n_2$ is the \emph{projective length} of $C$, that is, the number of proportionality classes
of nonzero columns of a generator matrix of $C$.
\end{itemize}
The author does not know such a nice interpretation for the subsequent values $n_3,n_4,\dots$

At some point the sequence must stabilize, more precisely, if $C^\perp$ is generated by its codewords
of weight at most $i_0$, then $n_{i_0}=n_{i_0+1}=\dots=k=\dim(C)$.
Putting $C^\perp$ in systematic form, one sees one can take $i_0\leq k+1$.

\entry
\label{geom_naive}
Let $C\subset\F^n$ be a linear code of dimension $k$, and let $G$ be a generator matrix for $C$.
We will suppose that $C$ has full support, that is, $G$ has no zero column, or with the notations of~\ref{longueurs}, $n_1=n$.
In many applications, the properties of $C$ that are of interest are preserved when a column of $G$ is replaced with a proportional one, or
when columns are permuted. So (\cite{TVbook}\cite{TVieee}) these properties only depend on the projective set $\Pi_C\subset\PP^{k-1}$ of proportionality classes
of columns of $G$, where possibly elements of $\Pi_C$ may be affected multiplicities to reflect the fact that some columns of $G$ may be
repeated (up to proportionality).

However, in our context we need to keep track of the ordering of columns, since the product of codes works coordinatewise.
This can be done by considering the labeling
\beq
\nu_C:[n]\surj\Pi_C\subset\PP^{k-1}
\eeq
where $\nu_C$ maps $i\in[n]$ to the proportionality class in $\PP^{k-1}$ of the $i$-th column of $G$.
Note that, in particular, the image of $\nu_C$ is $\Pi_C$. It has $n_2$ elements (with the notations of~\ref{longueurs}),
and it spans $\PP^{k-1}$ (because $G$ has rank $k$).

In fact this description can be made slightly more intrinsic.
We can view $C\subset\F^n$ as an abstract vector
space $C$ equipped with $n$ linear forms $C\longto\F$, which span the dual vector space $C^\vee$.
That $C$ has full support means that each of these $n$ linear forms is nonzero,
so it defines a line in $C^\vee$. Seeing $\PP^{k-1}$ as the projective space of these lines, we retrieve
the definition of $\nu_C$.

\entry
\label{geom_schema}
Recall \cite[\S 4.1]{EGA2}\cite[\S II.7]{Hartshorne} that if $V$ is a finite dimensional $\F$-vector space, then
\beq
\PP(V)=\Proj S^\cdot V
\eeq
is the scheme whose points represent
lines in $V^\vee$, or equivalently, hyperplanes in $V$, or equivalently, invertible quotients of $V$. 
If $\cA$ is a $\F$-algebra, then giving a map $\nu:\Spec \cA\longto\PP(V)$ is the same as giving an invertible $\cA$-module $\cL$
and a $\F$-linear map $V\longto\cL$ whose image generates $\cL$ over $\cA$. The closure of the image of $\nu$ (which will be the full image of $\nu$ if $\cA$ is finite) is then the closed subscheme
of $\PP(V)$ defined by the homogeneous ideal $\bigoplus_{t\geq0}\ker(S^tV\longto\cL^{\tens t})$ of $S^\cdot V$.

We apply this with $V=C$ and $\cL=\cA=\F^n$. Indeed, that $C$ has full support means $\F^n*C=\F^n$, that is, $C$ generates $\F^n$
as a $\F^n$-module. So we deduce a morphism $[n]=\Spec\F^n\longto\PP(C)$. This morphism is precisely the $\nu_C$ defined in~\ref{geom_naive}:
indeed the $n$ points of $[n]$ correspond to the $n$ projections $\F^n\longto\F$, so their images in $\PP(C)$ correspond to the $n$ coordinate linear forms $C\longto\F$.
Recalling notations from~\ref{def_ItC}-\ref{C_algebra}, we then find:

\begin{proposition}
The map $\nu_C$ fits into the commutative diagram
\beq
\begin{array}{cccccc}
\nu_C:\!\!\!\!\! & [n] & \surj & \Pi_C & \!\!\subset & \PP^{k-1} \\
& \Arrowvert & & \Arrowvert & & \Arrowvert \\
& \Spec\F^n & \surj & \Proj C\deux[.] & \!\!\subset & \PP(C)
\end{array}
\eeq
where the homogeneous ideal defining $\Pi_C=\Proj C\deux[.]$ in $\PP^{k-1}=\PP(C)$
is $I^\cdot(C)$.
\end{proposition}
\begin{proof}
Indeed we have $\ker(S^tC\longto\F^n)=\ker(S^tC\longto C\deux[t])=I^t(C)$.
\end{proof}
Note that the linear span of $\Proj C\deux[.]$ is the whole of $\PP(C)$ since $I^1(C)=0$.

\entry
A possible application of what precedes is to the interpolation problem, where one seeks a subvariety $\Sigma\subset\PP^{k-1}$
passing through $\Pi_C$, in order to write $C$ as an evaluation code on $\Sigma$. Viewing $S^tC$ as the space of homogeneous
functions of degree $t$ on $\PP^{k-1}$, the homogeneous equations defining $\Sigma$ are then to be found in $I^\cdot(C)$.

Another consequence is the following, which explains the names in Definitions~\ref{def_sequences}-\ref{def_CMreg}.
We define the Hilbert function and the Castelnuovo-Mumford regularity of a closed subscheme in a projective space, as those of its homogeneous coordinate ring (see \eg \cite{Eisenbud}).
Then:

\begin{proposition}
\label{eq_def}
Let $C\subset\F^n$ be a linear code with full support, and $\nu_C:[n]\surj\Pi_C\subset\PP^{k-1}$ the associated projective spanning map.
Then:
\begin{enumerate}[(i)]
\item\label{eq_def_i} The dimension sequence of $C$ is equal to the Hilbert function of $\Pi_C$.
\item\label{eq_def_ii} The regularity $r(C)$ of $C$ is equal to the Castelnuovo-Mumford regularity of $\Pi_C$.
\item\label{eq_def_iii} The stable value of the dimension sequence is the projective length of $C$:
\beq
\dim C\deux[t]=n_2
\eeq
for $t\geq r(C)$.
\end{enumerate}
\end{proposition}
\begin{proof}
As established during the construction of $\nu_C$ in~\ref{geom_schema}, the homogeneous coordinate ring of $\Pi_C$ is $S^\cdot C/I^\cdot(C)=C\deux[.]$.
This gives point~\ref{eq_def_i}, 
and then point~\ref{eq_def_ii} follows by \cite[Th.~4.2(3)]{Eisenbud}
(see also \cite[Lect.~14]{Mumford}).
To show point~\ref{eq_def_iii}, first recall that if $\cA^\cdot$ is a graded algebra such that the dimension $\dim\cA^t$ becomes constant
for $t\gg0$, then this stable value is precisely the length of the finite projective scheme $\Proj \cA^\cdot$.
Now here $\Pi_C$ is a reduced union of some $\F$-points (as an image of $[n]$), so this length is precisely the number $n_2$ of these points.
\end{proof}
Another (perhaps more concrete) proof of point~\ref{eq_def_iii} will be given in Theorem~\ref{descr_stable}.
For more on the geometric significance of~\ref{eq_def_ii}, see the discussion in~\ref{Gale}.

\begin{remark}
\label{H1}
From the short exact sequence of sheaves on $\PP(C)$ 
\beq
0\longto\fJ_{\Pi_C}\cO_{\PP(C)}(t)\longto\cO_{\PP(C)}(t)\longto\cO_{\Pi_C}(t)\longto0
\eeq
one can form a long exact sequence in cohomology, in which the first terms can be
identified as $\Gamma(\PP(C),\fJ_{\Pi_C}\cO_{\PP(C)}(t))=I^t(C)$
and $\Gamma(\PP(C),\cO_{\PP(C)}(t))=S^tC$, leading to a short exact sequence
\beq
0\longto C\deux[t]\longto\Gamma(\Pi_C,\cO_{\Pi_C}(t))\longto H^1(\PP(C),\fJ_{\Pi_C}\cO_{\PP(C)}(t))\longto0.
\eeq
Now $\Gamma(\Pi_C,\cO_{\Pi_C}(t))$ is a vector space of dimension $n_2$ over $\F$, and
choosing a subset $S\subset[n]$ of size $\abs{S}=n_2$ mapped bijectively onto $\Pi_C$ by $\nu_C$,
we can identify this vector space with $\F^S$. From this we finally deduce an identification
\beq
H^1(\PP(C),\fJ_{\Pi_C}\cO_{\PP(C)}(t))\simeq \F^S/\pi_S(C\deux[t]).
\eeq
In particular, since $C$ and $C\deux[t]$ have the same projective length for $t\geq1$
(this should be obvious, but if not see~\ref{repete_et_dual}-\ref{tranches_puissances}), we find that
\beq
\begin{split}
\dim H^1(\PP(C),\fJ_{\Pi_C}\cO_{\PP(C)}(t))&=n_2-\dim(C\deux[t])\\
&=\dim((C\deux[t])^\perp/\langle x\in (C\deux[t])^{\perp}\,;\;w(x)\leq 2\rangle)
\end{split}
\eeq
is the minimum number of parity-check relations of weight at least $3$ necessarily appearing in any set
of relations defining $C\deux[t]$.
\end{remark}

\entry
\label{Segre}
Let $C,C'\subset\F^n$ be two linear codes. Choose corresponding generating matrices $G,G'$.
For $1\leq i\leq n$, let $p_i\in\F^k$ be the $i$-th column of $G$
and $p'_i\in\F^{k'}$ the $i$-th column of $G'$. Then $C,C'$ are the respective images of the evaluation maps
\beq
\begin{array}{ccc}
\F[X_1,\dots,X_k]_1 & \longto & \F^n\\
L & \mapsto & (L(p_1),\dots,L(p_n))
\end{array}
\eeq
and
\beq
\begin{array}{ccc}
\F[Y_1,\dots,Y_{k'}]_1 & \longto & \F^n\\
L' & \mapsto & (L'(p'_1),\dots,L'(p'_n))
\end{array}
\eeq
defined on spaces of linear homogeneous polynomials in $k$ and $k'$ variables.
Then, $C*C'$ is the image of the evaluation map
\beq
\begin{array}{ccc}
\F[X_1,\dots,X_k;Y_1,\dots,Y_{k'}]_{1,1} & \longto & \F^n\\
B & \mapsto & (B(p_1;p'_1),\dots,B(p_n;p'_n))
\end{array}
\eeq
defined on the space of bilinear homogenous polynomials in $k+k'$ variables, that is, on polynomials of the form
\beq
B(X_1,\dots,X_k;Y_1,\dots,Y_{k'})=\sum_{i,j}\mu_{i,j}X_iY_j
\eeq
where $\mu_{i,j}\in\F$, $1\leq i\leq k$, $1\leq j\leq k'$.

This is just a reformulation of~\ref{tens_et_*}. Geometrically, it corresponds to the \emph{Segre} construction.

Suppose $C,C'\subset\F^n$ have full support,
and let $\nu_C:[n]\surj\Pi_C\subset\PP^{k-1}$ and $\nu_{C'}:[n]\surj\Pi_{C'}\subset\PP^{k'-1}$ be the associated projective spanning maps.
Composing this pair of maps with the Segre embedding we get
\beq
[n]\xrightarrow{\;(\nu_C,\nu_{C'})\;}\PP^{k-1}\times\PP^{k'-1}\longto\PP^{kk'-1}
\eeq
which should be essentially $\nu_{C*C'}$, except that its image $\Pi_{C*C'}$ might not span $\PP^{kk'-1}$ as requested, so we have
to replace $\PP^{kk'-1}$ with the linear span of the image $\langle\Pi_{C*C'}\rangle$, which is a $\PP^{\dim(C*C')-1}$.

More intrinsically, we have $\PP^{k-1}=\PP(C)$, $\PP^{k'-1}=\PP(C')$, and $\PP^{kk'-1}=\PP(C\tens C')$.
The linear span $\langle\Pi_{C*C'}\rangle$ of $\Pi_{C*C'}$ is then easily identified:
we have $C*C'=C\tens C'/I(C,C')$, so 
\beq
\langle\Pi_{C*C'}\rangle=\PP(C*C')\subset\PP(C\tens C')
\eeq
is the linear subspace cut by $I(C,C')$
(where we view elements of $I(C,C')\subset C\tens C'$ as linear homogeneous functions on $\PP^{kk'-1}=\PP(C\tens C')$).

We summarize this with the commutative diagram
\beq
\begin{array}{cccccccc}
\nu_{C*C'}:\!\!\!\!\! & [n] & \surj & \Pi_{C*C'} & \!\!\subset & \langle\Pi_{C*C'}\rangle & \!\subset & \PP^{kk'-1} \\
& \Arrowvert & & \Arrowvert & & \Arrowvert & & \Arrowvert \\
& \Spec\F^n & \surj & \Proj (C*C')\deux[.] & \!\!\subset & \PP(C*C') & \!\subset & \PP(C\tens C').
\end{array}
\eeq

\entry
\label{Veronese}
Keep the same notations as in the previous entry. First, in coordinates, if $C$ is the image of the evaluation map
\beq
\begin{array}{ccc}
\F[X_1,\dots,X_k]_1 & \longto & \F^n\\
L & \mapsto & (L(p_1),\dots,L(p_n))
\end{array}
\eeq
defined on the space of linear homogeneous polynomials in $k$ variables, then $C\deux[t]$ is the image of the evaluation map
\beq
\begin{array}{ccc}
\F[X_1,\dots,X_k]_t & \longto & \F^n\\
Q & \mapsto & (Q(p_1),\dots,Q(p_n))
\end{array}
\eeq
defined on the space of homogeneous polynomials of degree $t$ in $k$ variables.

This is just a reformulation of~\ref{def_ItC}. Geometrically, it corresponds to the \emph{Veronese} construction.

Suppose $C\subset\F^n$ has full support, and let $\nu_C:[n]\surj\Pi_C\subset\PP^{k-1}$ be the associated projective spanning map.
Composing with the $t$-fold Veronese embedding we get
\beq
[n]\xrightarrow{\;\nu_C\;}\PP^{k-1}\longto\PP^{\binom{k+t-1}{t}-1},
\eeq
the image of which spans the linear subspace
\beq
\langle\Pi_{C\deux[t]}\rangle=\PP(C\deux[t])\subset\PP(S^tC)
\eeq
cut by $I^t(C)$ (where now we see elements of $I^t(C)\subset S^tC$ not as homogeneous functions of degree $t$ on $\PP^{k-1}=\PP(C)$,
but as linear homogeneous functions on $\PP^{\binom{k+t-1}{t}-1}=P(S^tC)$).

Again we summarize this with the commutative diagram
\beq
\begin{array}{cccccccc}
\nu_{C\deux[t]}:\!\!\!\!\!\! & [n] & \surj & \Pi_{C\deux[t]} & \!\!\subset & \langle\Pi_{C\deux[t]}\rangle & \!\subset & \PP^{\binom{k+t-1}{t}-1} \\
& \Arrowvert & & \Arrowvert & & \Arrowvert & & \Arrowvert \\
& \Spec\F^n & \surj & \Proj C\deux[t\cdot] & \!\!\subset & \PP(C\deux[t]) & \!\subset & \PP(S^tC).
\end{array}
\eeq

\entry
\label{distance_geom}
This geometric view is especially interesting when one considers the distance problem.
Given $C\subset\F^n$ with full support, and $\nu_C:[n]\surj\Pi_C\subset\PP^{k-1}$ the associated projective spanning map,
nonzero codewords $c\in C$ correspond to hyperplanes $H_c\subset\PP^{k-1}$,
and the weight of $c$ is $w(c)=n-\abs{\nu_C^{-1}(H_c)}$. As a consequence, the minimum distance of $C$ is 
\beq
\dmin(C)=n-\max_{\substack{H\subset\PP^{k-1}\\ \textrm{hyperplane}}}\abs{\nu_C^{-1}(H)}.
\eeq
Applying the Veronese construction which identifies hyperplanes in $\PP^{\binom{k+t-1}{t}-1}$ with hypersurfaces of degree $t$ in $\PP^{k-1}$, we find likewise
\beq
\dmin(C\deux[t])=n-\max_{\substack{H\subset\PP^{k-1},\,H\not\supseteq\Pi_C\\ \textrm{hypersurface of degree $t$}}}\abs{\nu_C^{-1}(H)}.
\eeq
Note that here we have to add the extra condition $H\not\supseteq\Pi_C$, reflecting the fact that $I^t(C)$ could be nonzero.
This makes the distance problem slightly more delicate as soon as $t\geq2$.

In many code constructions, often the very same argument that gives a lower bound on the minimum distance shows at the same time that the code has ``full dimension''. 
For example, if $C(D,G)=\im(L(D)\longto\F^G)$ is the algebraic-geometry code defined in \ref{exAG}, then, provided $m=\deg(D)<n=\abs{G}$,
a function in $L(D)$ can have at most $m$ zeroes in $G$,
from which we get at the same time injectivity of the evaluation map, so $\dim(C(D,G))=\dim(L(D))$, and $\dmin(C(D,G))\geq n-m$. 

On the other hand if a code is defined as a power of another code, we have to deal separately with the fact that it could have dimension smaller than expected.
Given $\nu_C:[n]\surj\Pi_C\subset\PP^{k-1}$, to show $\dmin(C\deux[t])\geq n-m$ one has to show that for any homogenous form of degree $t$ on $\PP^{k-1}$, either:
\begin{itemize}
\item $\nu_C^*F$ has at most $m$ zeroes in $[n]$, or
\item $\nu_C^*F$ vanishes on all of $[n]$.
\end{itemize}


\entry
\label{personal_speculations}
Now the author would like to share some (personal) speculations about the objects constructed so far.

From Proposition~\ref{eq_def}, we see that the dimension sequence of a linear code $C$ is a notion that has been already well studied,
albeit under a different (but equivalent) form.
In fact, its study can also be reduced to an interpolation problem: since $\dim C\deux[t]=\binom{k+t-1}{t}-\dim I^t(C)$, to estimate the Hilbert function
we can equivalently count the hypersurfaces of degree $t$ passing through $\Pi_C$. 
This problem is not really of a coding-theoretic nature. We can do the same thing for powers of a linear subspace
in any finite-dimensional algebra $\cA$, not only in $\F^n$.

However, things change if one is also interested in the distance sequence of $C$. While we're still doing geometry over $\F$,
that is, over a base of dimension~$0$, now, following the philosophy of Arakelov theory, the introduction of metric data (such as
defined, here, by the Hamming metric) is very similar to passing to a base of dimension~$1$.
In this way, the study of the joint dimension and distance sequences of a code might be viewed as a finite field analogue
of the study of the ``arithmetic Hilbert(-Samuel) function'' associated in~\cite{Laurent} to interpolation matrices over a number field,
and further analyzed in~\cite{HSSdim0}.
For example the monotonicity results that will be given in \ref{monotone}-\ref{strict_monotone} are very similar in spirit to those of \cite[5.2]{HSSdim0};
in turn, keeping Remark~\ref{H1} in mind, a natural interpretation is as the size of some $H^1$ decreasing, as in \cite[p.~102]{Mumford}.

For another illustration of this principle, to give an upper bound on $\dmin(C\deux[t])$ one has to find a nonzero codeword of small weight in $C\deux[t]$, that is,
a function $P\in S^tC$ whose zero locus intercepts a large part of, but not all, the image of $[n]$ under $\nu_{C\deux[t]}$.
This is somehow reminiscent of the situation in transcendental number theory, where one has to construct an auxiliary function that is small
but nonzero, which often involves a Siegel lemma. Conversely, to give a lower bound on $\dmin(C\deux[t])$, one has to show that
for all $P\in S^tC$, either $P$ vanishes on all the image (which means $P\in I^t(C)$), or else it misses a certain part
of it, of controlled size. Perhaps one could see this as a loose analogue of a zero lemma.

\section{Basic structural results and miscellaneous properties}
\label{sect_misc}

\vspace{.5\baselineskip}
In this section we study basic properties of codes with respect to componentwise product,
while aiming at the widest generality.
This means including the case of ``degenerated'' codes (\eg not having full support, or having repeated columns,
or also decomposable codes) that are often not of primary interest to coding theorists;
the hurried reader should feel free to skip the corresponding entries. 

This said, it turns out these degenerated codes sometimes appear in some natural situations,
which motivates having them treated here for reference. For instance, even if a code $C$
is indecomposable, its powers $C\deux[t]$ might be decomposable.
Also, to study a code $C$, it can be useful to filter it by a chain of subcodes $C_i$
(see \eg \ref{filtration}, or \cite{CGGOT}\cite{Mirandola}\cite{Wei}), and even for the nicest $C$,
the $C_i$ under consideration might very well then be degenerated.

\mysubsection{Support}
\entry
\label{debut_support}
From now on, by the $i$-th column of a linear code $C\subset\F^n$,
we will mean the $i$-th coordinate projection $\pi_i:C\longto\F$, which is an element of the dual vector space $C^\vee$.

This name is justified because, given a generator matrix $G$, which corresponds to a basis of $C$ over $\F$,
the column vector of the coordinates of $\pi_i$ with respect to this basis is precisely the $i$-th column of $G$.

\entry
\label{1_support}
A possible definition of the support of words or codes, in terms of our product $*$, can be given as follows.
First, note that for all $S,T\subset[n]$, we have $1_S*1_T=1_{S\cap T}$.
In particular, $1_S$ is an idempotent of $\F^n$. In fact, as a linear endomorphism of $\F^n$, we have
\beq
1_S*\cdot = \iota_S\circ\pi_S
\eeq
where $\pi_S:\F^n\surj\F^S$ and $\iota_S:\F^S\inj\F^n$ are the natural linear maps.

Then the support of a word $x\in\F^n$ can be defined as the smallest, or the intersection, of all subsets $S\subset[n]$ such that
\beq
1_S*x=x.
\eeq

Likewise the support of a linear code $C\subset\F^n$ is the smallest, or the intersection, of all subsets $S\subset[n]$ such that
\beq
1_S*C=C.
\eeq

\entry
\label{support_et_dual}
Equivalently, for $i\in[n]$, we have $i\in\Supp(C)$ if and only if the $i$-th column of $C$ is nonzero.
This may be rephrased in terms of vectors of weight $1$ in the dual code:
\beq
i\not\in\Supp(C)\quad\equivaut\quad 1_{\{i\}}\in C^\perp.
\eeq
As a consequence we have
\beq
\langle x\in C^{\perp}\,;\;w(x)\leq 1\rangle^{\perp}=\iota(\F^{\Supp(C)})
\eeq
(where $\iota=\iota_{\Supp(C)}:\F^{\Supp(C)}\inj\F^n$) and we retrieve the relation
\beq
n_1=\dim\langle x\in C^{\perp}\,;\;w(x)\leq 1\rangle^{\perp}=\abs{\Supp(C)}
\eeq
as stated in~\ref{longueurs}.

%

\begin{lemma}
\label{Supp*=intSupp}
If $c_1,\dots,c_t\in\F^n$ are words of the same length, then
\beq
\Supp(c_1*\cdots*c_t)=\Supp(c_1)\cap\cdots\cap\Supp(c_t).
\eeq
If $C_1,\dots,C_t\subset\F^n$ are linear codes of the same length, then
\beq
\Supp(C_1*\cdots*C_t)=\Supp(C_1)\cap\cdots\cap\Supp(C_t).
\eeq
In particular, for $C\subset\F^n$ and $t\geq1$ we have
\beq
\Supp(C\deux[t])=\Supp(C).
\eeq
\end{lemma}
\begin{proof}
Obvious.
\end{proof}

\entry
In most applications we can discard the $0$ columns of a linear code without affecting its good properties, that is, we can replace $C$
with its projection on $\Supp(C)$ so that it then has full support. 

In particular, given $C_1,\dots,C_t\subset\F^n$, if we let $I=\Supp(C_1)\cap\cdots\cap\Supp(C_t)$ and we replace each $C_i$ with $\pi_I(C_i)$,
this replaces $C_1*\cdots*C_t$ with $\pi_I(C_1*\cdots*C_t)$, which does not change its essential parameters (dimension, weight distribution...).
In this way, many results on products of codes can be reduced to statements on products of codes which all have full support.
However, this intersection $I$ may be strictly smaller than some of the $\Supp(C_i)$, so replacing $C_i$ with $\pi_I(C_i)$ might change some
relevant parameter of this code. In some applications, namely when both the parameters of $C_1*\cdots*C_t$ and those of the $C_i$ are relevant,
this added difficulty has to be taken into account carefully.

\mysubsection{Decomposable codes}
We recast some classical results of~\cite{Slepian} in the light of the $*$ operation,
elaborating from~\ref{1_support}.
Beside reformulating elementary notions in a fancy language, what is done here
will also appear naturally while studying automorphisms in~\ref{monomial_tsfo}
and following.

\begin{definition}
\label{def_stab_alg}
Let $C\subset\F^n$ be a linear code. The extended stabilizing algebra of $C$ is
\beq
\widehat{\cA}(C)=\{a\in\F^n\,;\;a*C\subset C\},
\eeq
and the (proper) stabilizing algebra of $C$ is
\beq
\cA(C)=1_{\Supp(C)}*\widehat{\cA}(C)=\{a\in\F^n\,;\;\Supp(a)\subset\Supp(C),\,a*C\subset C\}.
\eeq
\end{definition}

Clearly $\widehat{\cA}(C)$ is a subalgebra of $\F^n$, while projection $\pi_{\Supp(C)}$ identifies $\cA(C)$ with a subalgebra of $\F^{\Supp(C)}$
(the identity element of $\cA(C)$ is the idempotent $1_{\Supp(C)}$ of $\F^n$).
Moreover we have
\beq
\widehat{\cA}(C)=\cA(C)\oplus\iota(\F^{[n]\moins\Supp(C)})
\eeq
where $\iota=\iota_{[n]\moins\Supp(C)}$ is the natural inclusion $\F^{[n]\moins\Supp(C)}\inj\F^n$.

\begin{proposition}
\label{proprietes_algebre}
Let $C,C'\subset\F^n$ be two linear codes of the same length. Then
\beq
\cA(C)*\cA(C')\subset\cA(C*C'),
\eeq
and for all $t\geq1$
\beq
\cA(C)=\cA(C)\deux[t]\subset\cA(C\deux[t]).
\eeq
Also we have
\beq
\cA(\cA(C))=\cA(C).
\eeq
\end{proposition}
\begin{proof}
If $a*C\subset C$ and $a'*C'\subset C'$, then $(a*a')*C*C'\subset C*C'$. Using Lemma~\ref{Supp*=intSupp} and passing to the linear span we find $\cA(C)*\cA(C')\subset\cA(C*C')$ as claimed.
Induction then gives $\cA(C)\deux[t]\subset\cA(C\deux[t])$.
Last, we have $\cA(C)=\cA(C)\deux[t]$ and $\cA(\cA(C))=\cA(C)$ because $\cA(C)$ is an algebra under $*$, with unit $1_{\Supp(C)}$.  
\end{proof}

\begin{definition}
\label{def_decomposable}
Let $C\subset\F^n$ be a linear code and $\cP=\{P_1,\dots,P_s\}$ a partition of $\Supp(C)$.
We say that $C$ decomposes under $\cP$ if $1_{P_i}\in\cA(C)$ for all $i$.

Equivalently, this means there are linear subcodes $C_1,\dots,C_s\subset C$ with $\Supp(C_i)=P_i$ such that
\beq
C=C_1\oplus\cdots\oplus C_s.
\eeq
\end{definition}

To show the equivalence, write $C_i=1_{P_i}*C$, so by definition $C_i$ is a subcode of $C$ if and only if $1_{P_i}\in\cA(C)$.

\entry
We recall that the set of partitions of a given set $S$ forms a lattice under refinement.
In particular if $\cP=\{P_1,\dots,P_s\}$ and $\cQ=\{Q_1,\dots,Q_t\}$ are two partitions
of $S$, their coarsest common refinement is the partition
\beq
\cP\wedge\cQ=\{P_i\cap Q_j\,;\;P_i\cap Q_j\neq\emptyset\}.
\eeq
More generally, if $S,T$ are two sets, $\cP$ is a partition of $S$, and $\cQ$ a partition of $T$,
then $\cP\wedge\cQ$, formally defined by the very same formula as above, is a partition of $S\cap T$.

\begin{lem-def}
\label{finest_dec}
If $C$ decomposes under two partitions $\cP,\cQ$ of $\Supp(C)$, then it decomposes under $\cP\wedge\cQ$.
Hence there is a \emph{finest} partition $\cP(C)$ under which $C$ decomposes.

If $\cP(C)=\{A_1,\dots,A_r\}$, we have
\beq
C=C_1\oplus\cdots\oplus C_r
\eeq
where the $C_i=1_{A_i}*C$ are called the \emph{indecomposable components} of $C$.
This is the finest decomposition of $C$ as a direct sum of nonzero subcodes with pairwise disjoint supports.
\end{lem-def}
\begin{proof}
If $1_{P_i}\in\cA(C)$ and $1_{Q_j}\in\cA(C)$, then $1_{P_i\cap Q_j}=1_{P_i}*1_{Q_j}\in\cA(C)$.
\end{proof}

\begin{proposition}
\label{structure_algebre}
We have
\beq
\dim\cA(C)=\abs{\cP(C)}.
\eeq
More precisely, if $\cP(C)=\{A_1,\dots,A_r\}$, then
\beq
\cA(C)=\langle 1_{A_1},\dots,1_{A_r}\rangle=\langle 1_{A_1}\rangle\oplus\cdots\oplus\langle 1_{A_r}\rangle.
\eeq
\end{proposition}
\begin{proof}
Let $V=\langle 1_{A_1},\dots,1_{A_r}\rangle$.
Obviously the $1_{A_i}$ are linearly independent so $\dim(V)=r$;
and by definition we have $1_{A_i}\in\cA(C)$, so $V\subset\cA(C)$.

Conversely, let $x\in\cA(C)$. We want to show $x\in V$.
Let $\lambda_1,\dots,\lambda_s\in F$ be the elements that appear at least once as a coordinate of $x$ over $\Supp(C)$, and for
each such $\lambda_j$, let $B_j\subset\Supp(C)$ be the set of indices on which $x$ takes coordinate $\lambda_j$,
so $x=\lambda_11_{B_1}+\cdots+\lambda_s1_{B_s}$.
For each $j$, there is a Lagrange interpolation polynomial $P$ such that $P(\lambda_j)=1$
and $P(\lambda_{j'})=0$ for $j'\neq j$.
Evaluating $P$ on $x$ in the algebra $\cA(C)$ we find $1_{B_j}=P(x)\in\cA(C)$. This means $C$ decomposes under
the partition $\cQ=\{B_1,\dots,B_s\}$, hence $\cP(C)$ refines $\cQ$.
So, for all $j$, we get that $B_j$ is a union of some of the $A_i$, and $1_{B_j}\in V$.
The conclusion follows.
\end{proof}

\begin{corollary}
\label{decomp_produit}
Let $C,C'\subset\F^n$ be two linear codes of the same length.
Then $\cP(C*C')$ is a (possibly strict) refinement of $\cP(C)\wedge\cP(C')$.
For $t\geq 1$, $\cP(C\deux[t])$ is a (possibly strict) refinement of $\cP(C)$.

More generally, if $C$ decomposes under a partition $\cP$ of $\Supp(C)$
as
\beq
C=C_1\oplus\cdots\oplus C_s
\eeq
and $C'$ under a partition $\cP'$ of $\Supp(C')$
as
\beq
C=C'_1\oplus\cdots\oplus C'_{s'}
\eeq
then $C*C'$ decomposes under $\cP\wedge\cP'$ (which is a partition of $\Supp(C*C')$)
as
\beq
C*C'=\bigoplus_{i,j}C_i*C'_j
\eeq
where we keep only those of the $i,j$ for which $C_i*C'_j\neq0$.
(However, these $C_i*C'_j$ need not necessarily be indecomposable, even if the $C_i$ and $C'_j$ are.)

And for any $t\geq1$, the $t$-th power $C\deux[t]$ also decomposes under $\cP$ as
\beq
C\deux[t]=C_1\deux[t]\oplus\cdots\oplus C_s\deux[t].
\eeq
(However, these $C_i\deux[t]$ need not necessarily be indecomposable, even if the $C_i$ are.)
\end{corollary}
\begin{proof}
Everything is clear and can be proved directly.
An alternative proof for the first assertion is as a consequence of Propositions~\ref{proprietes_algebre} and~\ref{structure_algebre}.
\end{proof}

\begin{example}
Note that the parity $[3,2,2]_2$ code $C$ is indecomposable, while its square is the
trivial $[3,3,1]_2$ code, which decomposes totally. That is, this gives an example
where $\cP(C\deux)=\{\{1\},\{2\},\{3\}\}$ strictly refines $\cP(C)=\{\{1,2,3\}\}$,
and $\cA(C)=\cA(C)\deux=\un\subsetneq\cA(C\deux)=(\F_2)^3$.
\end{example}

\entry
\label{adapt_algebres}
We gave results only for the proper stabilizing algebra. However, since $\widehat{\cA}(C)=\cA(C)\oplus\iota(\F^{[n]\moins\Supp(C)})$,
one immediately deduces similar statements for the extended algebra.

For instance, Proposition~\ref{proprietes_algebre} is replaced with $\widehat{\cA}(C)*\widehat{\cA}(C')\subset\widehat{\cA}(C*C')$,
$\widehat{\cA}(C)=\widehat{\cA}(C)\deux[t]\subset\widehat{\cA}(C\deux[t])$,
and $\widehat{\cA}(\widehat{\cA}(C))=\widehat{\cA}(\cA(C))=\widehat{\cA}(C)$.

Instead of Definition~\ref{def_decomposable}, we say that $C$ \emph{weakly decomposes} under a partition $\cQ$ of $[n]$ if, for
each $Q\in\cQ$, we have $1_Q\in\widehat{\cA}(C)$. This means there are subcodes $C_i\subset C$ with disjoint supports
such that $C=\bigoplus_i C_i$, with each $\Supp(C_i)$ included (possibly strictly) in some $Q_i\in\cQ$.

There is a finest partition of $[n]$ under which $C$ weakly decomposes, it is $\widehat{\cP}(C)=\cP(C)\cup\{\{j\};j\not\in\Supp(C)\}$.
Then Proposition~\ref{structure_algebre} becomes
\beq
\widehat{\cA}(C)=\bigoplus_{Q\in\widehat{\cP}(C)}\langle1_Q\rangle.
\eeq

Last, $\widehat{\cP}(C*C')$ is a (possibly strict) refinement of $\widehat{\cP}(C)\wedge\widehat{\cP}(C')$, and if $C$ weakly decomposes
under $\cQ$ as $C=\bigoplus_i C_i$ and $C'$ weakly decomposes under $\cQ'$ as $C'=\bigoplus_i C'_i$, then $C*C'$ weakly decomposes
under $\cQ\wedge\cQ'$ as $C*C'=\bigoplus_{i,j} C_i*C'_j$.

Additional properties of $\widehat{\cA}(C)$, involving the dual code $C^{\perp}$, will be given in \ref{AC_ACperp} and \ref{CCperpperp}.

\mysubsection{Repeated columns}
\entry
\label{def_eq_prop}
We keep the same notations as in~\ref{debut_support}: by the columns of a linear code $C\subset\F^n$ we mean the $n$ coordinate
projections $C\longto\F$. Then:

\begin{definition*}
We define an equivalence relation $\sim$ (or $\sim_C$) on $\Supp(C)$
by setting $i\sim j$ when the $i$-th and $j$-th columns of $C$ are proportional.
By abuse of language we also say these are two \emph{repeated} columns.

We let
\beq
\cU(C)=\Supp(C)/\sim
\eeq
be the set of equivalence classes of $\sim$, which is a partition of $\Supp(C)$.
\end{definition*}

\begin{lemma}
\label{lemme_col_prop}
Let $i,j\in\Supp(C)$, $i\neq j$. Then
\beq
i\sim j\quad\equivaut\quad\exists x\in C^\perp,\,\Supp(x)=\{i,j\}.
\eeq
Conversely
\beq
i\not\sim j\quad\equivaut\quad\exists c\in C,\,\pi_i(c)=1,\,\pi_j(c)=0
\eeq
(and then likewise with $i,j$ permuted).
\end{lemma}
\begin{proof}
Basic manipulation in linear algebra.
\end{proof}

\begin{proposition}
Let $C\subset\F^n$ be a linear code. Then $\cU(C)$ is a refinement of $\cP(C)$.
\end{proposition}
\begin{proof}
We have to show that if $A,B\in\cP(C)$, $A\neq B$, and $i\in A$, $j\in B$, then $i\not\sim j$.
Since $i\in\Supp(C)$, we can find $c\in C$ with $\pi_i(c)=1$. Then $1_A*c\in C$ satisfies $\pi_i(1_A*c)=1$, $\pi_j(1_A*c)=0$, and we conclude with Lemma~\ref{lemme_col_prop}.
\end{proof}

\entry
\label{def_slices}
An equivalent formulation for Lemma~\ref{lemme_col_prop} is: $i\not\sim j$ if and only if $\dim(1_{\{i,j\}}*C)=2$.

Conversely a subset $B\subset\Supp(C)$ is contained an equivalence class for $\sim$ if and only if $\dim(1_B*C)=1$.
In particular, $B\in\cU(C)$ if and only if $B$ is maximal for this property.

\begin{definition*}
We call these $1_B*C$, for $B\in\cU(C)$, the \emph{one-dimensional slices} of $C$.
\end{definition*}

If $c\in C$ is nonzero over $B$, then $v=1_B*c$ is a generator of the corresponding slice: $1_B*C=\langle v\rangle$.

Beware that since $\cU(C)$ might be a strict refinement of $\cP(C)$, this slice $1_B*C$ need not actually be a subcode of $C$,
or equivalently, $v$ need not actually belong to $C$.

\entry
\label{repete_et_dual}
If $\cU(C)=\{B_1,\dots,B_s\}$ and $v_1,\dots,v_s$ are corresponding slice generators, then $1_{[\Supp(C)]}=1_{B_1}+\dots+1_{B_s}$
from which it follows 
\beq
C=1_{[\Supp(C)]}*C\;\subset\;\langle1_{B_1},\dots,1_{B_s}\rangle*C\,=\,\langle v_1\rangle\oplus\cdots\oplus\langle v_s\rangle.
\eeq
The right hand side is easily identified
thanks to~\ref{support_et_dual} and Lemma~\ref{lemme_col_prop}:
\beq
\langle v_1\rangle\oplus\cdots\oplus\langle v_s\rangle=\langle x\in C^{\perp}\,;\;w(x)\leq 2\rangle^{\perp}.
\eeq
As a consequence we retrieve the relation
\beq
n_2=\dim\langle x\in C^{\perp}\,;\;w(x)\leq 2\rangle^{\perp}=s=\abs{\cU(C)}
\eeq
as stated in~\ref{longueurs}.

\entry
\label{concrete_1dim}
To restate all this more concretely,
choose a set of representatives $S=\{j_1,\dots,j_s\}\subset\Supp(C)$, with $j_i\in B_i$, so each nonzero column of $C$ is repeated
from one (and only one) column indexed by $S$.
Then a codeword $c\in C$ is entirely determined over $B_i$ by its value at $j_i$.
More precisely, after possibly multiplying by scalars,
we can suppose our slice generators are \emph{normalized} with respect to $S$, that is, $v_i$ is $1$ at $j_i$ for all $i$.
Then for each $c\in C$, the slice of $c$ over $B_i$ is
$1_{B_i}*c=\pi_{j_i}(c)v_i$.

Said otherwise, $\pi_S$ induces a commutative diagram
\begin{center}
\begin{tikzpicture}
\node (NW) at (0,0) {$C$};
\node (NE) at (3,0) {$\langle v_1\rangle\oplus\cdots\oplus\langle v_s\rangle$};
\node (SW) at (0,-1.2) {$\pi_S(C)$};
\node (SE) at (3,-1.2) {$\F^S$};
\draw[white] (NW) edge node[black] {$\subset$} (NE);
\draw[white] (SW) edge node[black] {$\subset$} (SE);
\draw[->,font=\scriptsize,>=angle 90] (NW) edge node[rotate=270,above] {$\simeq\,$} (SW);
\draw[->,font=\scriptsize,>=angle 90] (NE) edge node[rotate=270,above] {$\simeq\,$} (SE);
\end{tikzpicture}
\end{center}
identifying $C$ with the code $\pi_S(C)$ of length $\abs{S}=s=n_2$, which has full support in $\F^S$ and no repeated column 
(so dual distance $\ddual(\pi_S(C))\geq3$ by~\ref{support_et_dual} and Lemma~\ref{lemme_col_prop}).
Each column of $C$ is repeated from one column of $\pi_S(C)$, or more precisely, each
\beq
(\lambda_{j_1},\dots,\lambda_{j_s})\in\pi_S(C)
\eeq
extends uniquely to
\beq
\lambda_{j_1}v_1+\cdots+\lambda_{j_s}v_s\in C.
\eeq

\begin{proposition}
\label{tranches_puissances}
Let $C,C'\subset\F^n$ be linear codes of the same length, and let $i,j\in\Supp(C)\cap\Supp(C')$. Then the $i$-th and $j$-th columns are
repeated in $C*C'$ if and only if they are repeated in $C$ and in $C'$. Said otherwise,
\beq
\cU(C*C')=\cU(C)\wedge\cU(C').
\eeq
If $v_1,\dots,v_s$ are slice generators for $C$ and $w_1,\dots,w_{s'}$ are slice generators for $C'$,
then those among the $v_i*w_j$ that are nonzero form a family of slice generators for $C*C'$.

In particular $\cU(C\deux[t])=\cU(C)$ for all $t\geq1$, and $(v_1)^t,\dots,(v_s)^t$ are slice generators for $C\deux[t]$. If $S\subset\Supp(C)$ is a set of representatives for $\sim_C$, then the dimension sequences of $C$ and $\pi_S(C)$ are the same:
\beq
\dim(C\deux[t])=\dim(\pi_S(C)\deux[t])
\eeq
for all $t\geq0$. Hence they also have the same regularity: $r(C)=r(\pi_S(C))$.
\end{proposition}
\begin{proof}
Suppose $\pi_i=\lambda\pi_j$ on $C$ and $\pi_i=\lambda'\pi_j$ on $C'$, for some $\lambda,\lambda'\in\F^\times$.
Then $\pi_i=\lambda\lambda'\pi_j$ on $C*C'$: indeed it is so on elementary product vectors, and this extends by linearity.

Conversely, suppose for example $i\not\sim_{C}j$, so by Lemma~\ref{lemme_col_prop} we can find $c\in C$ with $\pi_i(c)=1$, $\pi_j(c)=0$.
Since $i\in\Supp(C')$, we can find $c'\in C'$ with $\pi_i(c')=1$.
Then $\pi_i(c*c')=1$, $\pi_j(c*c')=0$, hence $i\not\sim_{C*C'}j$.

The rest follows easily (note $\pi_S(C\deux[t])=\pi_S(C)\deux[t]$).
\end{proof}

\mysubsection{Extension of scalars}
\vspace{.5\baselineskip}

Let $\F\subset\K$ be a field extension.
In many applications, one is given a ``nice'' linear code over $\K$ and one wants to deduce from it a ``nice'' linear code over $\F$.
Several techniques have been designed for this task, especially when the extension has finite degree: subfield subcodes, trace codes, and concatenation.
How these operations behave with respect to the product $*$ turns out to be quite difficult to analyze, although we will give
results involving concatenation in~\ref{debut_concat} and following.

In the other direction, base field extension (or extension of scalars) allows to pass from a linear code $C\subset\F^n$ over $\F$ to a linear code $C_{\K}\subset\K^n$ over $\K$.
In general this operation is less useful for practical applications, however in some cases it can be of help in order to prove theorems. The definition is simple: we let $C_{\K}$ be the $\K$-linear span of $C$ in $\K^n$ (where we implicitly used the chain of inclusions $C\subset\F^n\subset\K^n$).

\begin{lemma}
\label{def_extension}
Let $C\subset\F^n$ be a linear code over $\F$. Then:
\begin{enumerate}[(i)]
\item\label{def_extension_tens} The inclusion $C\tens_{\F}\K\subset\F^n\tens_{\F}\K=\K^n$ induces the identification $C\tens_{\F}\K=C_{\K}$.
\item\label{def_extension_G} If $G$ is a generator matrix for $C$ over $\F$, then $G$ is a generator matrix for $C_{\K}$ over $\K$.
\item\label{def_extension_H} If $H$ is a parity-check matrix for $C$ over $\F$, then $H$ is a parity-check matrix for $C_{\K}$ over $\K$.
\end{enumerate}
\end{lemma}
\begin{proof}
Basic manipulation in linear algebra.
\end{proof}

Extension of scalars is compatible with most operations on codes:
\begin{lemma}
\label{compatibilites_extension}
\begin{enumerate}[(i)]
\item\label{compatibilites_extension_dual}
If $C\subset\F^n$ is a linear code, then
\beq
(C^\perp)_{\K}=(C_{\K})^\perp\quad\subset\K^n.
\eeq
\item\label{compatibilites_extension_inclusion}
If $C,C'\subset\F^n$ are linear codes, then
\beq
C\subset C'\quad\equivaut\quad C_{\K}\subset C'_{\K}.
\eeq
\item\label{compatibilites_extension_+*}
Let $C,C'\subset\F^n$ be linear codes. Then:
\beq
(C+C')_{\K}=C_{\K}+C'_{\K}
\eeq
\beq
(C\cap C')_{\K}=C_{\K}\cap C'_{\K}
\eeq
and
\beq
(C*C')_{\K}=C_{\K}*C'_{\K}
\eeq
(where on the left hand side, $*$ denotes product in $\F^n$, and on the right hand side, in $\K^n$).
\item\label{compatibilites_extension_tens}
Let $C\subset\F^m,C'\subset\F^n$ be a linear codes. Then:
\beq
(C\oplus C')_{\K}=C_{\K}\oplus C'_{\K}\quad\subset\K^{m+n}
\eeq
\beq
(C\tens C')_{\K}=C_{\K}\tens C'_{\K}\quad\subset\K^{m\times n}.
\eeq
\end{enumerate}
\end{lemma}
\begin{proof}
Routine verifications, using Lemma~\ref{def_extension}.
\end{proof}

\begin{proposition}
\label{compatibilites_extension_decomp}
If $C\subset\F^n$ is a linear code, then $\cP(C_{\K})=\cP(C)$, and $\cA(C_{\K})=\cA(C)_{\K}$ in $\K^n$.
In particular $C$ is indecomposable if and only if $C_{\K}$ is indecomposable.
\end{proposition}
\begin{proof}
For $P\subset\Supp(C)$ we have $1_P*C\subset C$ $\equivaut$ $1_P*C_{\K}\subset C_{\K}$
by Lemma~\ref{compatibilites_extension}\itemref{compatibilites_extension_inclusion}-\itemref{compatibilites_extension_+*}.
Conclude with Proposition~\ref{structure_algebre}.
\end{proof}

\begin{proposition}
\label{compatibilites_extension_tranches}
If $C\subset\F^n$ is a linear code, then $\cU(C_{\K})=\cU(C)$. If $v_1,\dots,v_s\in\F^n$ are
slice generators for $C$, then they are also for $C_{\K}$.
\end{proposition}
\begin{proof}
Obvious: $\pi_i=\lambda\pi_j$ on $C$ $\equivaut$ $\pi_i=\lambda\pi_j$ on $C_{\K}$.
\end{proof}

\entry
For $C\subset\F^n$ and $S\subset[n]$ we let
\beq
C_S=\iota_S(\iota_S^{-1}(C))=C\cap\iota_S(\F^S)=\{c\in C\,;\;\Supp(c)\subset S\}
\eeq
be the largest subcode of $C$ with support in $S$.

Also we recall from~\cite{Wei} that for $1\leq i\leq\dim(C)$, the $i$-th generalized Hamming weight $w_i(C)$ of $C$ is the smallest integer $s$
such that $C$ admits a linear subcode of dimension $i$ and support size $s$. Equivalently:
\beq
w_i(C)=\min\{\abs{S}\,;\;S\subset[n],\,\dim(C_S)\geq i\}.
\eeq
In particular $w_1(C)=\dmin(C)$.

\begin{proposition}
\label{extension_restriction}
Let $\F\subset\K$ be a field extension, and $C\subset\F^n$ a linear code over $\F$.
Then for any subset $S\subset[n]$ we have
\beq
(C_S)_{\K}=(C_{\K})_S.
\eeq
In particular, $\dim_{\K}((C_{\K})_S)=\dim_{\F}(C_S)$.
\end{proposition}
\begin{proof}
Write $C_S=C\cap\iota_S(\F^S)$ and use Lemma~\ref{compatibilites_extension}\itemref{compatibilites_extension_+*} (for $\cap$).
\end{proof}

\begin{corollary}
\label{parametres_extension}
Let $\F\subset\K$ be a field extension, and $C\subset\F^n$ a linear code over $\F$. Then we have
\beq
\dim_{\K}(C_{\K})=\dim_{\F}(C)
\eeq
and
\beq
w_i(C_{\K})=w_i(C)
\eeq
for all $i$. In particular, $\dmin(C_{\K})=\dmin(C)$.
\end{corollary}
\begin{proof}
The first equality follows from Lemma~\ref{def_extension}\itemref{def_extension_tens}.
The second follows from Proposition~\ref{extension_restriction} and the definition of the generalized Hamming weights. 
\end{proof}

\begin{lemma}
\label{descente_mots}
Let $\F\subset\K$ be a field extension, and suppose $\lambda_1,\dots,\lambda_r\in\K$ are linearly independent over $\F$.
Let $C\subset\F^n$ be a linear code, and $x_1,\dots,x_r\in\F^n$ be arbitrary words.
Set
\beq
x=\lambda_1x_1+\cdots+\lambda_rx_r\in\K^n.
\eeq
Then we have
\beq
\Supp(x)=\bigcup_i\Supp(x_i)
\eeq
and
\beq
x\in C_{\K}\quad\equivaut\quad\forall i,\,x_i\in C.
\eeq
\end{lemma}
\begin{proof}
The only nontrivial point is the implication $x\in C_{\K}$ $\implique$ $x_i\in C$. It is in fact a consequence of Lemma~\ref{def_extension}\itemref{def_extension_H}.
\end{proof}

Alternative proofs for Propositions~\ref{compatibilites_extension_decomp} and~\ref{extension_restriction}
could also be given using the following:

\begin{proposition}
\label{descente_sous-code}
Let $\F\subset\K$ be a field extension, $C\subset\F^n$ a linear code over $\F$, and $C'\subset C_{\K}$ a linear subcode over $\K$.
Then there is a linear subcode $C_0\subset C$ over $\F$
of support
\beq
\Supp(C_0)=\Supp(C')
\eeq
such that
\beq
C'\subset(C_0)_{\K}
\eeq
and
\beq
\dim_{\K}(C')\leq\dim_{\F}(C_0)\leq\min(\dim_{\F}(C),[\K:\F]\dim_{\K}(C')).
\eeq
Moreover, if the extension is finite separable, we can take
\beq
C_0=\tr_{\K/\F}(C')
\eeq
where we extended the trace map $\tr_{\K/\F}$ to a map $\K^n\longto\F^n$ by letting it act componentwise.
\end{proposition}
\begin{proof}
Choose a basis $(\lambda_i)$ of $\K$ over $\F$, and
decompose each element $x$ of a ($\K$-)basis of $C'$ as a finite sum
$x=\lambda_1x_1+\cdots+\lambda_rx_r$ for some $x_i\in\F^n$ (after possibly renumbering the $\lambda_i$).
Then apply Lemma~\ref{descente_mots}.
When the extension is separable we have $x_i=\tr_{\K/\F}(\lambda_i^*x)$,
where the basis $(\lambda_i^*)$ is dual to $(\lambda_i)$ with respect to the trace bilinear form.
\end{proof}

\begin{lemma}
\label{mot_plein}
Let $C\subset\F^n$ be a linear code of dimension $k$ over $\F$. Then there exists an extension field $\K$ of finite degree $[\K:\F]\leq k$, and a codeword $c\in C_{\K}$, such that
\beq
\Supp(c)=\Supp(C).
\eeq
\end{lemma}
\begin{proof}
Let $G$ be a generator matrix for $C$. Let $\F_0\subset\F$ be the prime subfield of $\F$
(that is, $\F_0=\Q$ is $\car(\F)=0$, and $\F_0=\Z/p\Z$ if $\car(\F)=p>0$), and let $\F_1=\F_0(G)\subset\F$
be the field generated over $\F_0$ by the entries of $G$.

So $\F_1$ is finitely generated over a prime field, and as such we contend it admits finite extensions of any degree
(this is clear if $\F$, and thus also $\F_1$, is a finite field, which is the case in most applications; and for completeness a proof of the general case will be given five lines below).

Let then $\K_1$ be an extension of $\F_1$ of degree $k$, and let $\K$ be a compositum of $\F$ and $\K_1$.
Now if $c_1,\dots,c_k\in\F^n$ are the rows of $G$, and if $\lambda_1,\dots,\lambda_k\in\K_1$ are linearly independent over $\F_1$,
we set $c=\lambda_1c_1+\cdots+\lambda_kc_k\in C_{\K}$ and conclude with Lemma~\ref{descente_mots}.
\end{proof}

Concerning the general case of the claim made in the middle of this proof, it can be established as follows:
write $\F_1$ as a finite extension (say of degree $d$) of a purely transcendental extension of $\F_0$, and let $\F_{1,c}$ be its constant field, that is, the algebraic closure of $\F_0$ in $\F_1$;
proceeding as in \cite{Stichtenoth} Prop.~3.6.1 and Lemma~3.6.2, one then gets (a) that $\F_{1,c}$ is finite over $\F_0$ (more precisely, of degree at most $d$), and (b) that any algebraic extension of $\F_{1,c}$ is linearly disjoint from $\F_1$.
Now (a) means $\F_{1,c}$ is either a number field or a finite field, and as such it admits finite extensions of any degree, for instance cyclotomic extensions do the job; and then by (b), such an extension of $\F_{1,c}$ induces an extension of $\F_1$ of the same degree.
(An alternative, more geometric proof, would be to consider $\F_1$ as the field of functions of a projective variety over $\F_0$,
and then get properties (a) and (b) of $\F_{1,c}$ from finiteness of cohomology and its properties under base field extension.)

\mysubsection{Monotonicity}
\begin{theorem}
\label{monotone}
Let $C\subset\F^n$ be a linear code. Then for $t\geq 1$ we have
\beq
\dim(C\deux[t+1])\geq\dim(C\deux[t]).
\eeq
Also the generalized Hamming weights satisfy
\beq
w_i(C\deux[t+1])\leq w_i(C\deux[t])
\eeq
for $1\leq i\leq\dim(C\deux[t])$, and
\beq
w_i((C\deux[t+1])^\perp)\geq w_i((C\deux[t])^\perp)
\eeq
for $1\leq i\leq\dim((C\deux[t+1])^\perp)$.

In particular, the minimum distances satisfy $\dmin(C\deux[t+1])\leq\dmin(C\deux[t])$,
and the dual distances, $\ddual(C\deux[t+1])\geq\ddual(C\deux[t])$.
\end{theorem}
\begin{proof}
Thanks to Lemmas~\ref{compatibilites_extension}\itemref{compatibilites_extension_dual},\itemref{compatibilites_extension_+*} and~\ref{mot_plein}, and Corollary~\ref{parametres_extension}, it suffices to treat
the case where there is $c\in C$ with $\Supp(c)=\Supp(C)$.

The multiplication map $c*\cdot$ is then injective from $C\deux[t]$ into $C\deux[t+1]$, so
\beq
\dim(C\deux[t])=\dim(c*C\deux[t])\leq\dim(C\deux[t+1]).
\eeq
Likewise if $C'\subset C\deux[t]$ has dimension $i$ and support weight $w_i(C\deux[t])$, we have $\dim(c*C')=\dim(C')=i$
and
\beq
w_i(C\deux[t+1])\leq\abs{\Supp(c*C')}=\abs{\Supp(C')}=w_i(C\deux[t]).
\eeq

Now extend $c$ to $\widetilde{c}\in(\F^n)^\times$ by setting it equal to $1$ out of $\Supp(C)$, that is, formally,
\beq
\widetilde{c}=c+1_{[n]\moins\Supp(C)}.
\eeq
The multiplication map $\widetilde{c}*\cdot$ is then injective as a linear endomorphism of $\F^n$, and on $C\deux[t]$ it coincides with the multiplication map $c*\cdot$ as above.
So $\widetilde{c}*\cdot$ sends $C\deux[t]$ into $C\deux[t+1]$, which implies that it sends $(C\deux[t+1])^\perp$ into $(C\deux[t])^\perp$
(here this is easily checked, but see Corollary~\ref{c*_dual} to put it in a more general context).

Then if $C'\subset (C\deux[t+1])^\perp$ has dimension $i$ and support weight $w_i((C\deux[t+1])^\perp)$, we have $\dim(\widetilde{c}*C')=\dim(C')=i$
and
\beq
w_i((C\deux[t])^\perp)\leq\abs{\Supp(\widetilde{c}*C')}=\abs{\Supp(C')}=w_i((C\deux[t+1])^\perp).
\eeq
\end{proof}

This shows that the regularity $r(C)$ is well defined (in \ref{def_CMreg}). One can then give a slightly stronger monotonicity
result for the dimension sequence:
\begin{corollary}
\label{strict_monotone}
For $1\leq t< r(C)$, we have
\beq
\dim(C\deux[t+1])>\dim(C\deux[t]).
\eeq
\end{corollary}
\begin{proof}
Again it suffices to treat the case where there is $c\in C$ with $\Supp(c)=\Supp(C)$.
Let $t\geq1$, and suppose $\dim(C\deux[t+1])=\dim(C\deux[t])$. Then necessarily
\beq
C\deux[t+1]=c*C\deux[t]
\eeq
so
\beq
C\deux[t+2]=C*C\deux[t+1]=C*(c*C\deux[t])=c*C\deux[t+1]=c^2*C\deux[t].
\eeq
We continue in the same way and, for all $i\geq0$, we find $C\deux[t+i]=c^i*C\deux[t]$,
hence $\dim(C\deux[t+i])=\dim(C\deux[t])$. This means precisely $t\geq r(C)$.
\end{proof}

An alternative proof can be given using Proposition~\ref{tranches_puissances}
to reduce to the case where $C$ has dual distance at least~$3$, and then concluding
with Proposition~\ref{dim_et_ddual} below.

\mysubsection{Stable structure}
\entry
\label{notations_stable}
In what follows we use the same notations as in \ref{repete_et_dual}-\ref{concrete_1dim}.
So $C\subset\F^n$ is a linear code and
\beq
\abs{\cU(C)}=\dim\langle x\in C^{\perp}\,;\;w(x)\leq 2\rangle^{\perp}=n_2
\eeq
is its projective length. We choose a set
of representatives $S=\{j_1,\dots,j_{n_2}\}\subset\Supp(C)$, and associated normalized slice generators $v_1,\dots,v_{n_2}$ for $C$,
that is, $v_i\in\F^n$ with pairwise disjoint supports such that
\beq
C\subset\langle v_1\rangle\oplus\cdots\oplus\langle v_{n_2}\rangle
\eeq
and $\pi_{j_i}(v_i)=1$, so we have an isomorphism
\beq
\begin{array}{cccc}
\phi: & \pi_S(C) & \overset{\simeq}{\longto} & C\\
& (\lambda_{j_1},\dots,\lambda_{j_{n_2}}) & \mapsto & \lambda_{j_1}v_1+\cdots+\lambda_{j_{n_2}}v_{n_2}
\end{array}
\eeq
inverse to $\pi_S$.
Here $\pi_S(C)\subset\F^S$ has full support and no repeated column.

Then by Proposition~\ref{tranches_puissances} for all $t\geq1$, we have an inclusion
\beq
C\deux[t]\subset\langle (v_1)^t\rangle\oplus\cdots\oplus\langle (v_{n_2})^t\rangle
\eeq
and an isomorphism
\beq
\begin{array}{cccc}
\phi_t: & \pi_S(C)\deux[t] & \overset{\simeq}{\longto} & C\deux[t]\\
& (\lambda_{j_1},\dots,\lambda_{j_{n_2}}) & \mapsto & \lambda_{j_1}(v_1)^t+\cdots+\lambda_{j_{n_2}}(v_{n_2})^t
\end{array}
\eeq
inverse to $\pi_S$ (observe $\pi_S(C)\deux[t]=\pi_S(C\deux[t])$).

\begin{theorem}
\label{descr_stable}
Let $C\subset\F^n$ be a linear code of dimension $k$.
Then, with the notations of~\ref{notations_stable}, the code $C$ has regularity $r(C)\leq n_2-k+1$, and we have
\beq
C\deux[t]=\langle (v_1)^t\rangle\oplus\cdots\oplus\langle (v_{n_2})^t\rangle
\eeq
for all $t\geq r(C)$.

In particular the stable value of the dimension sequence of $C$ is its projective length:
$\dim(C\deux[t])=n_2$ for $t\geq r(C)$.
\end{theorem}
\begin{proof}
Because of the inclusion $C\deux[t]\subset\langle (v_1)^t\rangle\oplus\cdots\oplus\langle (v_{n_2})^t\rangle$ we have $\dim(C\deux[t])\leq n_2$
for all $t$. However for $t=1$ we have $\dim(C)=k$. So the dimension sequence can increase at most $n_2-k$ times, which, joint with Corollary~\ref{strict_monotone}, implies the bound on $r(C)$.

To conclude it suffices to show that there exists one $t$ with $\dim(C\deux[t])=n_2$.

Since $S$ is a set of representatives for $\sim$, Lemma~\ref{lemme_col_prop} gives,
for all $i,i'\in S$, $i\neq i'$, a word $x_{i,i'}\in\pi_S(C)$ which is $1$ at $i$ and $0$ at $i'$.
Fixing $i$ and letting $i'\neq i$ vary we find
\beq
1_{\{i\}}=\prod_{i'\in S\moins\{i\}}x_{i,i'}\in \pi_S(C)\deux[n_2-1]\;\subset\F^S.
\eeq
Now this holds for all $i\in S$, so $\dim(\pi_S(C)\deux[n_2-1])=n_2$.
Then applying $\phi_{n_2-1}$ we find $\dim(C\deux[n_2-1])=n_2$ as claimed.
\end{proof}

\begin{corollary}
Let $C\subset\F^n$ be a linear code and $t\geq0$ an integer.
The following are equivalent:
\begin{enumerate}[(i)]
\item $t\geq r(C)$
\item $\dim(C\deux[t])=\dim(C\deux[t+1])$
\item $\dim(C\deux[t])=n_2$ the projective length of $C$
\item $C\deux[t]$ is generated by some codewords with pairwise disjoint supports
\item $(C\deux[t])^\perp$ is generated by its codewords of weight at most $2$
\item there is a subset $S\subset[n]$ such that $\pi_S:C\deux[t]\surj\F^S$ is onto,
and every nonzero column of $C\deux[t]$ is repeated from a column indexed by $S$.
\end{enumerate}
\end{corollary}
\begin{proof}
The first equivalence (i)$\equivaut$(ii) is essentially Corollary~\ref{strict_monotone}.
Also (iv)$\equivaut$(v)$\equivaut$(vi) is clear.
Now Theorem~\ref{descr_stable} gives (i)$\equivaut$(iii)$\implique$(iv).
Conversely suppose (iv), so $C\deux[t]$ is generated by $r=\dim(C\deux[t])$ codewords $c_1,\dots,c_r$ with pairwise disjoint supports.
Then necessarily these codewords are slice generators for $C\deux[t]$, so $r=\abs{\cU(C\deux[t])}=\abs{\cU(C)}$,
where for the last equality we used Proposition~\ref{tranches_puissances}.
But by~\ref{repete_et_dual} we have $\abs{\cU(C)}=n_2$, so (iii) holds.
\end{proof}

\mysubsection{Adjunction properties}
\entry
\label{discussion_trace}
If $\cA$ is a finite dimensional algebra over $\F$ (with unit), then letting $\cA$ act on itself by multiplication (on the left) allows to identify $\cA$
with a subalgebra of the algebra of linear endomorphisms $\End(\cA)$. We then define the trace linear form on $\cA$ as the linear form inherited
from the usual trace in $\End(\cA)$, so formally $\tr(a)=\tr(x\mapsto ax)$ for $a\in\cA$. We also define the trace bilinear form $\langle\cdot|\cdot\rangle$ on $\cA$, by the formula $\langle x|y\rangle=\tr(xy)$ for $x,y\in\cA$.
Moreover the identity $\tr(xy)=\tr(yx)$ then shows that $\langle\cdot|\cdot\rangle$ is in fact a \emph{symmetric} bilinear form,
and that for any $a$, the left-multiplication-by-$a$ map $a\cdot$ and
the right-multiplication-by-$a$ map $\cdot a$, which are elements of $\End(\cA)$, are adjoint to each other with respect to $\langle\cdot|\cdot\rangle$.

In the particular case $\cA=\F^n$, this construction identifies $\F^n$ with the algebra of diagonal matrices of size $n$ over $\F$.
The trace function is the linear map $\tr(x)=x_1+\cdots+x_n$ and the trace bilinear form is the standard scalar product $\langle x|y\rangle=x_1y_1+\cdots+x_ny_n$,
where $x=(x_1,\dots,x_n)$, $y=(y_1,\dots,y_n)$, $x_i,y_j\in\F$.

\begin{proposition}
\label{c*_autoadjoint}
For any $c\in\F^n$, the multiplication-by-$c$ map $c*\cdot$ acting on $\F^n$ is autoadjoint
with respect to the standard scalar product:
\beq
\langle c*x|y\rangle=\langle x|c*y\rangle
\eeq
for all $x,y\in\F^n$.
\end{proposition}
\begin{proof}
This can be checked directly, or seen as a special case of the discussion~\ref{discussion_trace}.
\end{proof}

\begin{corollary}
\label{c*_dual}
Let $C_1,C_2\subset\F^n$ be two linear codes. Then for any $c\in\F^n$ we have
\beq
c*C_1\subset C_2\quad\equivaut\quad c*C_2^\perp\subset C_1^\perp
\eeq
and for any linear code $C\subset\F^n$ we have
\beq
C*C_1\subset C_2\quad\equivaut\quad C*C_2^\perp\subset C_1^\perp.
\eeq
\end{corollary}
\begin{proof}
The first assertion is a consequence of Proposition~\ref{c*_autoadjoint}, and the second follows after passing to the linear span.
\end{proof}

\begin{corollary}
\label{cor_adj}
For any two linear codes $C,C'\subset\F^n$ we have
\beq
C*(C*C')^\perp\subset C'^\perp.
\eeq
Given two integers $t'\geq t\geq 0$ we have
\beq
C\deux[t]*(C\deux[t'])^\perp\subset(C\deux[t'-t])^\perp.
\eeq
(Note: it is easy to construct examples where the inclusion is strict.)
\end{corollary}
\begin{proof}
Apply the second equivalence in Corollary~\ref{c*_dual}.
\end{proof}

Recall from \ref{def_stab_alg}, \ref{structure_algebre}, and \ref{adapt_algebres} we described the extended stabilizing algebra
\beq
\widehat{\cA}(C)=\{a\in\F^n\,;\;a*C\subset C\}
\eeq
of a linear code $C$ as
\beq
\widehat{\cA}(C)=\bigoplus_{Q\in\widehat{\cP}(C)}\langle1_Q\rangle,
\eeq
where $\widehat{\cP}(C)=\cP(C)\cup\{\{j\};j\not\in\Supp(C)\}$ is the partition of $[n]$ associated with the decomposition
of $C$ into indecomposable components.

\begin{corollary}
\label{AC_ACperp}
For any linear code $C$ we have $\widehat{\cA}(C)=\widehat{\cA}(C^\perp)$,
hence also $\widehat{\cP}(C)=\widehat{\cP}(C^\perp)$.

In particular, a linear code $C$ of length $n\geq2$ is indecomposable with full support if and only if $C^\perp$ is.
\end{corollary}
\begin{proof}
Apply the first equivalence in Corollary~\ref{c*_dual}.
\end{proof}

The following interesting characterization of $\widehat{\cA}(C)$ was apparently first noticed by Couvreur and Tillich:

\begin{corollary}
\label{CCperpperp}
For any linear code $C$ we have
\beq
\widehat{\cA}(C)=(C*C^\perp)^\perp,
\eeq
or equivalently, $C*C^\perp$ is the space of words orthogonal to the $1_Q$ for $Q\in\widehat{\cP}(C)$.

In particular, if $C$ is indecomposable with full support, of length $n\geq2$, then
\beq
C*C^\perp=\un^\perp
\eeq
is the $[n,n-1,2]$ parity code.
\end{corollary}
\begin{proof}
The first inclusion in Corollary~\ref{cor_adj}, applied with $C'=C^\perp$, shows that $(C*C^\perp)^\perp$ stabilizes $C$, hence
\beq
(C*C^\perp)^\perp\subset\widehat{\cA}(C).
\eeq
By duality, to get the converse inclusion, it suffices now to show
\beq
C*C^\perp\subset\widehat{\cA}(C)^\perp,
\eeq
\ie we have to show $C*C^\perp$ orthogonal to the $1_Q$, for $Q\in\widehat{\cP}(C)$.
Now if $Q=\{j\}$ for $j$ out of $\Supp(C)$ this is clear.
Otherwise we have $Q\in\cP(C)$, and projecting onto $Q$ we can now suppose that $C$ is indecomposable with full support,
in which case we have to show $C*C^\perp$ orthogonal to $\un$, which is obvious
(this can also be seen as the second inclusion in Corollary~\ref{cor_adj} applied with $t=t'=1$).

\end{proof}

\mysubsection{Symmetries and automorphisms}
\entry
For any integer $n$, we let the symmetric group $\fS_n$ act on the right on $\F^n$ by the formula
\beq
(x_1,\dots,x_n)^\sigma=(x_{\sigma(1)},\dots,x_{\sigma(n)})
\eeq
where $\sigma\in\fS_n$, $x_i\in\F$.

Equivalently, if $x\in\F^n$ is a row vector, then
\beq
x^\sigma=xP_\sigma
\eeq
where $P_\sigma$ is the permutation matrix with entries $(P_\sigma)_{i,j}=1_{\{i=\sigma(j)\}}$.

For $\sigma,\tau\in\fS_n$ we have $(x^\sigma)^\tau=x^{\sigma\tau}$.

\begin{lemma}
\label{lemme_*_sym}
For $x,y\in\F^n$ and $\sigma\in\fS_n$ we have $x^\sigma*y^\sigma=(x*y)^\sigma$.
\end{lemma}
\begin{proof}
Obvious.
\end{proof}

\begin{definition}
If $C\subset\F^n$ is a linear code, its group of symmetries is
\beq
\fS(C)=\{\sigma\in\fS_n\,;\;C^\sigma=C\}=\{\sigma\in\fS_n\,;\;\forall c\in C,\,c^\sigma\in C\}.
\eeq
\end{definition}

\begin{proposition}
Let $C,C'\subset\F^n$ be two linear codes. Then we have
\beq
\fS(C)\cap\fS(C')\subset\fS(C*C').
\eeq
Given two integers $t,t'\geq1$, then
\beq
t|t'\quad\implique\quad\fS(C\deux[t])\subset\fS(C\deux[t']).
\eeq
\end{proposition}
\begin{proof}
Direct consequence of Lemma~\ref{lemme_*_sym}.
\end{proof}

\begin{example}
\label{ex_sym}
Let $C$ be the one-dimensional code of length $2$ over $\F_5$ generated by the row vector $(1,2)$.
Then $\fS(C\deux[t])=\{1\}$ for $t$ odd, and $\fS(C\deux[t])=\fS_2$ for $t$ even. 
\end{example}

\entry
\label{monomial_tsfo}
We define $\Aut(\F^n)$ as the group with elements the pairs $(\sigma,a)$ with $\sigma\in\fS_n$ and $a\in(\F^n)^\times$,
and composition law given by $(\sigma,a)(\tau,b)=(\sigma\tau,a^\tau*b)$ for $\sigma,\tau\in\fS_n$ and $a,b\in(\F^n)^\times$.
This is a semidirect product of $\fS_n$ and $(\F^n)^\times$, with $(\F^n)^\times$ normal.
We let $\Aut(\F^n)$ act on the right on $\F^n$, where $(\sigma,a)$ acts as
\beq
x\mapsto x^\sigma*a.
\eeq

For $a\in(\F^n)^\times$, let $D(a)\in\F^{n\times n}$ be the associated diagonal matrix.
Then the map
\beq
(\sigma,a)\mapsto P_\sigma D(a)
\eeq
is an isomorphism of $\Aut(\F^n)$ with the group of $n\times n$ monomial matrices (note $D(a^\tau)=P_\tau^{-1}D(a)P_\tau$).
The latter acts on the right on $\F^n$, seen as a space of row vectors, and this isomorphism preserves
the actions.

\begin{lemma}
\label{lemme_*_aut}
For $x,y\in\F^n$, $\sigma\in\fS_n$, and $a,b\in(\F^n)^\times$, we have
$(x^\sigma*a)*(y^\sigma*b)=(x*y)^\sigma*(a*b)$.
\end{lemma}
\begin{proof}
Obvious.
\end{proof}

\entry
By the definition of $\Aut(\F^n)$ as a semidirect product, we have a split exact sequence
\beq
1\longto(\F^n)^\times\longto\Aut(\F^n)\overset{\pi}{\longto}\fS_n\longto 1.
\eeq
\begin{definition*}
Given two subgroups $H,H'$ of $\Aut(\F^n)$, we write
\beq
H\;\,\widehat{\subset}\,\;H'
\eeq
when $\pi(H)\subset\pi(H')$ and $H\cap(\F^n)^\times\subset H'\cap(\F^n)^\times$. If $H'$ is finite
(for example if $\F$ is finite) this implies that $\abs{H}$ divides $\abs{H'}$.
\end{definition*}
We also set $H\;\widehat{\sim}\;H'$ when $H\;\widehat{\subset}\;H'$ and $H'\;\widehat{\subset}\;H$.
This implies $\abs{H}=\abs{H'}$.

\begin{definition}
Let $C\subset\F^n$ be a linear code. Then the group of linear automorphisms of $C$ in $\F^n$ is
\beq
\Aut(C)=\{(\sigma,a)\in\Aut(\F^n)\,;\;C^\sigma*a=C\}.
\eeq
\end{definition}

We also let
\beq
\widehat{\fS}(C)=\pi(\Aut(C))=\{\sigma\in\fS_n\,;\;\exists a\in(\F^n)^\times,\,C^\sigma*a=C\}
\eeq
be the group of \emph{projective symmetries} of $C$, and we note that
\beq
\Aut(C)\cap(\F^n)^\times=\widehat{\cA}(C)^\times
\eeq
by Definition~\ref{def_stab_alg}, so we get an exact sequence
\beq
1\longto\widehat{\cA}(C)^\times\longto\Aut(C)\overset{\pi}{\longto}\widehat{\fS}(C)\longto 1.
\eeq

\begin{proposition}
\label{aut_divise}
Let $C,C'\subset\F^n$ be two linear codes. Then we have
\beq
\widehat{\fS}(C)\cap\widehat{\fS}(C')\subset\widehat{\fS}(C*C')
\eeq
and
\beq
\widehat{\cA}(C)^\times\widehat{\cA}(C')^\times\subset\widehat{\cA}(C*C')^\times.
\eeq
Given two integers $t,t'\geq1$, then
\beq
t|t'\quad\implique\quad\Aut(C\deux[t])\;\,\widehat{\subset}\,\;\Aut(C\deux[t']).
\eeq
\end{proposition}
\begin{proof}
The first inclusion is a direct consequence of Lemma~\ref{lemme_*_aut}.
The second follows from Proposition~\ref{proprietes_algebre} (and~\ref{adapt_algebres}).
Then together they imply the last assertion.
\end{proof}

\entry
Given a subset $S\subset[n]$ and a partition $\cU$ of $S$, we define
\beq
\fS(\cU)=\{\sigma\in\fS_n\,;\;\forall U\in\cU, \sigma(U)\in\cU\}
\eeq
and
\beq
\widehat{\cA}(\cU)=\bigoplus_{U\in\cU}\langle 1_U\rangle\oplus\bigoplus_{j\not\in S}\langle 1_{\{j\}}\rangle
\eeq
which is a subalgebra of $\F^n$.

\begin{proposition}
\label{aut_decomp}
Let $v_1,\dots,v_s\in\F^n$, $v_i\neq0$, be vectors with pairwise disjoint supports, and let
\beq
C=\langle v_1\rangle\oplus\cdots\oplus\langle v_s\rangle\subset\F^n
\eeq
be the linear code they generate. Let $B_i=\Supp(v_i)$, so $\cU=\{B_1,\dots,B_s\}$ is a partition of $\Supp(C)$.
Then we have $\widehat{\fS}(C)=\fS(\cU)$ and $\widehat{\cA}(C)=\widehat{\cA}(\cU)$.
\end{proposition}
\begin{proof}
Obviously the $\langle v_i\rangle$ are the indecomposable components of $C$, and an automorphism of a code must map indecomposable components to indecomposable components. This implies $\widehat{\fS}(C)\subset\fS(\cU)$.
Conversely, let $\sigma\in\fS(\cU)$. We have to construct $a\in(\F^n)^\times$ such that $C^\sigma*a=C$.
First, we set the coordinates of $a$ equal to $1$ out of $\Supp(C)$. Now $\sigma$ determines a permutation $j\mapsto j'$ of $[s]$, such that $\sigma(B_j)=B_{j'}$ (so $\abs{B_{j'}}=\abs{B_j}$).
Then, for $i\in\Supp(C)$, we have $i\in B_j$ for some $j$, and we can just set $\pi_i(a)=\pi_i(v_j)/\pi_{\sigma(i)}(v_{j'})$. This gives $v_{j'}^\sigma*a=v_j$, hence, letting $j'$ vary, $C^\sigma*a=C$ as claimed.

Last, $\widehat{\cA}(C)=\widehat{\cA}(\cU)$ follows from Proposition~\ref{proprietes_algebre} (and~\ref{adapt_algebres}).
\end{proof}

\entry
\label{aut_stable}
Given a subset $S\subset[n]$ and a partition $\cU$ of $S$, we define $C(\cU)=\bigoplus_{U\in\cU}\langle 1_U\rangle$.
\begin{corollary*}
Let $C\subset\F^n$ be a linear code.
Then for all $t\geq r(C)$ we have $\Aut(C\deux[t])\;\widehat{\sim}\:\Aut(C(\cU))$.
\end{corollary*}
\begin{proof}
Consequence of Theorem~\ref{descr_stable} and Proposition~\ref{aut_decomp}.
\end{proof}
Said otherwise, up to $\widehat{\sim}$, the sequence $\Aut(C\deux[t])$ becomes ultimately constant.
For all $t,t'\geq r(C)$ we have $\Aut(C\deux[t])\;\widehat{\sim}\:\Aut(C\deux[t'])$.

\entry
\label{open_pb_sym}
Here are three open problems that, by lack of time, the author did not try to address.

First, considering Example~\ref{ex_sym}, Proposition~\ref{aut_divise}, and Corollary~\ref{aut_stable}, it might be interesting to compare $\Aut(C\deux[t])$ and $\Aut(C\deux[t'])$ for all $t,t'$, not only for $t|t'$ or for $t,t'\geq r(C)$.
By Proposition~\ref{proprietes_algebre} (and~\ref{adapt_algebres}) we have $\widehat{\cA}(C\deux[t])\subset\widehat{\cA}(C\deux[t+1])$ for all $t\geq 1$,
so a key point would be to compare $\widehat{\fS}(C\deux[t])$ and $\widehat{\fS}(C\deux[t+1])$.

Second, note that, as defined, $\Aut(C)$ is a subgroup of $\Aut(\F^n)$, and its action on $C$ need not be faithful.
Another perhaps equally interesting object is the group $\Aut^{\textrm{in}}(C)$ of invertible linear endomorphisms of $C$
(seen as an abstract vector space) that preserve the Hamming metric. We might call $\Aut^{\textrm{in}}(C)$ the group of ``internally defined'' automorphisms
(or isometries) of $C$.
Obviously, an element of $\Aut(C)$ acts on $C$ through an element of $\Aut^{\textrm{in}}(C)$,
and conversely, the McWilliams equivalence theorem~\cite{McWilliams} shows that all elements of $\Aut^{\textrm{in}}(C)$
arise in this way. So we have an identification
\beq
\Aut^{\textrm{in}}(C)=\Aut(C)/\Aut^0(C)
\eeq
where $\Aut^0(C)\subset\Aut(C)$ is the kernel of the action of $\Aut(C)$ on $C$.
It then appears very natural to try compare the $\Aut^{\textrm{in}}(C\deux[t])$ as $t$ varies (and for this, it might be useful to compare the $\Aut^0(C\deux[t])$ first).

Last, we were interested here only in groups acting linearly on codes. However, when $\F$ is a nonprime finite field, we can also consider the action of the Frobenius,
which preserves the Hamming metric, leading to the notion of semilinear automorphism. One could then try to extend the study to this semilinear setting.

\section{Estimates involving the dual distance}
\label{sect_dual_distance}

\entry
\label{def_ddual}
A characterization of the dual distance $\ddual(C)$ of a linear code $C\subset\F^n$
is as the smallest possible length of a linear dependence relation between columns of $C$.
In case $C=\F^n$, there is no such relation, but it might then be convenient to set $\ddual(\F^n)=n+1$.

This can be rephrased as:
\begin{lemma*}
Let $0\leq m\leq n$. Then we have $\ddual(C)\geq m+1$ if and only if, 
for any set of indices $J\subset[n]$ of size $\abs{J}=m$,
and for any $j\in J$,
there is a codeword $y\in C$ with coordinate $\pi_{j}(y)=1$
and $\pi_{j'}(y)=0$ for $j'\in J\moins\{j\}$.
\end{lemma*}
Equivalently, $\ddual(C)\geq m+1$ if and only if, 
for any $J\subset[n]$ of size $\abs{J}=m$, $\dim(\pi_J(C))=m$.

\entry
From this we readily derive the following properties:
\begin{lemma*}
\label{proprietes_ddual}
Let $C\subset\F^n$ be a linear code. Then:
\begin{enumerate}[(i)]
\item\label{proprietes_ddual_sous-code}
For any subcode $C'\subset C$, we have $\ddual(C')\leq\ddual(C)$.
\item\label{proprietes_ddual_projection}
For any set of indices $S\subset[n]$, we have $\ddual(\pi_S(C))\geq\min(\abs{S}+1,\ddual(C))$.
\item\label{proprietes_ddual_MDS}
We have $\ddual(C)\leq\dim(C)+1$ with equality if and only if $C$ is MDS.
\end{enumerate}
\end{lemma*}

\entry
\label{ddual*_et_ddual}
The simplest estimate involving products of codes and the dual distance is probably the following:
\begin{proposition*}
Let $C_1,C_2\subset\F^n$ be two linear codes with full support, i.e. dual distances $d_1^\perp,d_2^\perp\geq2$.
Then we have
\beq
\ddual(C_1*C_2)\geq\min(n+1,\,d_1^\perp+d_2^\perp-2).
\eeq
\end{proposition*}
\begin{proof}
It suffices to show that any subset $J\subset[n]$ of size $m=\min(n,d_1^\perp+d_2^\perp-3)$ satisfies the condition
in Lemma~\ref{def_ddual}. So pick $j\in J$ and write $J\moins\{j\}=A_1\cup A_2$ with $\abs{A_1}=d_1^\perp-2$
and $\abs{A_2}\leq d_2^\perp-2$. Then by Lemma~\ref{def_ddual} we can find $y_1\in C_1$ that is $1$ at $j$ and $0$ over $A_1$,
and $y_2\in C_2$ that is $1$ at $j$ and $0$ over $A_2$. Then $y=y_1*y_2\in C_1*C_2$ is $1$ at $j$ and $0$ over $J\moins\{j\}$
as requested.
\end{proof}

From this one deduces the following estimate, that the author first learned from A.~Couvreur:
\begin{corollary}
\label{reg_et_ddual}
Let $C\subset\F^n$ have full support and no repeated column, i.e. dual distance $d^\perp\geq3$.
Then for all $t\geq1$ we have
\beq
\dim(C\deux[t])\geq\min(n,\,1+(d^\perp-2)t).
\eeq
As a consequence $C$ has regularity
\beq
r(C)\leq\left\lceil\frac{n-1}{d^\perp-2}\right\rceil
\eeq
and for $t\geq r(C)$ we have $C\deux[t]=\F^n$.
\end{corollary}
\begin{proof}
Write $\dim(C\deux[t])\geq\ddual(C\deux[t])-1$ and make induction on $t$ using Proposition~\ref{ddual*_et_ddual}.
\end{proof}

In fact it is possible to say slightly better, as will be seen below.

\begin{proposition}
\label{dim_et_ddual}
Let $C_1,C_2\subset\F^n$ be two linear codes. Suppose $C_2$ has full support, i.e. dual distance $d_2^\perp\geq2$. Then we have
\beq
\dim(C_1*C_2)\geq\min(n_1,\,k_1+d_2^\perp-2)
\eeq
where $n_1=\abs{\Supp(C_1)}$ and $k_1=\dim(C_1)$.
\end{proposition}
\begin{proof}
We can suppose $C_1$ has a generator matrix of the form
\beq
G_1=
\left(\begin{array}{c|c|c}
I_{k_1} & X & 0
\end{array}\right)
\eeq
where $I_{k_1}$ is the $k_1\times k_1$ identity matrix, and $X$ is a $k_1\times (n_1-k_1)$ matrix with no zero column. 
Then, multiplying rows of $G_1$ with suitable codewords of $C_2$ given by Lemma~\ref{def_ddual},
one constructs codewords in $C_1*C_2$ that form the rows of a matrix of the form
\beq
\left(\begin{array}{c|c}
I_k & Z
\end{array}\right)
\eeq
with $k=\min(n_1,\,k_1+d_2^\perp-2)$.
See~\cite[Lemma~6]{Singleton}, for more details.
\end{proof}

For example, if $C_1,C_2\subset\F^n$ have full support, and $C_2$ is MDS of dimension $k_2$ so $d_2^\perp=k_2+1$, we find
\beq
\dim(C_1*C_2)\geq\min(n,\,k_1+k_2-1).
\eeq

\entry
\label{reg_et_k+ddual}
In another direction, an easy induction on Proposition~\ref{dim_et_ddual}
also shows that if $C\subset\F^n$ is a linear code with full support and no repeated column,
then for any integers $t_0\geq0$, $a\geq 1$, and $j\geq 0$, we have
\beq
\dim(C\deux[t_0+aj])\geq\min(n,\,k_0+(d_a^\perp-2)j)
\eeq
where $k_0=\dim(C\deux[t_0])$ and $d_a^\perp=\ddual(C\deux[a])$.
As a consequence $C$ has regularity
\beq
r(C)\leq t_0+a\left\lceil\frac{n-k_0}{d_a^\perp-2}\right\rceil
\eeq
and for $t\geq r(C)$ we have $C\deux[t]=\F^n$.
We retrieve Corollary~\ref{reg_et_ddual} by setting $t_0=0$, $k_0=1$, $a=1$.

\begin{corollary}
\label{poids*_et_ddual}
Let $C_1,C_2\subset\F^n$ be two linear codes. Suppose $C_2$ has full support, i.e. dual distance $d_2^\perp\geq2$.
Fix an integer $i$ in the interval $1\leq i\leq\dim(C_1)$,
and set
\beq
m=\min(w_i(C_1)-i,\,d_2^\perp-2)\geq0
\eeq
where $w_i(C_1)$ is the $i$-th generalized Hamming weight of $C_1$.
Then for all $j$ in the interval $1\leq j\leq i+m$ we have
\beq
w_j(C_1*C_2)\leq w_i(C_1)-i-m+j.
\eeq
In particular (for $C_1$ nonzero) setting $i=j=1$ we find
\beq
\dmin(C_1*C_2)\leq\max(1,\,d_1-d_2^\perp+2)
\eeq
where $d_1=w_1(C_1)=\dmin(C_1)$.
\end{corollary}
\begin{proof}
Since $w_j(C_1*C_2)\leq w_{i+m}(C_1*C_2)-i-m+j$ (proof: shortening), it suffices to show $w_{i+m}(C_1*C_2)\leq w_i(C_1)$.
But then, just take $C'\subset C_1$ with support size $w_i(C_1)$ and dimension $i$, and observe that $C'*C_2\subset C_1*C_2$
has support size $w_i(C_1)$ and dimension at least $i+m$
by Proposition~\ref{dim_et_ddual}.
\end{proof}

\entry
\label{poidsdual*_et_ddual}
The same works for the dual weights of a product, improving on~\ref{ddual*_et_ddual}:
\begin{corollary*}
Let $C_1,C_2\subset\F^n$ be two linear codes with full support.
Then for all $i$ in the interval $1\leq i\leq n-\dim(C_1*C_2)$
we have
\beq
w_i((C_1*C_2)^\perp)\geq w_{i+d_2^\perp-2}(C_1^\perp).
\eeq
\end{corollary*}
\begin{proof}
By Corollary~\ref{cor_adj} we have $(C_1*C_2)^\perp*C_2\subset C_1^\perp$, so for all $j$,
\beq
w_j(C_1^\perp)\leq w_j((C_1*C_2)^\perp)*C_2).
\eeq
Set $m=\min(w_i((C_1*C_2)^\perp)-i,\,d_2^\perp-2)\geq0$.
Then for $1\leq j\leq i+m$ we can apply Corollary~\ref{poids*_et_ddual} with $C_1$ replaced by $(C_1*C_2)^\perp$, to get
\beq
\begin{split}
w_j((C_1*C_2)^\perp*C_2)&\leq w_i((C_1*C_2)^\perp)-i-m+j\\
&=\max(j,\,w_i((C_1*C_2)^\perp)-i+j-d_2^\perp+2).
\end{split}
\eeq
We combine these two inequalities, and we note that $C_1$ having full support implies $w_j(C_1^\perp)\geq j+1$, so the only possibility left in the $\max$ is
\beq
w_j(C_1^\perp)\leq w_i((C_1*C_2)^\perp)-i+j-d_2^\perp+2.
\eeq
Now for $j=i$ this gives $w_i((C_1*C_2)^\perp)\geq i+d_2^\perp-2$ so in fact $m=d_2^\perp-2$.
Then setting $j=i+d_2^\perp-2$ finishes the proof.
\end{proof}

\entry
The last assertion in Corollary~\ref{poids*_et_ddual}
extends to distances with rank constraints (as defined in~\ref{dist_rang}):
\begin{proposition*}
Let $C_1,C_2\subset\F^n$ be two linear codes of the same length $n$.
Suppose $C_2$ has full support, i.e. dual distance $d_2^\perp\geq2$.

Let $C_1$ be equipped with an arbitrary rank function, and $C_2$ with the trivial rank function.
Equip then $C_1*C_2$ with the product rank function.
Then for $1\leq i\leq\dim(C_1)$ we have
\beq
\dmin_{,i}(C_1*C_2)\leq\max(1,\dmin_{,i}(C_1)-d_2^\perp+2).
\eeq
\end{proposition*}
\begin{proof}
Let $x\in C_1$ with rank $\rk(x)\leq i$ and weight $w=\dmin_{,i}(C_1)$.
Choose $j\in\Supp(x)$. Then Lemma~\ref{def_ddual} gives $y\in C_2$ nonzero at $j$
but vanishing at $m=\min(w-1,d_2^\perp-2)$ other positions in $\Supp(x)$.
Then $z=x*y\in C_1*C_2$ has rank $\rk(z)\leq\rk(x)$ and weight $w-m$.
\end{proof}
In particular, if $C\subset\F^n$ is a linear code with dual distance $d^\perp\geq2$, then, with the convention of~\ref{rang_puissance_cano}, for all $t,t'\geq0$ we have
\beq
\dmin_{,i}(C\deux[t+t'])\leq\max(1,\,\dmin_{,i}(C\deux[t])-(d^\perp-2)t').
\eeq


\entry
\label{filtration}
Since the dual distance behaves nicely under projection, its use combines well with the following
filtration inequality:
\begin{lemma*}
Let $C,C'\subset\F^n$ be two linear codes. Suppose $C$ equipped with a filtration
\beq
0=C_0\subset C_1\subset\cdots\subset C_\ell=C
\eeq
by linear subcodes $C_i$. For $1\leq i\leq\ell$ set $T_i=\Supp(C_i)\moins\Supp(C_{i-1})$.
Then we have
\beq
\dim(C*C')\geq\sum_{i=1}^\ell \dim(\pi_{T_i}(C_i)*\pi_{T_i}(C')).
\eeq
In particular if $\ell=k=\dim(C)$ and $\dim(C_i)=i$ for all $i$, we have
\beq
\dim(C*C')\geq\sum_{i=1}^k \dim(\pi_{T_i}(C')).
\eeq
\end{lemma*}
\begin{proof}
We have a filtration
\beq
0=C_0*C'\subset C_1*C'\subset\cdots\subset C_\ell*C'=C*C'
\eeq
where for all $i\geq 1$ we have $C_{i-1}*C'\subset\ker\pi_{T_i}$, hence
\beq
\begin{split}
\dim(C_i*C')-\dim(C_{i-1}*C')&\geq\dim(C_i*C')-\dim((C_i*C')\cap\ker\pi_{T_i})\\
&=\dim(\pi_{T_i}(C_i*C')).
\end{split}
\eeq
Then we observe $\pi_{T_i}(C_i*C')=\pi_{T_i}(C_i)*\pi_{T_i}(C')$ and we sum over $i$.

In case $\ell=k=\dim(C)$ and $\dim(C_i)=i$ for all $i$, we can pick $c_i\in C_i\moins C_{i-1}$ and we have $\pi_{T_i}(C_i)=\langle v_i\rangle$,
where $v_i=\pi_{T_i}(c_i)\in\F^{T_i}$ has full support
(except if $T_i=\emptyset$, but then the contribution is $0$, which is fine).
The conclusion follows.
\end{proof}

From this, another bound involving the dual distance was established by D.~Mirandola:

\begin{theorem}[\cite{Mirandola}]
Let $q$ be a prime power. Fix an odd integer $D\geq3$.
Then for all $\epsilon>0$, there is an integer $N$ such that, for any integers $n,k$ such that $n-k\geq N$ and for any $[n,k]_q$ linear code $C$
with dual distance $d^\perp\geq D$ we have
\beq
\dim(C\deux)\geq k + \left(\frac{1}{2}-\epsilon\right)\frac{D-1}{2}\log_q^2(n-k).
\eeq
\end{theorem}

The proof uses two main ingredients.
The first is to transform the condition on $d^\perp$ into a lower bound on the terms $\dim(\pi_{T_i}(C))$ that appear in Lemma~\ref{filtration}
(applied with $C'=C$). For this one can use any of the classical bounds of coding theory, applied to $\pi_{T_i}(C)^\perp$.
For example, the Singleton bound gives $\dim(\pi_{T_i}(C))\geq\min(\abs{T_i},d^\perp-1)$, as already mentioned; Mirandola also uses the Hamming bound.
Then, in order to optimize the resulting estimates, one needs the filtration of $C$ to be constructed with some control on the $\abs{T_i}$; for this one uses the Plotkin bound.
However this leads to quite involved computations, and a careful analysis remains necessary in order to make all this work.

\entry
\label{Gale}
Some of the results above become especially interesting when seen from the geometric point of view.
Recall from Proposition~\ref{eq_def} that $r(C)$ is also equal to the Castelnuovo-Mumford regularity of the projective set of points $\Pi_C\subset\PP^{k-1}$ associated to $C$.
Then $\Pi_C$ admits syzygies
\beq
\begin{split}
0\longto\bigoplus_j\cO_{\PP^{k-1}}&(-a_{k-1,j})\longto\cdots\\
&\cdots\longto\bigoplus_j\cO_{\PP^{k-1}}(-a_{1,j}) \longto\cO_{\PP^{k-1}}\longto\cO_{\Pi_C}\longto 0
\end{split}
\eeq
for some integers $a_{i,j}\geq i$, and we have
\beq
r(C)=\max_{i,j}(a_{i,j}-i)
\eeq
(see \eg Chapter~4 of \cite{Eisenbud}).
Thus, from estimates such as the one in Corollary~\ref{reg_et_ddual} (or~\ref{reg_et_k+ddual}),
we see that important information on the syzygies of $\Pi_C$ can be extracted from the \emph{dual} code $C^\perp$.

This situation is very similar to that of~\cite{EP} although, admittedly, the results proved there are much deeper.
There, the duality of codes, seen from a geometric point of view, is called Gale duality, in reference to another context where it is of use.

\section{Pure bounds}
\label{sect_pure_bounds}

\vspace{.5\baselineskip}
Here we consider bounds on the basic parameters (dimension, distance) of a code and its powers, or of a family of codes and their product.
In contrast with section~\ref{sect_dual_distance}, no other auxiliary parameter (such as a dual distance) should appear.
Most of the material here will be taken from \cite{agis} and \cite{Singleton}.
The results are more involved and in some places we will only give partial proofs, but the reader can refer to the original papers for details.

\mysubsection{The generalized fundamental functions}
\entry
In many applications, a linear code is ``good'' when both its dimension and its minimum distance are ``large''.
In order to measure to what extent this is possible,
it is customary to consider the ``fundamental functions of linear block coding theory'' given by
\beq
a(n,d)=\max\{k\geq0\,;\;\exists C\subset\F^n,\,\dim(C)=k,\,\dmin(C)\geq d\}
\eeq
and
\beq
\alpha(\delta)=\limsup_{n\to\infty}\frac{a(n,\lfloor\delta n\rfloor)}{n}.
\eeq

Now suppose we need a code $C$ such that all powers $C,C\deux,\dots,C\deux[t]$,
up to a certain $t$, are good (see \ref{appl_multilin_algo}-\ref{OT+} for situations where this condition naturally appears).
Thanks to Theorem~\ref{monotone}, to give a lower bound on the dimension and the minimum distance of all these codes simultaneously,
it suffices to do so only for $\dim(C)$ and for $\dmin(C\deux[t])$.
This motivates the introduction of the generalized fundamental functions
\beq
a\deux[t](n,d)=\max\{k\geq0\,;\;\exists C\subset\F^n,\,\dim(C)=k,\,\dmin(C\deux[t])\geq d\}
\eeq
and
\beq
\alpha\deux[t](\delta)=\limsup_{n\to\infty}\frac{a\deux[t](n,\lfloor\delta n\rfloor)}{n},
\eeq
first defined in~\cite{agis}


If the base field $\F$ is not clear from the context, we will use more explicit notations such as $a_{\F}\deux[t]$ and $\alpha_{\F}\deux[t]$.
Also if $q$ is a prime power, we set $a_q\deux[t]=a_{\Fq}\deux[t]$
and $\alpha_q\deux[t]=\alpha_{\Fq}\deux[t]$, where $\Fq$ is the finite field
with $q$ elements.

\entry
\label{at_monotone}
From Theorem~\ref{monotone} we get at once:
\begin{lemma*}
Let $t\geq1$. Then for all $n,d$ we have
\beq
a\deux[t+1](n,d)\leq a\deux[t](n,d)
\eeq
and for all $\delta$ we have
\beq
\alpha\deux[t+1](\delta)\leq\alpha\deux[t](\delta).
\eeq
\end{lemma*}
As a consequence, any upper bound on the usual fundamental functions passes to the generalized functions.
However, improvements can be obtained by working directly on the latter.

Concerning lower bounds, we have the following:

\begin{proposition}
\label{ex_a}
Let $t\geq1$. Then for all $1\leq d\leq n$ we have
\beq
a\deux[t](n,d)\geq\left\lfloor\frac{n}{d}\right\rfloor.
\eeq
Moreover if $n\leq\abs{\F}+1$ we also have
\beq
a\deux[t](n,d)\geq\left\lfloor\frac{n-d}{t}\right\rfloor+1.
\eeq
\end{proposition}
\begin{proof}
For the first inequality, partition the set $[n]$ of coordinates
into $\left\lfloor\frac{n}{d}\right\rfloor$ subsets of size $d$ or $d+1$, and consider
the code $C$ spanned by their characteristic vectors. Then $C\deux[t]=C$ has dimension $\left\lfloor\frac{n}{d}\right\rfloor$
and minimum distance $d$.

For the second inequality, consider the (possibly extended) Reed-Solomon code,
obtained by evaluating polynomials of degree up to $\left\lfloor\frac{n-d}{t}\right\rfloor$
at $n$ given distinct elements of $\F$ (or possibly also at infinity). It has dimension $\left\lfloor\frac{n-d}{t}\right\rfloor+1$,
and its $t$-th power is also a Reed-Solomon code, obtained by evaluating polynomials of degree up to $t\left\lfloor\frac{n-d}{t}\right\rfloor\leq n-d$,
so of minimum distance at least $d$.
\end{proof}

It will be a consequence of the product Singleton bound that these inequalities are tight,
leading to the exact determination of the functions $a\deux[t]$ and $\alpha\deux[t]$ when $\F$ is infinite
(see Corollary~\ref{cor_Singleton}).

\entry
\label{tau_q_fini}
On the other hand, when $\F=\Fq$ is a finite field, the corresponding (generalized) fundamental functions $a_q\deux[t]$ and $\alpha_q\deux[t]$
are much more mysterious.

For example, note that the function $\alpha_q\deux[t]$ is nontrivial if and only if there is an asymptotically good family of linear codes over $\Fq$
whose $t$-th powers also form an asymptotically good family (and then also do all powers between $1$ and $t$).
We let
\beq
\tau(q)=\sup\{t\geq1\,;\;\exists\delta>0,\,\alpha_q\deux[t](\delta)>0\}\;\in\,\N\cup\{\infty\}
\eeq
be the supremum of the integers $t$ for which this holds, for a given $q$.

There is no $q$ for which it is known whether $\tau(q)$ is finite or infinite, although
algebraic-geometry codes will provide examples showing that
\beq
\tau(q)\to\infty
\eeq
as $q\to\infty$.

It is also true that $\tau(q)\geq2$ for all $q$, that is,
there exists an asymptotically good family of $q$-ary linear codes whose squares also form an asymptotically good family.
But as we will see, to include the case of small $q$ requires a quite intricate construction.

This leads to the author's favorite open problem on this topic: try to improve (any side of) the estimate
\beq
2\leq\tau(2)\leq\infty.
\eeq
That is, answer one of these two questions:
does there exist an asymptotically good family of binary linear codes whose cubes also form an asymptotically good family? or instead of cubes, is it possible with powers of some arbitrarily high given degree?

\mysubsection{An upper bound: Singleton}
\entry
\label{3p_Singleton}
The Singleton bound is one of the simplest upper bounds on the parameters of (possibly nonlinear) codes.
In the linear case, it states that for any $C\subset\F^n$ of dimension $k$ and minimum distance $d$, we have
\beq
k+d\leq n+1.
\eeq
At least three strategies of proof can be devised:
\begin{enumerate}[(i)]
\item
\label{Singleton_shortening}
\emph{Shortening.}
Shorten $C$ at any set of coordinates $I$ of size $\abs{I}=k-1$,
that is consider the subcode made of codewords vanishing at $I$. This subcode
has codimension at most $k-1$, since it is defined by the vanishing of $k-1$
linear forms, so it is nonzero. Hence $C$ contains a nonzero codeword $c$
supported in $[n]\moins I$, and $d\leq w(c)\leq n-k+1$.
\item
\label{Singleton_duality}
\emph{Duality.}
Let $H$ be a parity matrix for $C$, and let $C^\perp$ be the dual code.
Recall that codewords of $C$ are precisely linear relations between
columns of $H$. So $\dmin(C)\geq d$ means any $d-1$ columns of $H$
are linearly independent, hence $n-k=\dim(C^\perp)=\rk(H)\geq d-1$.
\item
\label{Singleton_puncturing}
\emph{Puncturing.}
Puncture $C$ at any set of coordinates $J$ of size $\abs{J}=n-k+1$.
By dimension, the corresponding projection $C\longto(\Fq)^{[n]\moins J}$ is not injective.
This means there are two codewords in $C$ that differ only over $J$, hence $d\leq \abs{J}\leq n-k+1$.
\end{enumerate}
Of course these methods are not entirely independent, since shortening is somehow the dual operation to puncturing.

Note that proofs \itemref{Singleton_shortening} and \itemref{Singleton_duality} work only in the linear case,
while proof \itemref{Singleton_puncturing} remains valid for general codes (suppose $\F=\Fq$ finite, set $k=\lfloor\log_q\abs{C}\rfloor$,
and use cardinality instead of dimension to show the projection noninjective).

Also a variant of proof \itemref{Singleton_puncturing} is to puncture at a set of coordinates of size $d-1$ (instead of $n-k+1$),
and conclude using injectivity of the projection (instead of noninjectivity).

\entry
Now let $C_1,\dots,C_t\subset\F^n$ be linear codes, and set $k_i=\dim(C_i)$ and $\widetilde{d}=\dmin(C_1*\cdots*C_t)$.
By the shortening argument of \itemref{Singleton_shortening} above, we see that for any choice of $I_i\subset[n]$ of size $\abs{I_i}=k_i-1$,
there is a nonzero codeword $c_i\in C_i$ supported in $[n]\moins I_i$.
If we could do so as the intersection of the supports of the $c_i$ be nonempty (and the $I_i$ be pairwise disjoint),
then $c_1*\cdots*c_t$ would be a nonzero codeword in $C_1*\cdots*C_t$ of weight at most $n-(k_1-1)-\cdots-(k_t-1)$.

Although this argument is incomplete, it makes plausible that, perhaps under a few additional hypotheses,
the Singleton bound should extend to products of codes essentially in the form of a linear inequality $k_1+\cdots+k_t+\widetilde{d}\leq n+t$.

A result of this sort has been proved in \cite{Singleton} and will be discussed in \ref{Singleton_general} below.
It turns out that the case $t=2$ is already of interest:

\begin{proposition}
\label{Singleton_t=2}
Let $C_1,C_2\subset\F^n$ be linear codes, with nondisjoint supports.
Set $k_i=\dim(C_i)$ and $\widetilde{d}=\dmin(C_1*C_2)$.
Then
\beq
\widetilde{d}\leq\max(1,\,n-k_1-k_2+2).
\eeq
\end{proposition}

It is interesting to try to prove this result with each of the three methods in~\ref{3p_Singleton}.
Quite surprisingly, the shortening approach \itemref{Singleton_shortening} does not appear to adapt easily (or at least, the author did not succeed).
On the other hand, the duality approach \itemref{Singleton_duality} and the puncturing approach \itemref{Singleton_puncturing} will give two very different proofs.

A first common step is to reduce to the case where $C_1$ and $C_2$ both have full support, by projecting on $I=\Supp(C_1)\cap\Supp(C_2)$.
Indeed, setting $\overline{C}_i=\pi_I(C_i)$, $\overline{k}_i=\dim(\overline{C}_i)$, $\overline{n}=\abs{I}$, and $J_i=\Supp(C_i)\moins I$,
we have $\dmin(\overline{C}_1*\overline{C}_2)=\widetilde{d}$, while $\overline{k}_i\geq k_i-\abs{J_i}$ so $\overline{n}-\overline{k}_1-\overline{k}_2+2\leq n-k_1-k_2+2$.
Hence the result holds for $C_1$ and $C_2$ as soon as it holds for $\overline{C}_1$ and $\overline{C}_2$.

So in the two proofs below we suppose that $C_1$ and $C_2$ both have full support. Also we set $d_i=\dmin(C_i)$.

\begin{proof}[First proof of Proposition~\ref{Singleton_t=2}]
We will reason by duality (the reader can check that our argument reduces to \ref{3p_Singleton}\itemref{Singleton_duality} when $C_1=\un$).

If $\widetilde{d}=1$ the proof is finished, so we can suppose $\widetilde{d}\geq 2$ and we have to show $\widetilde{d}\leq n-k_1-k_2+2$.

First, by Corollary \ref{cor_adj} we have
\beq
n-k_2=\dim(C_2^\perp)\geq\dim(C_1*(C_1*C_2)^\perp),
\eeq
while by Proposition~\ref{dim_et_ddual}
\beq
\dim(C_1*(C_1*C_2)^\perp)\geq\min(n,\,k_1+\widetilde{d}-2).
\eeq
Then Corollary \ref{poids*_et_ddual} gives $\widetilde{d}\leq d_1$, so $k_1+\widetilde{d}\leq n+1$ by the classical Singleton bound. Thus
\beq
\min(n,\,k_1+\widetilde{d}-2)=k_1+\widetilde{d}-2
\eeq
and we conclude.
\end{proof}

\begin{proof}[Second proof of Proposition~\ref{Singleton_t=2}]
We distinguish two cases:
\begin{itemize}
\item high dimension: suppose $k_1+k_2>n$, show $\widetilde{d}=1$
\item low dimension: suppose $k_1+k_2\leq n$, show $\widetilde{d}\leq n-k_1-k_2+2$.
\end{itemize}
To start with, we reduce the low dimension case to the high dimension case using a puncturing argument similar to \ref{3p_Singleton}\itemref{Singleton_puncturing}.

So suppose $k_1+k_2\leq n$, and puncture at any set of coordinates $J$ of size $\abs{J}=n-k_1-k_2+1$.
If one of the projections $C_i\longto\F^{[n]\moins J}$ is not injective, then $d_i\leq \abs{J}\leq n-k_1-k_2+2$, and the proof is finished since $\widetilde{d}\leq d_i$ by Corollary \ref{poids*_et_ddual}. On the other hand, if both projections are injective, set $\overline{C}_i=\pi_{[n]\moins J}(C_i)$ and $\overline{n}=k_1+k_2-1$.
Then we have $\dim(\overline{C}_i)=k_i$ with $k_1+k_2>\overline{n}$, while $\widetilde{d}\leq\dmin(\overline{C}_1*\overline{C}_2)+\abs{J}$.
Replacing $C_i$ with $\overline{C}_i$ we are now reduced to the high dimension case. It is treated in the following Lemma.
\end{proof}

\begin{lemma}
\label{NK}
Let $C_1,C_2\subset\F^n$ be linear codes, both with full support.
Set $k_i=\dim(C_i)$ and suppose $k_1+k_2>n$. Then there are codewords $c_1\in C_1$ and $c_2\in C_2$,
the product of which has weight
\beq
w(c_1*c_2)=1.
\eeq
In particular we have $\dmin(C_1*C_2)=\dmin_{,1}(C_1*C_2)=1$.
\end{lemma}
This Lemma was first proved by N.~Kashyap, as follows:
\begin{proof}
Let $H_i$ be a parity-check matrix for $C_i$.
It is enough to find a pair of disjoint subsets $A_1,A_2\subset[n]$ and a coordinate $j\in [n]\moins(A_1\cup A_2)$ such that, for each $i=1,2$:
\begin{itemize}
\item the columns of $H_i$ indexed by $A_i$ are linearly independent
\item the columns of $H_i$ indexed by $A_i\cup\{j\}$ are linearly dependent.
\end{itemize}
These can be found by a simple greedy algorithm:\\
\texttt{
\indent Initialize $A_1=A_2=\emptyset$\\
\indent FOR $j=1,\dots,n$\\
\indent\indent IF columns of $H_1$ indexed by $A_1\cup\{j\}$ are independent\\
\indent\indent THEN\\
\indent\indent\indent  append $j$ to $A_1$\\
\indent\indent ELSE\\
\indent\indent\indent IF columns of $H_2$ indexed by $A_2\cup\{j\}$ are independent\\
\indent\indent\indent THEN\\
\indent\indent\indent\indent  append $j$ to $A_2$\\
\indent\indent\indent ELSE\\
\indent\indent\indent\indent Output $A_1,A_2,j$ and STOP.\\
}
The stopping criterion must be met for some value of $j$, since $\rk(H_1)+\rk(H_2)=(n-k_1)+(n-k_2)<n$.
\end{proof}
It is remarkable that, in this case, we can show that the minimum distance of $C_1*C_2$ is attained by a codeword in product form
(compare with Example~\ref{dmin_non_produit}).

Now we state the product Singleton bound for general $t$:

\begin{theorem}
\label{Singleton_general}
Let $C_1,\dots,C_t\subset\F^n$ be linear codes; if $t\geq3$ suppose these codes all have \emph{full support}.
Then there are codewords $c_1\in C_1,\dots,c_t\in C_t$, the product of which has weight
\beq
1\leq w(c_1*\cdots*c_t)\leq\max(t-1,\,n-(k_1+\cdots+k_t)+t)
\eeq
where $k_i=\dim(C_i)$. As a consequence we have
\beq
\dmin(C_1*\cdots*C_t)\leq\dmin_{,1}(C_1*\cdots*C_t)\leq\max(t-1,\,n-(k_1+\cdots+k_t)+t).
\eeq
\end{theorem}
The proof is a direct generalization of the second proof of Proposition~\ref{Singleton_t=2} above.
The puncturing step is essentially the same, so the only difficulty is to find the proper generalization of Lemma~\ref{NK}.
We refer to \cite{Singleton} for the details.

By Proposition~\ref{ex_a} we see that this upper bound is tight.
Also it turns out that, for $t\geq3$, the condition that the codes have full support is necessary.
Actually, projecting on the intersection of the supports (as in the case $t=2$) allows to slightly relax this condition, but not to remove it entirely.
More details can be found in \cite{Singleton}.

\begin{corollary}
\label{cor_Singleton}
For any field $\F$ we have
\beq
a\deux[t](n,d)=\left\lfloor\frac{n}{d}\right\rfloor\qquad\text{for $\,1\leq d\leq t$,}
\eeq
\beq
a\deux[t](n,d)=\left\lfloor\frac{n-d}{t}\right\rfloor+1\qquad\text{for $\,t<d\leq n\leq\abs{\F}+1$,}
\eeq
and (in case $\F$ finite)
\beq
a\deux[t](n,d)\leq\left\lfloor\frac{n-d}{t}\right\rfloor+1\qquad\text{for $\,t<d\leq n$ with $n>\abs{\F}+1$.}
\eeq
Likewise, we have $\alpha\deux[t](0)=1$, and
\beq
\alpha\deux[t](\delta)\leq\frac{1-\delta}{t}\qquad\text{for $\,0<\delta\leq 1$}
\eeq
with equality when $\F$ is infinite.
\end{corollary}
\begin{proof}
Consequence of Theorem~\ref{Singleton_general} and Proposition~\ref{ex_a},
and of the inequality $a\deux[t](n,d)\leq a\deux[d](n,d)$ (Lemma~\ref{at_monotone}) in case $d\leq t$.
\end{proof}

It should be noted that the bound in Theorem~\ref{Singleton_general} holds not only for the minimal distance $\dmin$
of the product code, but also for the distance with rank constraint $\dmin_{,1}$.
As a consequence, the estimates in Corollary~\ref{cor_Singleton} hold in fact for the functions
$a\deux[t](n,d)_1=\max\{k\geq0\,;\;\exists C\subset\F^n,\,\dim(C)=k,\,\dmin_{,1}(C\deux[t])\geq d\}$
and
$\alpha\deux[t](\delta)_1=\limsup_{n\to\infty}\frac{a(n,\lfloor\delta n\rfloor)_1}{n}$.

Another interesting remark is that, for $t\geq2$, the function $\alpha\deux[t]$ is not continuous at $\delta=0$.
This might make wonder whether the definitions of the functions $a\deux[t]$ and $\alpha\deux[t]$ are the ``right'' ones.
For example one could ask how these functions are modified when one considers only codes with no repeated columns.

\mysubsection{Lower bounds for $q$ large: AG codes}
\entry
When $\F=\Fq$ is a finite field, the inequality $a_q\deux[t](n,d)\leq\left\lfloor\frac{n-d}{t}\right\rfloor+1$ in Corollary~\ref{cor_Singleton}
might be strict for $n>q+1$, because the length of Reed-Solomon codes is bounded.

In this setting, a classical way to get codes sharing most of the good properties of Reed-Solomon codes, but without this limitation on the length,
is to consider so-called algebraic-geometry codes constructed from curves of higher genus.

We recall from Example~\ref{exAG} that $C(D,G)$ is the code obtained by evaluating functions from
the Riemann-Roch space $L(D)$, associated with a divisor $D$,
at a set $G$ of points out of the support of $D$, on an algebraic curve
over $\F$.

\begin{proposition}
\label{prop_AG}
Let $q$ be a prime power and $t\geq1$ an integer. Suppose there is a curve $X$ of genus $g$ over $\Fq$
having at least $n$ rational points. Then we have
\beq
a_q\deux[t](n,d)\geq\left\lfloor\frac{n-d}{t}\right\rfloor+1-g\qquad\textrm{for $\,t<d\leq n-tg$.}
\eeq
\end{proposition}
\begin{proof}
Let $G$ be a set of $n$ rational points on $X$, and $D$ a divisor of degree $\deg(D)=\left\lfloor\frac{n-d}{t}\right\rfloor$ with support disjoint from $G$.
Then we have $g\leq\deg(D)\leq t\deg(D)<n$ so by the Goppa estimates
\beq
\dim(C(D,G))=l(D)\geq\deg(D)+1-g
\eeq
and
\beq
\dmin(C(tD,G))\geq n-t\deg(D)\geq d.
\eeq
The conclusion follows since $C(D,G)\deux[t]\subset C(tD,G)$ by Example~\ref{exAG}.
\end{proof}

\entry
\label{cor_AG}
We let $N_q(g)$ be the largest integer $n$ such that there is a curve of genus $g$ over $\Fq$ with $n$ rational points, and we define the Ihara constant
\beq
A(q)=\limsup_{g\to\infty}\frac{N_q(g)}{g}.
\eeq
\begin{corollary*}
For any prime power $q$ and for any integer $t\geq1$ we have
\beq
\alpha_q\deux[t](\delta)\geq\frac{1-\delta}{t}-\frac{1}{A(q)}\qquad\textrm{for $\,0<\delta\leq1-\frac{t}{A(q)}$.}
\eeq
As a consequence,
\beq
\tau(q)\geq\lceil A(q)\rceil-1.
\eeq
\end{corollary*}
\begin{proof}
Apply Proposition~\ref{prop_AG} with $g\to\infty$ and $n/g\to A(q)$.
\end{proof}

\entry
\label{A(q)}
For any prime power $q$ we have~\cite{DV}
\beq
A(q)\leq q^{1/2}-1,
\eeq
and in the other direction the following are known:
\begin{enumerate}[(i)]
\item
\label{A(q)_Serre}
There is a constant $c>0$ such that, for any prime power $q$, we have~\cite{Serre82} 
\beq
A(q)>c\log(q).
\eeq
\item
\label{A(q)_Ihara}
If $p$ is a prime, then for any integer $s\geq1$ we have~\cite{Ihara}
\beq
A(p^{2s})=p^s-1.
\eeq
\item
\label{A(q)_GSBB}
If $p$ is a prime, then for any integer $s\geq1$ we have~\cite{GSBB}
\beq
A(p^{2s+1})\geq \left(\frac{1}{2}\left(\frac{1}{p^s-1}+\frac{1}{p^{s+1}-1}\right)\right)^{-1}.
\eeq
\end{enumerate}
Then from Corollary~\ref{cor_AG} and from~\itemref{A(q)_Serre} we see that $\tau(q)\to\infty$ for $q\to\infty$, as claimed in~\ref{tau_q_fini}.
Moreover, we also see that the lower bound in Corollary~\ref{cor_AG} asymptotically matches the upper bound in Corollary~\ref{cor_Singleton}.

\mysubsection{Lower bounds for $q$ small: concatenation}
\entry
\label{debut_concat}
The lower bound in Corollary~\ref{cor_AG} is too weak to give some nontrivial information on $\tau(q)$ for $q$ small.
A useful tool in such a situation is \emph{concatenation}, which allows to construct codes over small alphabets from codes over large alphabets.
For technical reasons we present this notion in a more general context.

First, if $\cA_1,\dots,\cA_t$ and $\cB$ are sets, seen as ``alphabets'', and if
\beq
\Phi:\cA_1\times\cdots\times\cA_t\longto\cB
\eeq
is any map,
then for any integer $n$, applying $\Phi$ componentwise we get a map that (by a slight abuse of notation)
we also denote
\beq
\Phi:(\cA_1)^n\times\cdots\times(\cA_t)^n\longto\cB^n.
\eeq
From this point on, we can proceed as in~\ref{notation_point}: given subsets
$S_1\subset(\cA_1)^n,\dots,S_t\subset(\cA_t)^n$ we let
\beq
\dot{\Phi}(S_1,\dots,S_t)=\{ \Phi(c_1,\dots,c_t)\;;\;c_1\in S_1,\dots,c_t\in S_t\}\;\subset\,\cB^n.
\eeq
If moreover $\cA_1,\dots,\cA_t,\cB$ are $\F$-vector spaces
and $C_1\subset(\cA_1)^n,\dots,C_t\subset(\cA_t)^n$ are $\F$-linear subspaces, we let
\beq
\Phi(C_1,\dots,C_t)=\langle\dot{\Phi}(C_1,\dots,C_t)\rangle\;\subset\,\cB^n.
\eeq

For instance, if $t=2$, $\cA_1=\cA_2=\cB=\F$, and $\Phi$ is multiplication in $\F$,
we retrieve the definition of $C_1*C_2$ given in~\ref{def_*}.

\entry
Concatenation in the usual sense corresponds to $t=1$, $\cA=\F_{q^r}$, $\cB=(\Fq)^m$,
and
\beq
\phi:\F_{q^r}\inj(\Fq)^m
\eeq
an injective $\Fq$-linear map.
Using the natural identification $((\Fq)^m)^n=(\Fq)^{nm}$, we see that
if $C$ is a $[n,k]_{q^r}$ code, then $\phi(C)$ is a $[nm,kr]_q$ code.

In the classical terminology, $\phi(C)\subset(\Fq)^{nm}$ is called the concatenated code obtained from the external code $C\subset(\F_{q^r})^n$
and the internal code $C_\phi=\phi(\F_{q^r})\subset(\Fq)^m$, using the symbol mapping $\phi$.

It is easily seen that we have $\dmin(\phi(C))\geq\dmin(C)\dmin(C_\phi)$.
Now let $t\geq2$ be an integer. It turns out that if $\phi$ is well chosen,
then there is a $T\geq t$ such that $\dmin(\phi(C)\deux[t])$ can be estimated from $\dmin(C\deux[T])$ similarly.
To state this more precisely we need to introduce the following notations:

\entry
\label{notations_prop_concat}
Suppose we are given $h$ symmetric maps $\psi_1:(\F_{q^r})^t\longto\F_{q^r}$, $\dots$, $\psi_h:(\F_{q^r})^t\longto\F_{q^r}$ such that:
\begin{itemize}
\item viewed over $\F_{q^r}$, each $\psi_i$ is a polynomial of degree $D_i\leq T$
\item viewed over $\F_q$, each $\psi_i$ is $t$-multilinear.
\end{itemize}
In case $t>q$, suppose also these $\psi_i$ satisfy the Frobenius exchange condition~\ref{def_F_sym} in Appendix~\ref{sect_criterion}.

Then by Theorem~\ref{th_criterion}, the map
\beq
\Psi:(\psi_1,\dots,\psi_h):(\F_{q^r})^t\longto(\F_{q^r})^h
\eeq
admits a symmetric algorithm,
that is, one can find an integer $m$ and $\Fq$-linear maps $\phi:\F_{q^r}\longto(\Fq)^m$ and $\omega:(\Fq)^m\longto(\F_{q^r})^h$
such that the following diagram is commutative
\begin{center}
\begin{tikzpicture}
\node (NW) at (0,0) {$(\F_{q^r})^t$};
\node (NE) at (2.5,0) {$(\F_{q^r})^h$};
\node (SW) at (0,-1.2) {$((\Fq)^m)^t$};
\node (SE) at (2.5,-1.2) {$(\Fq)^m$};
\draw[->,font=\scriptsize,>=angle 90] (NW) edge node[above] {$\Psi$} (NE);
\draw[->,font=\scriptsize,>=angle 90] (SW) edge node[above] {$*$} (SE);
\draw[->,font=\scriptsize,>=angle 90] (NW) edge node[left] {$(\phi$,...,$\phi)$} (SW);
\draw[->,font=\scriptsize,>=angle 90] (SE) edge node[right] {$\omega$} (NE);
\end{tikzpicture}
\end{center}
where $*$ is componentwise multiplication in $(\Fq)^m$.

Now suppose both $\phi$ and $\omega$ are \emph{injective}.

\begin{proposition}
\label{prop_concat}
Let $C$ be a $[n,k]$ code over $\F_{q^r}$.
With the notations just above, suppose either:
\begin{itemize}
\item each polynomial $\psi_i$ is \emph{homogeneous} of degree $D_i$, or
\item $C$ contains the all-$1$ word $1_{[n]}$.
\end{itemize}
Then the $[nm,kr]$ code $\phi(C)$ over $\Fq$ satisfies
\beq
\dmin(\phi(C)\deux[t])\geq\dmin(C\deux[T]).
\eeq
\end{proposition}
\begin{proof}
This is a straightforward generalization of \cite[Prop.~8]{agis}.
Following the conventions in~\ref{debut_concat}, first viewing $\psi_i$ as a symmetric $t$-multilinear map,
it defines a $\Fq$-linear subspace $\psi_i(C,\dots,C)\subset(\F_{q^r})^n$.
Then viewing $\psi_i$ as a polynomial of degree $D_i$,
and using our hypothesis that either $\psi_i$ is homogeneous or $C$ contains $1_{[n]}$,
we deduce $\psi_i(C,\dots,C)\subset C\deux[D_i]$.
The commutative diagram in~\ref{notations_prop_concat} then translates into the diagram
\begin{center}
\begin{tikzpicture}
\node (NW) at (0,0) {$(C)^t$};
\node (NE) at (3.2,0) {$C\deux[D_1]\oplus\cdots\oplus C\deux[D_h]$};
\node (SW) at (0,-1.3) {$(\phi(C))^t$};
\node (SE) at (3.2,-1.3) {$\phi(C)\deux[t]$.};
\draw[->,font=\scriptsize,>=angle 90] (NW) edge node[above] {$\Psi$} (NE);
\draw[->,font=\scriptsize,>=angle 90] (SW) edge node[above] {$*$} (SE);
\draw[->,font=\scriptsize,>=angle 90] (NW) edge node[left] {$(\phi$,...,$\phi)$} (SW);
\draw[->,font=\scriptsize,>=angle 90] (SE) edge node[right] {$\omega$} (NE);
\end{tikzpicture}
\end{center}
Now let $z\in\phi(C)\deux[t]$ be a nonzero codeword of minimum weight
\beq
w(z)=\dmin(\phi(C)\deux[t]).
\eeq
From the diagram just above we can write $\omega(z)=(c_1,\dots,c_h)$ with $c_i\in C\deux[D_i]$,
and since $\omega$ is injective, there is at least one $i$ such that $c_i\neq0$.
On the other hand $\omega$ is defined blockwise on $(\Fq)^{nm}=((\Fq)^m)^n$, so
\beq
w(z)\geq w(c_i)\geq\dmin(C\deux[D_i]).
\eeq
We conclude since $D_i\leq T$ implies $\dmin(C\deux[D_i])\geq\dmin(C\deux[T])$ by Theorem~\ref{monotone}.
\end{proof}


\begin{theorem}
\label{bornes_concat}
Keep the notations above and suppose $t\leq q$.
Set $m=\binom{r+t-1}{t}$ and
\beq
T=q^{\left\lfloor\frac{(t-1)r}{t}\right\rfloor}+q^{\left\lfloor\frac{(t-2)r}{t}\right\rfloor}+\cdots+q^{\left\lfloor\frac{r}{t}\right\rfloor}+1.
\eeq
Then:
\begin{enumerate}[(i)]
\item
\label{bornes_concat_a}
We have
\beq
a_q\deux[t](nm,d)\geq ra_{q^r}\deux[T](n,d)
\eeq
for all $1\leq d\leq n$.
\item
\label{bornes_concat_alpha}
We have
\beq
\alpha_q\deux[t](\delta)\geq \frac{r}{m}\alpha_{q^r}\deux[T](m\delta)
\eeq
for all $0\leq\delta\leq1/m$.
\item
\label{bornes_concat_tau}
If $\tau(q^r)\geq T$, then $\tau(q)\geq t$.
\end{enumerate}
\end{theorem}
\begin{proof}
Since $t\leq q$, any symmetric $t$-multilinear map admits a symmetric algorithm by Theorem~\ref{th_criterion}.
Equivalently, for any $\Fq$-vector space $V$, the space of symmetric tensors $\Sym^t_{\Fq}(V)$ is spanned by elementary symmetric tensors.
We apply this with $V=(\F_{q^r})^\vee$, so we can find $m$ linear forms $\phi_1,\dots,\phi_m\in(\F_{q^r})^\vee$ such that $\phi_1^{\tens t},\dots,\phi_m^{\tens t}$
span $\Sym^t_{\Fq}((\F_{q^r})^\vee)=(S^t_{\Fq}\F_{q^r})^\vee$. Note that this implies that $\phi_1,\dots,\phi_m$ span $(\F_{q^r})^\vee$,
so $\phi=(\phi_1,\dots,\phi_m):\F_{q^r}\longto(\Fq)^m$ is injective, and also that $\Phi=(\phi_1^{\tens t},\dots,\phi_m^{\tens t}):(\F_{q^r})^t\longto(\Fq)^m$
induces an isomorphism $S^t_{\Fq}\F_{q^r}\simeq(\Fq)^m$ of $\Fq$-vector spaces.

However, by Theorem~\ref{th_poly_descr_StF} we also have an isomorphism $S^t_{\Fq}\F_{q^r}\simeq\prod_{I\in\cS}\F_{q^{r_I}}$ (for some $r_I|r$),
induced by a symmetric $t$-multilinear map $\Psi:(\F_{q^r})^t\longto\prod_{I\in\cS}\F_{q^{r_I}}$, so all this fits in a commutative diagram:
\begin{center}
\begin{tikzpicture}
\node (NW) at (0,0) {$(\F_{q^r})^t$};
\node (NE) at (3,0) {$\prod_{I\in\cS}\F_{q^{r_I}}$};
\node (SW) at (0,-1.2) {$((\Fq)^m)^t$};
\node (SE) at (3,-1.2) {$(\Fq)^m$.};
\node (undeplus) at (4.8,0) {$\subset(\F_{q^r})^{\abs{\cS}}$};
\draw[->,font=\scriptsize,>=angle 90] (NW) edge node[above] {$\Psi$} (NE);
\draw[->,font=\scriptsize,>=angle 90] (SW) edge node[above] {$*$} (SE);
\draw[->,font=\scriptsize,>=angle 90] (NW) edge node[left] {$(\phi$,...,$\phi)$} (SW);
\draw[->,font=\scriptsize,>=angle 90] (SE) edge node[right] {$\simeq$} (NE);
\draw[->,font=\scriptsize,>=angle 90] (NW) edge node[above] {$\Phi$} (SE);
\end{tikzpicture}
\end{center}
The components of $\Psi$ are \emph{homogeneous} polynomials $S_I$, which can be chosen of degree $D_I\leq T$ by Corollary~\ref{borne_degSI}.
Now we only have to apply Proposition~\ref{prop_concat} to get~\itemref{bornes_concat_a}, from which \itemref{bornes_concat_alpha} and \itemref{bornes_concat_tau} follow.
\end{proof}


\begin{example}
For $t=2$, the bounds in~\itemref{bornes_concat_a} and~\itemref{bornes_concat_alpha} give
\beq
a_q\deux(r(r+1)n/2,d)\geq ra_{q^r}\deux[q^{\lfloor r/2\rfloor}+1](n,d)
\eeq
and
\beq
\alpha_q\deux(\delta)\geq \frac{2}{r+1}\alpha_{q^r}\deux[q^{\lfloor r/2\rfloor}+1](r(r+1)\delta/2)
\eeq
for any prime power $q$ and for any $r$.

For $t=3$, they give
\beq
a_q\deux[3](r(r+1)(r+2)n/6,d)\geq ra_{q^r}\deux[q^{\lfloor 2r/3\rfloor}+q^{\lfloor r/3\rfloor}+1](n,d)
\eeq
and
\beq
\alpha_q\deux[3](\delta)\geq \frac{6}{(r+1)(r+2)}\alpha_{q^r}\deux[q^{\lfloor 2r/3\rfloor}+q^{\lfloor r/3\rfloor}+1](r(r+1)(r+2)\delta/6)
\eeq
for any $q\geq3$ and for any $r$. Note that the proof does not apply for $q=2$, since it requires $t\leq q$.
\end{example}

\entry
\label{borne_inf_alpha2}
Let $q$ be a prime power and let $t\leq q$ be an integer.
Then for any integer $r$, combining Corollary~\ref{cor_AG} with Theorem~\ref{bornes_concat}\itemref{bornes_concat_alpha}
we find
\beq
\alpha_q\deux[t](\delta)\geq\frac{r}{m}\left(\frac{1-m\delta}{T}-\frac{1}{A(q^r)}\right)
\eeq
where $m=\binom{r+t-1}{t}$ and $T=q^{\left\lfloor\frac{(t-1)r}{t}\right\rfloor}+\cdots+q^{\left\lfloor\frac{r}{t}\right\rfloor}+1$.

For this to be nontrivial we need $T<A(q^r)$. Since $A(q^r)\leq q^{r/2}-1$, this can happen only for $t=2$ and $r$ odd, and it turns out we have indeed $T<A(q^r)$ in this case:
this follows from \ref{A(q)}\itemref{A(q)_Ihara} if $q$ is a square, and from \ref{A(q)}\itemref{A(q)_GSBB} else.
As a consequence we see
\beq
\tau(q)\geq2
\eeq
for all $q$. In particular for $q=2$ and $r=9$ we find $m=45$, $T=17$, and $A(2^9)\geq\frac{465}{23}$, so~\cite{agis}
\beq
\alpha_2^{\langle 2\rangle}(\delta)\geq\frac{74}{39525}-\frac{9}{17}\,\delta\,\approx\, 0.001872-0.5294\,\delta.
\eeq




\section{Some applications}
\label{sect_appli}

\vspace{.5\baselineskip}
Product of codes being such a natural operation, it is no wonder it has already been used, since a long time,
implicitely or explicitely, in numerous applications.
Our aim here is to quickly survey the most significant of these applications, without entering too much into the historical details (for which the reader can refer to the literature),
but rather focusing on where the various bounds, structural results, and geometric interpretations presented in this text can be brought into play.

\mysubsection{Multilinear algorithms}
\entry
\label{def_multilin_algo}
Let $V_1,\dots,V_t$ and $W$ be finite-dimensional $\F$-vector spaces,
and let
\beq
\Phi:V_1\times\cdots\times V_t\longto W
\eeq
be a $t$-multilinear map.
A multilinear algorithm of length $n$ for $\Phi$ is a collection of $t+1$ linear
maps $\phi_1:V_1\longto\F^n$, $\dots$, $\phi_t:V_t\longto\F^n$ and $\omega:\F^n\longto W$,
such that the following diagram commutes:
\beq
\begin{CD}
V_1\times\cdots\times V_t @>{\Phi}>> W \\
@V{(\phi_1,\dots,\phi_t)}VV @AA{\omega}A \\
\F^n\times\cdots\times\F^n @>{*}>> \F^n
\end{CD}
\eeq
or equivalently, such that
\beq
\Phi(v_1,\dots,v_t)=\omega(\phi_1(v_1)*\cdots*\phi_t(v_t))
\eeq
for all $v_1\in V_1,\dots,v_t\in V_t$.

If we let $\epsilon_1,\dots,\epsilon_n$ be the canonical basis of $\F^n$
and $\pi_1,\dots,\pi_n$ the canonical projections $\F^n\to\F$,
then setting $w_j=\omega(\epsilon_j)\in W$ and $l_{i,j}=\pi_j\circ\phi_i\in V_i^\vee$, the last formula can also be written
\beq
\Phi(v_1,\dots,v_t)=\sum_{1\leq j\leq n}\left(\prod_{1\leq i\leq t}l_{i,j}(v_i)\right)w_j.
\eeq
Said otherwise, $\Phi$ can be viewed as a tensor in $V_1^\vee\tens\cdots\tens V_t^\vee\tens W$, and a multilinear algorithm
of length $n$ corresponds to a decomposition
\beq
\Phi=\sum_{1\leq j\leq n} l_{1,j}\tens\cdots\tens l_{t,j}\tens w_j
\eeq
as a sum of $n$ elementary tensors.
In turn, since elementary tensors are essentially the image of the Segre map in $V_1^\vee\tens\cdots\tens V_t^\vee\tens W$,
all this can be viewed geometrically in a way similar to \ref{Segre}.

\entry
\label{codes_multilin_algo}
More precisely, consider the linear codes
\beq
C_i=\phi_i(V_i)\;\subset\F^n
\eeq
and
\beq
C'=\omega^T(W^\vee)\;\subset(\F^n)^\vee=\F^n
\eeq
(where $\omega^T$ is the transpose of $\omega$). Let also $\overline{V_i}=V_i/\ker(\phi_i)$, and $\overline{W}=\im(\omega)$.
Note that $\Phi$ and the $\phi_i$ pass to the quotient, and as such they define a $t$-multilinear map $\overline{\Phi}:\overline{V_1}\times\cdots\times\overline{V_t}\longto\overline{W}$
as well as a multilinear algorithm of length $n$ for it.
So, after possibly replacing $\Phi$ with $\overline{\Phi}$, we can suppose the $\phi_i$ and $\omega^T$ are
injective, hence give identifications $C_i\simeq V_i$ and $C'\simeq W^\vee$.
Also we can suppose the $C_i$ and $C'$ all have full support (otherwise some coordinates are not ``used'' in the algorithm, and can be discarded).
Then $\Phi$ defines a point $P_\Phi$ in the projective space $\PP=\PP(C_1\tens\cdots\tens C_t\tens C')$, and the multilinear algorithm of length $n$ for $\Phi$
defines $n$ points in the Segre subvariety in $\PP$ whose linear span contains $P_\Phi$.
Now, untying all the definitions, we see these $n$ points are precisely those in the projective set of points $\Pi_{C_1*\cdots*C_t*C'}$
constructed in \ref{Segre}.

\entry
When $V_1=\cdots=V_t=V$ and $\Phi$ is a symmetric multilinear map, the algorithm is said symmetric if $\phi_1=\cdots=\phi_t=\phi$.
A symmetric algorithm of length $n$ for $\Phi$ corresponds to a decomposition of the associated tensor as a sum of $n$ elementary symmetric tensors
in $\Sym^t(V^\vee)\tens W$. In turn, the $\Sym^t(V^\vee)$ part of this tensor space is essentially the image of a Veronese map,
and links with the constructions in \ref{Veronese} could be given as above.

When $\F$ is finite it is not always true that a symmetric multilinear map admits a symmetric algorithm (counterexamples can be given as soon as $t>\abs{\F}$),
but Theorem~\ref{th_criterion} in Appendix~\ref{sect_criterion} provides a necessary and sufficient criterion for this to occur.

\entry
\label{appl_multilin_algo}
From a more concrete point of view, the multilinear algorithm for $\Phi$ can be interpreted as follows:
each $v_i\in V_i$ is splitted into $n$ local shares in $\Fq$ using the fixed map $\phi_i$, the shares are multiplied locally,
and the results are combined using $\omega$ to recover the final value $\Phi(v_1,\dots,v_t)$.
This is of interest in at least two contexts:
\begin{itemize}
\item
In algebraic complexity theory one is interested in having $n$ as small as possible,
in order to minimize the number of $t$-variable multiplications in $\Fq$ needed to compute $\Phi$.
This is relevant for applications in which the cost of a fixed linear operation is negligible compared to the cost of a $t$-variable multiplication.
We thus define
\beq
\mu(\Phi)
\eeq
the multilinear complexity of $\Phi$, as the smallest possible length of a multilinear algorithm for $\Phi$;
and when $\Phi$ is symmetric,
\beq
\mu^{\mathrm{sym}}(\Phi)
\eeq
the symmetric multilinear complexity of $\Phi$, as the smallest possible length of a symmetric multilinear algorithm for $\Phi$ (provided such an algorithm exists).
These are rank functions (in the sense of~\ref{def_rank}) on the corresponding spaces of multilinear maps.

There is a very broad literature on this subject, for various classes of multilinear maps $\Phi$.
We mention \cite{BD} for first pointing out the link between these questions and coding theory, and \cite{ChCh} for studying this link much further, in particular bringing AG codes into play. 
For more recent results and other points of view still close to the one presented here, we refer the reader to \cite{BR}\cite{BDEZ}\cite{CGLM}\cite{ChCh+}\cite{SL84}, 
and to the references therein for a more thorough historical coverage.
\item
One can also view this process as a very naive instance of multi-party computation, in which the local shares are given to $n$ remote users, who ``collectively'' compute $\Phi$.
Now this scheme has to be modified because it is too weak for most practical applications, in which it is customary to impose various security requirements.
For example, some shares could be altered by noise, or even by malicious users, to which the computation should remain robust.
Also, these malicious users should not be able to determine neither the entries nor the final value of the computation by putting their shares in common.
Of special importance is the case where $\Phi$ is multiplication in $\F$ (or in an extension field), since addition and multiplication are the basic gates in arithmetic circuits, that allow to represent arbitrary computable functions.

A more precise formalization of these problems, as well as some important initial constructions, can be found in \cite{BGW}\cite{CCD}\cite{CDM}.
For the more mathematically minded reader, especially if interested in the use of AG codes, a nice point of entry to the literature could be \cite{CC} and then \cite{CCX}.
\end{itemize}

Since the vectors in $\F^n$ involved in the computation are codewords in the $C_i$ or in their product $C_1*\cdots*C_t$, all the questions above
are linked to the possible parameters of these codes.
For example it is easily shown:

\begin{proposition}
\label{fin_multilin_algo}
If the given multilinear algorithm for $\Phi$ has length $n>\dim(C_1*\cdots*C_t)$, then one can puncture coordinates to deduce a shorter multilinear algorithm of length $\dim(C_1*\cdots*C_t)$.

Moreover, if the original algorithm is symmetric, then so is the punctured algorithm.
\end{proposition}
\begin{proof}
Let $S\subset[n]$ be an information set for $C_1*\cdots*C_t$, and
let $\sigma:\F^S\overset{\sim}{\longto}C_1*\cdots*C_t$ be the inverse of the natural projection $\pi_S$ (on $C_1*\cdots*C_t$).
Then $\pi_S\circ\phi_1,\dots,\pi_S\circ\phi_t$ and $\omega\circ\sigma$ define a multilinear algorithm of length $\abs{S}$ for $\Phi$.
\end{proof}

In this way $\mu(\Phi)$ (and if relevant, $\mu^{\mathrm{sym}}(\Phi)$) can be expressed as the dimension of a product code.

Likewise, in the multi-party computation scenario,
one is interested in constructing codes $C$ having high rate (for efficiency),
and such that $C\deux$ has high minimum distance (for resilience),
and also $C^\perp$ has high minimum distance (for privacy).
The reader can consult~\cite{CCX} for recent advances in this direction.

\mysubsection{Construction of lattices from codes}
\entry
Barnes and Sloane's Construction~D, first introduced in~\cite{BS}, and Forney's code formula from~\cite{Forney88-1}, are two closely related ways to construct lattices from binary linear codes.
Up to some details, they can be described as follows. 

Consider the lifting
\beq
\epsilon:\F_2\overset{\sim}{\longto}\{0,1\}\subset\Z.
\eeq
As in~\ref{debut_concat} we extend $\epsilon$ coordinatewise to $\epsilon:(\F_2)^n\longto\Z^n$.
Then given a chain of binary linear codes
\beq
\cC:\;C_0\subset C_1\subset\cdots\subset C_{a-1}\subset C_a=(\F_2)^n
\eeq
we construct a subset
\beq
\Lambda_{\cC}=\epsilon(C_0)+2\epsilon(C_1)+\cdots+2^{a-1}\epsilon(C_{a-1})+2^a\Z^n\;\subset\Z^n.
\eeq
It turns out that $\Lambda_{\cC}$ need not be a lattice in general, a fact that was sometimes overlooked in the literature.
One can show:

\begin{proposition}
With these notations, $\Lambda_{\cC}$ is a lattice if and only if the codes $C_i$ satisfy
\beq
C_i\deux\subset C_{i+1}.
\eeq
\end{proposition}
\begin{proof}
This follows from the relation
\beq
\epsilon(u)+\epsilon(v)=\epsilon(u+v)+2\epsilon(u*v)\;\in\Z^n
\eeq
which holds for any $u,v\in(\F_2)^n$.
\end{proof}

This observation was first made in Kositwattanarerk and Oggier's paper~\cite{KO}, where examples and a more careful analysis of the connection between the constructions in~\cite{BS} and~\cite{Forney88-1} are also given.
Roughly at the same time and independently, it was rediscovered by Boutros and Z\'emor, from whom the author learned it.

A feature of this Construction D is that, if $\Lambda_{\cC}$ is a lattice, then its parameters (volume, distance) can be estimated from those of the $C_i$.
One motivation for its introduction was to reformulate a construction of the Barnes-Wall lattices. In this case the $C_i$ are essentially Reed-Muller codes, and the condition $C_i\deux\subset C_{i+1}$ is satisfied, as noted in Example~\ref{exAG}.

\entry
What precedes can be generalized to codes over larger alphabets.
Let $p$ be a prime number, choose an arbitrary set of representatives $\cR$ for $\Z$ modulo $p$, and consider the lifting
\beq
\epsilon:\F_p\overset{\sim}{\longto}\cR\subset\Z.
\eeq
Then given a chain of binary linear codes
\beq
\cC:\;C_0\subset C_1\subset\cdots\subset C_{a-1}\subset C_a=(\F_p)^n
\eeq
we construct a subset
\beq
\Lambda_{\cC}=\epsilon(C_0)+p\epsilon(C_1)+\cdots+p^{a-1}\epsilon(C_{a-1})+p^a\Z^n\;\subset\Z^n.
\eeq
Again, $\Lambda_{\cC}$ need not be a lattice in general.
To give a criterion for this, one can introduce carry operations
\beq
\kappa_j:\F_p\times\F_p\longto\F_p
\eeq
for $1\leq j\leq a-1$, such that for each $x,y\in\F_p$ we have
\beq
\epsilon(x)+\epsilon(y)=\epsilon(x+y)+p\epsilon(\kappa_1(x,y))+\cdots+p^{a-1}\epsilon(\kappa_{a-1}(x,y))\mod p^a
\eeq
in $\Z$.

\begin{proposition}
With these notations, $\Lambda_\cC$ is a lattice if and only if for any $i,j$ we have
\beq
\kappa_j(C_i,C_i)\subset C_{i+j}
\eeq
(where $\kappa_j(C_i,C_i)$ is defined according to~\ref{debut_concat}).
\end{proposition}
\begin{proof}
Same as above.
\end{proof}

The usefulness of this criterion depends on our ability to control the $\kappa_j(C_i,C_i)$, which in turn depends on the choice of the lifting $\epsilon$, or equivalently, of the set of representatives $\cR$.

It turns out that there are at least two natural choices for this.

\entry
The first choice is to take $\cR=\{0,1,\dots,p-1\}$. We call the associated $\epsilon$ the ``naive'' lifting.
Then one has $\kappa_1=\kappa$ where
\beq
\kappa(x,y)=
\begin{cases}
1 & \textrm{if $x+y\geq p$}\\
0 & \textrm{else}
\end{cases}
\eeq
while
\beq
\kappa_j=0\quad\textrm{for $j>1$.}
\eeq

One drawback of this choice is that the expression for $\kappa$ does not appear to have much algebraic structure, except for the cocycle relation
\beq
\kappa(x,y)+\kappa(x+y,z)=\kappa(x,y+z)+\kappa(y,z).
\eeq
Anyway the fact that the higher $\kappa_j$ are $0$ should help in the computations.

\entry
Another choice is to take as $\cR$ a set of $p-1$-th roots of unity in $\Z$ modulo $p^a$, plus $0$.
We call the associated $\epsilon$ the multiplicative or Teichm\"uller lifting.
The $\kappa_j$ are then essentially given by the addition formulae for Witt vectors \cite[\S II.6]{SerreCL}\cite[Lect.~26]{Mumford},
more precisely this allows to express $\kappa_j$ as a symmetric homogeneous polynomial of degree $p^j$.
However since $\kappa_j$ is defined on $\F_p\times\F_p$, sometimes this expression can be simplified.

\begin{example*}
For $p=3$ we can take $\cR=\{0,1,-1\}$ as multiplicative representatives for any $a$.
The first carry operation is given by
\beq
\kappa_1(x,y)=-xy(x+y).
\eeq
Then an expression for $\kappa_2$ is $\kappa_2(x,y)=-xy(x+y)(x-y)^6$, however over $\F_3\times\F_3$ this is always $0$.
In fact the same holds for all the higher $\kappa_j$: as in the case $p=2$ we can take
\beq
\kappa_j=0\quad\textrm{for $j>1$.}
\eeq
\end{example*}

\vspace{.2\baselineskip}

\begin{corollary}
\label{cor_Teich}
When $\epsilon$ is the Teichm\"uller lifting,
a \emph{sufficient} condition for $\Lambda_\cC$ to be a lattice is that the codes $C_i$ satisfy
\beq
C_i\deux[p]\subset C_{i+1}.
\eeq
\end{corollary}
\begin{proof}
Indeed, since $\kappa_j$ can be expressed as a symmetric homogeneous polynomial of degree $p^j$,
we then have $\kappa_j(C_i,C_i)\subset C_i\deux[p^j]\subset C_{i+j}$.
\end{proof}

Natural candidates for a family of codes satisfying this condition is to take the $C_i$ evaluation codes, \eg (generalized) Reed-Muller codes, as in the binary case.
It remains to be investigated whether these lead to new examples of good lattices.
Also, note that we worked over $\Z$ for simplicity, but most of the discussion remains valid over the ring of integers of an algebraic number field,
possibly allowing further improvements. 
All this will be eventually considered in a forthcoming paper.

\mysubsection{Oblivious transfer}
\entry
\label{OT}
In an oblivious transfer (OT) protocol, Alice has two secrets $s_0,s_1\in\{0,1\}^n$, and Bob has a selection bit $b$. At the end of the protocol, Bob should get $s_b$ but no other information, and Alice should get no information on $b$.
In the Cr\'epeau-Kilian protocol, Alice and Bob achieve this through communication over a noisy channel \cite{CK}.
As a first step, emulating a coset coding scheme over a wiretap channel, they construct an almost-OT protocol,
in which Alice can cheat to learn $s$ with a certain positive probability.

Then, from this almost-OT protocol, they construct a true OT protocol.
More precisely, Alice chooses a $N\times n$ random matrix $A_0$ such that (in row vector convention)
\beq
1_{[N]}\cdot A_0=s_0,
\eeq
and she sets
\beq
A_1=A_0+(1_{[N]})^T\cdot (s_0+s_1).
\eeq
%
Bob has a selection bit $b$, and he chooses a random selection vector $v=(b_1,\dots,b_N)\in\{0,1\}^N$ such that
\beq
v\cdot (1_{[N]})^T=b.
\eeq
Using the almost-OT protocol $N$ times, for each $i$ he then learns the $i$-th row of $A_{b_i}$.
Putting theses rows together in a $N\times n$ matrix $B$, he finally finds
\beq
1_{[N]}\cdot B=v\cdot A_1+(1_{[N]}-v)\cdot A_0=s_b.
\eeq
One can then show that for $N\approx n^2$, Alice cannot cheat without Bob noticing it.

\entry
\label{OT+}
A slight drawback of this protocol is that the number~$N$ of channel uses grows quadratically in the size~$n$ of the secret,
so the overall communication rate tends to~$0$.
A first solution was proposed in~\cite{IKOPSW}. Their construction combines several sub-protocols, one of which is based on the results of~\cite{CC} already mentioned
in the discussion about multilinear algorithms and multi-party computation.

Another construction is proposed in~\cite{OZ}.
Quite interestingly it also makes an essential use of product of codes, while staying very close in spirit to the original Cr\'epeau-Kilian protocol.
The key idea is to replace the vector $1_{[N]}$ above, which is the generator matrix of a repetition code,
by the generator matrix $G$ of a code $C$ of fixed rate $R>0$ (so the secrets $s_0$ and $s_1$ also become matrices).
For Bob the reconstruction step is only slightly more complicated: if he is interested in $s_0$, he has to choose a random selection vector
\beq
b\in (C\deux)^\perp.
\eeq
One can then show that Alice cannot cheat as soon as $(C\deux)^\perp$ has dual distance at least $\delta n$ for some fixed $\delta>0$.
Note that the dual distance of $(C\deux)^\perp$ is just $\dmin(C\deux)$, and the linear span in the definition of $C\deux$ is relevant here
since the distance appears through a duality argument.

This raised the question of the existence of asymptotically good binary linear codes with asymptotically good squares, as discussed in section \ref{sect_pure_bounds} above,
to which a positive answer was finally given in \cite{agis}.

\mysubsection{Decoding algorithms}
\entry
\label{error-correcting_pairs}
There are several applications of products of codes to the decoding problem.
The first and probably the most famous of them is through the notion of error-correcting pairs~\cite{Kot}\cite{P92}.
If $C$ is a linear code of length $n$, a $t$-error-correcting pair for $C$ is a pair of codes $(A,B)$ of length $n$ such that:
\begin{enumerate}[(i)]
\item $A*B\;\perp\; C$
\item $\dim(A)>t$
\item $\dmin(A)>n-\dmin(C)$
\item $\ddual(B)>t$.
\end{enumerate}
(In \cite{P92} the product $A*B$ is defined without taking the linear span, but this is equivalent here since we are interested in the orthogonal).
Given such a pair, $C$ then admits a decoding algorithm that corrects $t$ errors with complexity $O(n^3)$.
This might be viewed as a way to reformulate most of the classical decoding algorithms for algebraic codes.

It is then natural to investigate the existence of error-correcting pairs for given $C$ and $t$. Results are known for many classes of codes,
and the study deeply involves properties of the $*$ product, such as those presented in this text.
For cyclic codes, a key result is the Roos bound, which is essentially the following:

\begin{proposition}[\cite{Roos}\cite{P96}]
\label{Roos_bound}
Let $A,B,C\subset\F^n$ be linear codes such that:
\begin{itemize}
\item $A*B\;\perp\; C$
\item $A$ has full support
\item $\dim(A)+\dmin(A)+\ddual(B)\geq n+3$.
\end{itemize}
Then
\beq
\dmin(C)\geq\dim(A)+\ddual(B)-1.
\eeq 
\end{proposition}
\begin{proof}
Since $C\subset(A*B)^\perp$, it suffices to show $\ddual(A*B)\geq\dim(A)+\ddual(B)-1$.
Then by Lemma~\ref{def_ddual}, setting
\beq
s=\dim(A)+\ddual(B)-2,
\eeq
it suffices
to show
\beq
\dim(\pi_S(A*B))=s
\eeq
for all $S\subset[n]$ of size $\abs{S}=s$.

Now for any such $S$ we have $\abs{[n]\moins S}=n+2-\dim(A)-\ddual(B)<\dmin(A)$,
so the projection $\pi_S:A\longto\pi_S(A)$ is injective, that is, $\dim(\pi_S(A))=\dim(A)$.
On the other hand, by Lemma~\ref{proprietes_ddual}\itemref{proprietes_ddual_projection}, we have $\ddual(\pi_S(B))\geq\ddual(B)$.
We can then conclude with Proposition~\ref{dim_et_ddual} applied to $\pi_S(A)$ and $\pi_S(B)$.
\end{proof}
Various improvements as well as generalizations of Proposition~\ref{Roos_bound} are given in~\cite{vLW}\cite{DP}, all of which can also be re-proved along these lines.

\entry
\label{power_decoding}
As another example of application of the product $*$ to the decoding problem,
we can cite the technique of so-called power syndrome decoding for Reed-Solomon codes \cite{SSB}.
If $c$ is codeword of a Reed-Solomon code $C$, then for any integer $i$, the (componentwise) power $c^i$ is also a codeword of a Reed-Solomon code of higher degree.
Now let $x$ be a received word, with error $e=x-c$. Then $x^i-c^i$ has support included in that of $e$.
Arranging the $x^i$ together, we thus get a virtual received word for an interleaved Reed-Solomon code, with a burst error.
Special decoding algorithms exist for this situation, and using them one eventually expects to get an improved decoding algorithm for the original $C$.
A more detailed analysis shows this works when $C$ has sufficiently low rate, allowing to decode it beyond half the minimum distance.

\mysubsection{Analysis of McEliece-type cryptosystems}
\entry
\label{appl_McEliece}
McEliece-type cryptosystems \cite{McEliece} rely on the fact that decoding a general linear code is a NP-hard problem \cite{BMT}.
First, Alice chooses a particular linear code $C$ with generator matrix $G$, for which an efficient decoding algorithm (up to a certain number $t$ of errors) is known,
and she sets $G'=SGP$ where $S$ is a randomly chosen invertible matrix and $P$ a random permutation matrix.
Her secret key is then the triple $(G,S,P)$, while her public key is essentially $G'$ (plus the number $t$).

Typically, $C$ is chosen among a class of codes with a strong algebraic structure, which is used for decoding.
However, multiplying $G$ by $S$ and $P$ allows to conceal this algebraic structure and make $G'$ look like the generator matrix of a general linear code $C'$.
A possible attack against such a scheme uses the fact that study of products $C_1*\cdots*C_t$ allows to find hidden algebraic relationships between subcodes $C_i$ of $C'$,
and ultimately, to uncover the algebraic structure of $C'$, from which a decoding algorithm could be designed.
This strategy was carried out successfully against certain variants of the McEliece cryptosystem, \eg when $C$ is a (generalized) Reed-Solomon code \cite{Wieschebrink}\cite{CGGOT}.

\entry
\label{McEliece_subfield}
However, the original McEliece cryptosystem remains unbroken.
There, $C$ is a binary Goppa code, which is constructed as a subfield subcode of an algebraic code defined over a larger field.
As D.~Augot pointed out to the author, one key difficulty in the analysis comes from the fact that it is not yet well understood how subfield subcodes behave under the $*$~product.

\appendix
\refstepcounter{section}
\section*{Appendix \thesection: A criterion for symmetric tensor decomposition}
\label{sect_criterion}

\mysubsection{Frobenius symmetric maps}

\entry
First we recall definitions from~\ref{def_multilin_algo}.
Let $\F$ be a field, let $V,W$ be finite-dimensional $\F$-vector spaces, and let
\beq
\Phi:V^t\longto W
\eeq
be a symmetric $t$-multilinear map.
A symmetric multilinear algorithm of length $n$ for $\Phi$ is a pair of linear maps $\phi:V\longto\F^n$ and $\omega:\F^n\longto W$,
such that the following diagram commutes:
\beq
\begin{CD}
V^t @>{\Phi}>> W \\
@V{(\phi,\dots,\phi)}VV @AA{\omega}A \\
(\F^n)^t @>{*}>> \F^n
\end{CD}
\eeq
or equivalently, such that
\beq
\Phi(v_1,\dots,v_t)=\omega(\phi_1(v_1)*\cdots*\phi_t(v_t))
\eeq
for all $v_1,\dots,v_t\in V$.

In turn, using the natural identification $\Sym^t(V;W)=\Sym^t(V^\vee)\tens W\subset(V^\vee)^{\tens t}\tens W$
to view $\Phi$ as an element in this tensor space, this corresponds to a decomposition
\beq
\Phi=\sum_{1\leq j\leq n} l_j^{\tens t}\tens w_j
\eeq
as a sum of $n$ elementary symmetric tensors.
Here, ``symmetry'' refers to the action of $\mathfrak{S}_t$ by permutation on the $t$ copies of $V^\vee$ in $(V^\vee)^{\tens t}\tens W$,
and we call elementary the symmetric tensors of the form $l^{\tens t}\tens w$ for $l\in V^\vee$, $w\in W$.

To show the equivalence, write $w_j=\omega(\epsilon_j)\in W$ and $l_j=\pi_j\circ\phi\in V^\vee$,
where $\epsilon_1,\dots,\epsilon_n$ and $\pi_1,\dots,\pi_n$ are the canonical bases of $\F^n$ and $(\F^n)^\vee$ respectively.

\entry
When $\F=\Fq$ is a finite field,
it turns out that not all symmetric multilinear maps admit a symmetric algorithm, or equivalently, not all symmetric tensors can be decomposed as a sum of elementary symmetric tensors.
There are at least two ways to see that.

The first is by a dimension argument: setting $r=\dim(V)$ and $s=\dim(W)$, we have
\beq
\dim\Sym^t(V;W)=\binom{r+t-1}{t}s
\eeq
which goes to infinity as $t\to\infty$, while
\beq
\dim\langle l^{\tens t}\,;\;l\in V^\vee\rangle\tens W\leq\frac{q^r-1}{q-1}s
\eeq
remains bounded. So, for $t$ big enough, $\langle l^{\tens t}\,;\;l\in V^\vee\rangle\tens W$
cannot be all of $\Sym^t(V;W)$, as claimed. However this proof is nonconstructive, because given an element in $\Sym^t(V;W)$,
it does not provide a practical way to check whether or not this element admits a symmetric algorithm.

The second is to show that the $l^{\tens t}$ all satisfy certain algebraic identities. So, by linearity,
if an element in $\Sym^t(V;W)$ does not satisfy these identities, then it can not admit a symmetric algorithm.
One such identity will come from the Frobenius property, $x^q=x$ for all $x\in\Fq$.
At first sight this gives a necessary condition for the existence of a symmetric algorithm.
However, by elaborating on this Frobenius property, we will show how to turn it into a necessary and sufficient
condition.

\entry
Symmetric tensor decomposition in characteristic $0$ has been extensively studied; see \eg \cite{CGLM} for a survey of recent results.
Over finite fields, perhaps the earliest appearance of the notion of symmetric bilinear algorithm was in~\cite{SL84}.

In this context, an important problem is the determination of
\beq
\mu_q^{\mathrm{sym}}(k)
\eeq
the symmetric bilinear complexity of $\F_{q^k}$ over $\Fq$, defined as the smallest possible length of a symmetric bilinear algorithm for the multiplication map $\F_{q^k}\times\F_{q^k}\longto\F_{q^k}$ (seen as a $\Fq$-bilinear map).
One can also consider
\beq
\mu_q(k)
\eeq
the (general) bilinear complexity of $\F_{q^k}$ over $\Fq$, defined similarly but without the symmetry condition.
A survey of results up to 2005 can be found in~\cite{BR}.
Quite strangely, although most authors gave constructions of symmetric algorithms, they only stated their results for the weaker complexity $\mu_q(k)$
(this is perhaps because another central topic in algebraic complexity theory is that of matrix multiplication, which is noncommutative).

At some point the development of the theory faced problems due to the fact that construction of symmetric algorithms from interpolation on a curve required a careful analysis of the $2$-torsion class group of the curve~\cite{CCX},
that was previously overlooked; in particular, some of the bounds cited in~\cite{BR} (especially the one from~\cite{STV})
had to be revised.
Finally the situation was clarified by the author in~\cite{ChCh+}; among the contributions of this work we can cite:
\begin{itemize}
\item emphasis is put on the distinction between symmetric and general bilinear complexity, rediscovering some results of~\cite{SL84}
\item it is shown that nonsymmetric bilinear algorithms are much easier to construct, since the $2$-torsion obstruction does not apply to them
\item for symmetric algorithms, a new construction is given (the idea of which originates from~\cite{21sep}) that bypasses the $2$-torsion obstruction,
repairing most of the broken bounds except perhaps for very small $q$.\footnote{For very small $q$, the effect of $2$-torsion on the bounds cannot be entirely discarded, but the methods of~\cite{CCX} allow to make it smaller, leading to the best results in this case.}
\end{itemize}

However, although this work settled a certain number of problems for bilinear algorithms, in particular showing that symmetric bilinear complexity is always well-defined,
at the same time it raised the question of the existence of symmetric algorithms for symmetric $t$-multilinear maps, $t\geq3$.

\begin{prop-def}
\label{def_F_red}
Let $V,W$ be $\Fq$-vector spaces and
\beq
f:V^t\longto W
\eeq
a symmetric $t$-multilinear map, for an integer $t\geq q$.
Then, for $1\leq i\leq t-q+1$, the map
\beq
\widetilde{f}^{(i)}:V^{t-q+1}\longto W
\eeq
defined by
\beq
\widetilde{f}^{(i)}(v_1,\dots,v_{t-q+1})=f(v_1,\dots,v_{i-1},\underset{\textrm{$q$ times}}{\underbrace{v_i,\dots,v_i}},v_{i+1},\dots,v_{t-q+1})
\eeq
for $v_1,\dots,v_{t-q+1}\in V$, is $t-q+1$-multilinear.

We call this $\widetilde{f}^{(i)}$ the $i$-th \emph{Frobenius reduced} of $f$.
\end{prop-def}
\begin{proof}
This follows from these two facts: (1) all elements $\lambda\in\Fq$ satisfy $\lambda^q=\lambda$, and
(2) the binomial coefficients $\binom{q}{j}$ are zero in $\Fq$ for $0<j<q$.
\end{proof}

When $i$ is not specified, we let $\widetilde{f}=\widetilde{f}^{(1)}$ be ``the'' Frobenius
reduced of $f$.

\begin{definition}
\label{def_F_sym}
Let $V,W$ be $\Fq$-vector spaces and
\beq
f:V^t\longto W
\eeq
a symmetric $t$-multilinear map.
We say that $f$ satisfies the \emph{Frobenius exchange condition},
or that $f$ is \emph{Frobenius symmetric} if, either:
\begin{itemize}
\item $t\leq q$, or
\item $t\geq q+1$ and $\widetilde{f}^{(1)}=\widetilde{f}^{(2)}$, that is,
\beq
f(\underset{\textrm{$q$ times}}{\underbrace{u,\dots,u}},v,z_1,\dots,z_{t-q-1})=f(u,\underset{\textrm{$q$ times}}{\underbrace{v,\dots,v}},z_1,\dots,z_{t-q-1})
\eeq
for all $u,v,z_1,\dots,z_{t-q-1}\in V$.
\end{itemize}
\end{definition}

We let $\Sym^t_{Frob}(V;W)\subset\Sym^t(V;W)$ be the subspace of Frobenius symmetric $t$-multilinear maps from $V$ to $W$ 
(in case the base field is not clear from the context, we will use a more precise notation such
as $\Sym^t_{\Fq,Frob}(V;W)$).

\entry
If $f:V^t\longto W$ is a symmetric $t$-multilinear map, for $t\geq q+1$, then in general its (first) Frobenius reduced $\widetilde{f}$
need not be symmetric. However:
\begin{proposition*}
Let $f\in\Sym^t(V;W)$ be a symmetric $t$-multilinear map, for an integer $t\geq q+1$. Then these two conditions are equivalent:
\begin{enumerate}[(i)]
\item $f$ satisfies the Frobenius exchange condition, that is, it is Frobenius symmetric;
\item its Frobenius reduced $\widetilde{f}$ is symmetric.
\end{enumerate}
Moreover, suppose these conditions are satisfied. Then:
\begin{itemize}
\item all the Frobenius reduced of $f$ are equal:
$\widetilde{f}^{(i)}=\widetilde{f}$ for $1\leq i\leq t-q+1$
\item $\widetilde{f}$ also satisfies the Frobenius exchange condition, that is, $\widetilde{f}$ also is Frobenius symmetric.
\end{itemize}
\end{proposition*}
\begin{proof}
Direct computation from the definitions.
\end{proof}

So, as a summary, for $t\leq q$ we have $\Sym^t_{Frob}(V;W)=\Sym^t(V;W)$ by definition, and for $t\geq q+1$ we have
\beq
f\in\Sym^t_{Frob}(V;W)\quad\equivaut\quad\widetilde{f}\in\Sym^{t-q+1}_{Frob}(V;W).
\eeq

Now we can state our main result:

\begin{theorem}
\label{th_criterion}
Let $V,W$ be $\Fq$-vector spaces of finite dimension,
and $f\in\Sym^t(V;W)$ a symmetric multilinear map.
Then $f$ admits a symmetric multilinear algorithm if and only if $f$ is Frobenius symmetric.

In particular when $t\leq q$ (for example, when $t=2$) this holds automatically.
\end{theorem}
\begin{proof}
Choose vectors $w_i$ forming a basis of $W$, and let $w_i^\vee$ be the dual basis.
Then $f$ admits a symmetric multilinear algorithm if and only if all the $w_i^\vee\circ f$
admit a symmetric multilinear algorithm, and $f$ is Frobenius symmetric if and only if
all the $w_i^\vee\circ f$ are Frobenius symmetric.
As a consequence, it suffices to prove the Theorem when $W=\Fq$, that is, when $f\in\Sym^t(V^\vee)$
is a symmetric multilinear form.

By definition, a symmetric multilinear form $f\in\Sym^t(V^\vee)$ admits a symmetric multilinear
algorithm if and only if it belongs to the subspace spanned by elementary symmetric
multilinear forms $l^{\tens t}$, for $l\in V^\vee$.
Thus, setting $\Sym^t_{Frob}(V^\vee)=\Sym^t_{Frob}(V,\Fq)$, we have to show the equality
\beq
\Sym^t_{Frob}(V^\vee)=\langle l^{\tens t}\,;\;l\in V^\vee\rangle
\eeq
of subspaces of $\Sym^t(V^\vee)$.

For this, we proceed by duality.
Given a subspace $V\subset\Sym^t(V^\vee)$, we let $V^\perp\subset S^tV$ be its orthogonal,
under the natural duality between $\Sym^t(V^\vee)$ and the $t$-th symmetric power $S^tV$. 
So we have to show the equality
\beq
\Sym^t_{Frob}(V^\vee)^\perp=\langle l^{\tens t}\,;\;l\in V^\vee\rangle^\perp
\eeq
of subspaces of $S^tV$.

By definition~\ref{def_F_sym}, a symmetric multilinear form $f\in\Sym^t(V^\vee)$ is Frobenius symmetric
if and only if it is orthogonal to all elements of the
form $(u^qv-uv^q)z_1\cdots z_{t-q-1}\in S^tV$.
Using biduality, this means we have
\beq
\Sym^t_{Frob}(V^\vee)^\perp=\fJ_t
\eeq
where $\fJ\subset S^{\cdot}V$ is the homogeneous ideal spanned by the $(u^qv-uv^q)\in S^{q+1}V$, for $u,v\in V$.

Now let $X_1,\dots,X_n$ denote a basis of $V$, so formally we can identify the symmetric algebra
of $V$ with the polynomial algebra in the $X_i$, as graded $\Fq$-algebras:
\beq
S^{\cdot}V=\Fq[X_1,\dots,X_n].
\eeq
With this identification, $\fJ$ then becomes the homogeneous ideal generated
by the polynomials $X_i^qX_j-X_iX_j^q$, for $1\leq i,j\leq n$.

Also this choice of a basis $X_1,\dots,X_n$ for $V$ gives an identification $V^\vee=(\Fq)^n$:
a linear form $l\in V^\vee$ corresponds to the $n$-tuple $(x_1,\dots,x_n)\in(\Fq)^n$
where $x_i=l(X_i)$. With this identification, the value of $l^{\tens t}\in\Sym^t(V^\vee)$
at an element $P(X_1,\dots,X_n)\in S^tV$ is precisely $P(x_1,\dots,x_n)$.
Thus $\langle l^{\tens t}\,;\;l\in V^\vee\rangle^\perp$ is the space of homogeneous
polynomials of degree $t$ that vanish at all points in $(\Fq)^n$, or equivalently,
that vanish at all points in the projective space $\PP^{n-1}(\Fq)$.

The conclusion then follows from Lemma~\ref{ideal_tout} below.
\end{proof}

\begin{lemma}
\label{ideal_tout}
The homogenous ideal in $\Fq[X_1,\dots,X_n]$ of the finite projective algebraic set $\PP^{n-1}(\Fq)$
is the homogeneous ideal generated by the $X_i^qX_j-X_iX_j^q$, for $1\leq i,j\leq n$.
\end{lemma}
This is a well-known fact, and a nice exercise, from elementary algebraic geometry.
One possible proof is by induction on $n$, in which one writes $\PP^{n}(\Fq)=(\Fq)^n\cup\PP^{n-1}(\Fq)$
and one shows at the same time the affine variant,
that the set of polynomials in $\Fq[X_1,\dots,X_n]$ vanishing at all points
in $(\Fq)^n$ is the ideal generated by the $X_i^q-X_i$, for $1\leq i\leq n$.
(See also \cite{MR} for more.)

\begin{definition}
\label{Frobenius_symmetric_algebra}
The \emph{Frobenius symmetric algebra} of a $\Fq$-vector space $V$ is
\beq
\SFrob V=S^\cdot V/(u^qv-uv^q)_{u,v\in V}
\eeq
the homogeneous quotient algebra of the symmetric algebra of $V$ by its graded ideal generated by elements of the form $u^qv-uv^q$.
\end{definition}
Also, by the $t$-th \emph{Frobenius symmetric power} of $V$ we mean the $t$-th graded part $\SFrob[t]V$ of this quotient algebra.
It comes equipped with a canonical Frobenius symmetric $t$-multilinear map $V^t\longto \SFrob[t]V$,
with the universal property that, given another $\Fq$-vector space $W$, then any Frobenius symmetric $t$-multilinear map $V^t\longto W$ uniquely factorizes through it.
Thus we get a natural identification
\beq
\Sym^t_{Frob}(V,W)=(\SFrob[t]V)^\vee\tens W
\eeq
of $\Fq$-vector spaces, and in particular
\beq
\Sym^t_{Frob}(V^\vee)=(\SFrob[t]V)^\vee.
\eeq

\entry
\label{PRM}
If $V$ has dimension $n$, with basis $X_1,\dots,X_n$, we have an identification $\SFrob V=\Fq[X_1,\dots,X_n]/(X_i^qX_j-X_iX_j^q)_{i,j}$,
hence by Lemma~\ref{ideal_tout}
\beq
\SFrob[t]V\simeq PRM_q(t,n-1)
\eeq
where the \emph{projective Reed-Muller code} $PRM_q(t,n-1)\subset(\Fq)^{\frac{q^n-1}{q-1}}$ is defined as the image of the map
\beq
\Fq[X_1,\dots,X_n]_t\longto(\Fq)^{\frac{q^n-1}{q-1}}
\eeq 
that evaluates homogeneous polynomials of degree $t$ at (a set of representatives of) all points in $\PP^{n-1}(\Fq)$.
In particular we have
\beq
\dim\SFrob[t]V=\dim PRM_q(t,n-1),
\eeq
which is computed in \cite{Sorensen} (and also by another mehtod in \cite{MR}).

Last, note that we have
\beq
PRM_q(t,n-1)=PRM_q(1,n-1)\deux[t]
\eeq
where $PRM_q(1,n-1)$ is the $[\frac{q^n-1}{q-1},n,q^{n-1}]$ simplex code.
So we see that the sequence of dimensions $\dim\SFrob[t]V$ coincides with the Hilbert sequence of this simplex code, as defined in~\ref{def_sequences}
(and also, as an illustration of Proposition~\ref{eq_def}, with that of the projective algebraic set $\PP^{n-1}(\Fq)$).

\entry
\label{ex_m_m'}
That Frobenius symmetry is a necessary condition in Theorem~\ref{th_criterion} was already known,
and is in fact easy to show: indeed, an elementary symmetric multilinear map $l^{\tens t}\tens w$,
for $l\in V^\vee$, $w\in W$, is obviously Frobenius symmetric, so any linear combination of such
maps should also be. The author learned it from I.~Cascudo, who provided the example of the map
\beq
\begin{array}{cccc}
m: & \F_4\times\F_4\times\F_4 & \longto & \F_4 \\
& (x,y,z) & \mapsto & xyz
\end{array}
\eeq
which is trilinear symmetric over $\F_2$, but does not admit a symmetric trilinear algorithm, since one can
find $x,y\in\F_4$ with $x^2y\neq xy^2$.

So the new part in Theorem~\ref{th_criterion} is that Frobenius symmetry is also a sufficient condition.
For example the map
\beq
\begin{array}{cccc}
m': & \F_4\times\F_4\times\F_4 & \longto & \F_4 \\
& (x,y,z) & \mapsto & x^2yz+xy^2z+xyz^2
\end{array}
\eeq
is trilinear symmetric over $\F_2$, and it satisfies $m'(x,x,y)=m'(x,y,y)$ for all $x,y\in\F_4$,
so it has to admit a symmetric trilinear algorithm. And indeed, one can check
\beq
\begin{split}
m'(x,y,z)=\tr(x)\tr(y)\tr(z)+\alpha^2\tr(&\alpha x)\tr(\alpha y)\tr(\alpha z)\\
&+\alpha\tr(\alpha^2 x)\tr(\alpha^2 y)\tr(\alpha^2 z)
\end{split}
\eeq
where $\F_4=\F_2[\alpha]/(\alpha^2+\alpha+1)$ and $\tr$ is the trace from $\F_4$ to $\F_2$.

\entry
Theorem~\ref{th_criterion} allows us to retrieve (and perhaps give a slightly more conceptual proof of)
a result of N.~Bshouty \cite[Th.~5]{Bshouty}:
given a finite dimensional commutative $\Fq$-algebra $\cA$, then the $t$-wise multiplication map
\beq
\begin{array}{cccc}
m^t: & \cA^t & \longto & \cA \\
& (a_1,\dots,a_t) & \mapsto & a_1\cdots a_t
\end{array}
\eeq
admits a symmetric multilinear algorithm if and only if, either:
\begin{itemize}
\item $t\leq q$, or
\item $t\geq q+1$ and all elements $a\in\cA$ satisfy $a^q=a$.
\end{itemize}
Indeed, this is easily seen to be equivalent to $m^t$ satisfying the Frobenius exchange condition~\ref{def_F_sym}.

So, for $t\leq q$, a symmetric algorithm for $m^t$ exists for any algebra $\cA$.
On the other hand, for $t>q$, it turns out that Bshouty's condition is very restrictive.
An example of algebra in which $m^t$ obviously admits a symmetric multilinear algorithm is $(\Fq)^n$
(one can take the trivial algorithm of length $n$, which is moreover easily seen minimal).
We show that, in fact, it is the only possibility, leading to the following ``all-or-nothing'' result:
\begin{proposition*}
If all elements $a$ in a commutative $\Fq$-algebra $\cA$ of dimension $n$ satisfy $a^q=a$,
then we have an isomorphism $\cA\simeq(\Fq)^n$ of $\Fq$-algebras.

As a consequence, the $t$-wise multiplication map $m^t$ in a commutative $\Fq$-algebra $\cA$ of dimension $n$
admits a symmetric multilinear algorithm if and only if, either:
\begin{itemize}
\item $t\leq q$, or
\item $t\geq q+1$ and $\cA\simeq(\Fq)^n$.
\end{itemize}
\end{proposition*}
\begin{proof}
Let $\cA$ be a commutative $\Fq$-algebra of dimension $n$ in which all elements $a$ satisfy $a^q=a$.
Then $a^{q^e}=a$ for arbitrarily large $e$, implying that $\cA$ has no (nonzero) nilpotent element.
This means the finite commutative $\Fq$-algebra $\cA$ is \emph{reduced}, and as such, it can be written as a product
of extension fields
\beq
\cA=\prod_i\F_{q^{r_i}}.
\eeq
Now using the condition $a^q=a$ once again, we see all $r_i=1$.
\end{proof}
(This proof, shorter than the author's original one, was inspired by a remark from I.~Cascudo.)

\mysubsection{Trisymmetric and normalized multiplication algorithms}
\entry
We conclude with a last application of Theorem~\ref{th_criterion}.
Let $q$ be a prime power, and $k\geq1$ an integer. For any $a\in\F_{q^k}$ we define a linear form $t_a:\F_{q^k}\longto\Fq$
by setting $t_a(x)=\tr(ax)$, where $\tr$ is the trace from $\F_{q^k}$ to $\F_q$.
It is well known that the map $a\mapsto t_a$
induces an isomorphism of $\Fq$-vector spaces $\F_{q^k}\simeq(\F_{q^k})^\vee$.

Now let $m:\F_{q^k}\times\F_{q^k}\longto\F_{q^k}$ be the multiplication map in $\F_{q^k}$, viewed as a tensor $m\in\Sym^2_{\Fq}((\F_{q^k})^\vee)\tens_{\Fq}\F_{q^k}$.
Then by construction, $\mu_q^{\mathrm{sym}}(k)$ is the rank of $m$
(in the sense of~\ref{def_rank}) with respect to the set of elementary symmetric tensors, that is, tensors of the form
\beq
t_{a}^{\tens 2}\tens b
\eeq
for $a,b\in\F_{q^k}$. By definition, ``symmetry'' here refers only to the first two factors of the tensors.
However, one could try to strengthen this notion as follows:

\begin{definition*}
The \emph{elementary trisymmetric tensors} in $\Sym^2_{\Fq}((\F_{q^k})^\vee)\tens_{\Fq}\F_{q^k}$ are those of the form
\beq
t_{a}^{\tens 2}\tens a
\eeq
for $a\in\F_{q^k}$.
We then define
\beq
\mu_q^{\mathrm{tri}}(k)
\eeq
the \emph{trisymmetric bilinear complexity} of $\F_{q^k}$ over $\Fq$,
as the rank of $m$ with respect to these.

Equivalently, $\mu_q^{\mathrm{tri}}(k)$ is the smallest possible length of a \emph{trisymmetric bilinear algorithm} for multiplication in $\F_{q^k}$ over $\Fq$,
that is, of a decomposition of $m$ as a linear combination of elementary trisymmetric tensors.
(Conveniently, if no such algorithm exists, we set $\mu_q^{\mathrm{tri}}(k)=\infty$.)
\end{definition*}
An equivalent definition can already be found in \cite{SL84}.
Now we determine the values of $q$ and $k$ for which $\mu_q^{\textrm{tri}}(k)$ is finite:

\begin{proposition}
\label{existence_tri}
A trisymmetric bilinear algorithm for multiplication in $\F_{q^k}$ over $\Fq$ exists for all $q$ and $k$, except precisely for $q=2$, $k\geq3$. 
\end{proposition}
\begin{proof}
Applying the trace, we see that a trisymmetric bilinear multiplication algorithm of the form
\beq
m=\sum_{1\leq j\leq n} \lambda_j\,t_{a_j}^{\tens 2}\tens a_j\;\in\Sym^2_{\Fq}((\F_{q^k})^\vee)\tens_{\Fq}\F_{q^k}
\eeq
for $\lambda_j\in\Fq$, $a_j\in\F_{q^k}$, corresponds to a symmetric trilinear algorithm
\beq
T=\sum_{1\leq j\leq n} \lambda_j\,t_{a_j}^{\tens 3}\;\in\Sym^3_{\Fq}((\F_{q^k})^\vee)
\eeq
for the symmetric trilinear form
\beq
\begin{array}{cccc}
T: & \F_{q^k}\times\F_{q^k}\times\F_{q^k} & \longto & \Fq \\
& (x,y,z) & \mapsto & \tr(xyz).
\end{array}
\eeq
By Theorem~\ref{th_criterion}, a symmetric trilinear form always admits a symmetric trilinear algorithm when $q\geq3$.
Now, suppose $q=2$. Then $T$ satisfies the Frobenius exchange condition if only if $\tr(x^2y)=\tr(xy^2)$ for all $x,y\in\F_{2^k}$.
On the other hand we have $\tr(a)=\tr(a^2)$ for all $a$, so in particular $\tr(x^2y)=\tr(x^4y^2)$, hence $T$ satisfies the Frobenius exchange condition if only if
\beq
\tr(x^4y^2)=\tr(xy^2)
\eeq
for all $x,y\in\F_{2^k}$. But this means $t_{x^4}=t_x$, or equivalently $x^4=x$, for all $x$. We conclude since this holds if and only if $k=1$ or $2$.
\end{proof}

\entry
\label{mu2tri2}
In case $q=k=2$, one can check that the trisymmetric bilinear complexity of $\F_4$ over $\F_2$ is $3$, so
\beq
\mu_2(2)=\mu_2^{\mathrm{sym}}(2)=\mu_2^{\mathrm{tri}}(2)=3.
\eeq
Indeed, setting $\F_4=\F_2[\alpha]/(\alpha^2\!+\!\alpha\!+\!1)$, a symmetric trilinear algorithm for $T$ is given by the tensor decomposition
\beq
T=t_1^{\tens 3}+t_\alpha^{\tens 3}+t_{\alpha^2}^{\tens 3}
\eeq
in $\Sym^3_{\F_2}((\F_4)^\vee)$. In fact $T$ is the only element of rank $3$ in $\Sym^3_{\F_2}((\F_4)^\vee)$.

This formula could be compared with that for the symmetric algorithm for $m'$ in~\ref{ex_m_m'}, which can be rewritten
\beq
m'\,=\,t_1^{\tens 3}\tens 1\,+\,t_\alpha^{\tens 3}\tens\alpha^2\,+\,t_{\alpha^2}^{\tens 3}\tens\alpha
\eeq
in $\Sym^3_{\F_2}((\F_4)^\vee)\tens_{\F_2}\F_4$.

It also motivates the following:

\begin{definition}
\label{def_normalized}
A \emph{normalized} trisymmetric bilinear algorithm of length $n$ for multiplication in $\F_{q^k}$ over $\Fq$ is a decomposition
of the multiplication tensor $m$ as a \emph{sum} of $n$ elementary trisymmetric tensors
\beq
m=\sum_{1\leq j\leq n} t_{a_j}^{\tens 2}\tens a_j
\eeq
in $\Sym^2_{\Fq}((\F_{q^k})^\vee)\tens_{\Fq}\F_{q^k}$, or equivalently, of the trace trilinear form $T$ as a \emph{sum} of $n$ cubes
\beq
T=\sum_{1\leq j\leq n} t_{a_j}^{\tens 3}
\eeq
in $\Sym^3_{\Fq}((\F_{q^k})^\vee)$. We then define
\beq
\mu_q^{\mathrm{nrm}}(k)
\eeq
the \emph{normalized} trisymmetric bilinear complexity of $\F_{q^k}$ over $\Fq$, as the smallest possible length of such a normalized algorithm
(Conveniently, if no such algorithm exists, we set $\mu_q^{\mathrm{nrm}}(k)=\infty$.)
\end{definition}

The new restriction here is we require the decomposition to be a sum, not a mere linear combination
(and as a consequence, this time $\mu_q^{\mathrm{nrm}}(k)$ cannot be interpreted in terms of a rank function).
This is somehow reminiscent of the distinction between orthogonal and self-dual bases in a nondegenerate quadratic space.
In fact, suppose instead of the trace trilinear form $T(x,y,z)=\tr(xyz)$ of $\F_{q^k}$ over $\Fq$, we're interested in the much more manageable trace bilinear form $B(x,y)=\tr(xy)$.
Since $B$ is nondegenerate, it has symmetric complexity $k$.
Moreover, symmetric algorithms for $B$ correspond to orthogonal bases of $\F_{q^k}$ over $\Fq$, although not necessarily self-dual.
See \eg \cite{SL80}\cite{JMV} for more on this topic and related questions.

\entry
\label{inegalites_tri_nrm}
The various notions of bilinear complexity defined so far can be compared.
Obviously (or for the first three, as a consequence of Lemma~\ref{fonctorialite_rang}) we always have
\beq
\mu_q(k)\leq\mu_q^{\mathrm{sym}}(k)\leq\mu_q^{\mathrm{tri}}(k)\leq\mu_q^{\mathrm{nrm}}(k).
\eeq

In the other direction, by \cite[Th.~1]{SL84} or \cite[Lemma~1.6]{ChCh+} we have
\beq
\mu_q^{\mathrm{sym}}(k)\leq2\,\mu_q(k)\qquad\textrm{for $\car(\Fq)\neq2$,}
\eeq
and by \cite[Th.~2]{SL84}
\beq
\mu_q^{\mathrm{tri}}(k)\leq4\,\mu_q^{\mathrm{sym}}(k)\qquad\textrm{for $q\neq2$, $\car(\Fq)\neq3$.}
\eeq

Now we want an upper bound on $\mu_q^{\mathrm{nrm}}(k)$. This can be stated in a greater generality.
Let $V$ be a $\Fq$-vector space, and let $F\in\Sym^t(V^\vee)$ be a symmetric $t$-multilinear form
with a $t$-symmetric algorithm of length $n$
\beq
F=\sum_{1\leq j\leq n} \lambda_j\,l_j^{\tens t},
\eeq
for $\lambda_j\in\Fq$, $l_j\in V^\vee$.
Now suppose each $\lambda_j$ can be written as a sum of $g$ $t$-th powers in $\Fq$
\beq
\lambda_j=\xi_{j,1}^t+\cdots+\xi_{j,g}^t.
\eeq
Then we have
\beq
F=\sum_{1\leq j\leq n}((\xi_{j,1}l_j)^{\tens t}+\cdots+(\xi_{j,g}l_j)^{\tens t})
\eeq
so $F$ is a sum of $gn$ $t$-th powers in $\Sym^t(V^\vee)$.

In particular, if $g=g(t,q)$ is the smallest integer such that any element in $\Fq$ is a sum of $g$ $t$-th powers in $\Fq$ 
(if no such integer exists we set $g=\infty$), we find
\beq
\mu_q^{\mathrm{nrm}}(k)\leq g(3,q)\mu_q^{\mathrm{tri}}(k).
\eeq

Determination of $g(t,q)$ is an instance of Waring's problem
(note that determination of $\mu_q^{\mathrm{nrm}}(k)$, that is, of the shortest decomposition of $T$ as a sum of cubes in $\Sym^3_{\Fq}((\F_{q^k})^\vee)$, also is!).
For $t=3$ the answer is well known (see \eg \cite{Singh}):

\begin{lemma}
\label{Waring}
For $q\neq2,4,7$, we have
\beq
g(3,q)=2
\eeq
\ie any element in $\Fq$ is a sum of two cubes (and this is optimal).
The exceptions are: $g(3,2)=1$, $g(3,4)=\infty$, and $g(3,7)=3$.
\end{lemma}

In $\F_4=\F_2[\alpha]/(\alpha^2\!+\!\alpha\!+\!1)$, note that every nonzero $x$ satisfy $x^3=1$. As a consequence, neither $\alpha$ nor $\alpha^2$ can be written as a sum of cubes,
and $g(3,4)=\infty$ as asserted.

\begin{proposition}
A normalized trisymmetric multiplication algorithm for $\F_{q^k}$ over $\Fq$ exists for all $q$ and $k$, except precisely for $q=2$, $k\geq3$ and for $q=4$, $k\geq2$.
More precisely, we have:
\begin{enumerate}[(i),itemsep=1ex]
\item $\mu_q^{\mathrm{nrm}}(1)=1\quad$ for all $q$,
\item $\mu_2^{\mathrm{nrm}}(2)=3$,
\item $\mu_2^{\mathrm{nrm}}(k)=\infty\quad$ for $k\geq3$,
\item $\mu_4^{\mathrm{nrm}}(k)=\infty\quad$ for $k\geq2$,
\item $\mu_7^{\mathrm{nrm}}(k)\leq 3\,\mu_7^{\mathrm{tri}}(k)\quad$ for $k\geq2$, 
\item $\mu_q^{\mathrm{nrm}}(k)\leq 2\,\mu_q^{\mathrm{tri}}(k)\quad$ for $q\neq2,4,7$ and $k\geq2$.
\end{enumerate}
\end{proposition}
\begin{proof}
Item (i) is obvious, (ii) comes from \ref{mu2tri2}, and (iii)(v)(vi) from Proposition~\ref{existence_tri} joint with the discussion in~\ref{inegalites_tri_nrm} and Lemma~\ref{Waring}.

Now to prove (iv) we have to show that, for $k\geq2$, there is no normalized multiplication algorithm for $\F_{4^k}$ over $\F_4$.
We proceed by contradiction, so we suppose we have a decomposition
\beq
T=\sum_{1\leq j\leq n} t_{a_j}^{\tens 3}
\eeq
which implies, for any $x\in\F_{4^k}$,
\beq
\tr(x^3)=\sum_{1\leq j\leq n}\tr(a_jx)^3.
\eeq
Now let $\alpha\in\F_4$ with $\alpha^2=\alpha+1$. The trace function $\tr:\F_{4^k}\longto\F_4$ is surjective, so $\alpha=\tr(z)$ for some $z\in\F_{4^k}$.
Moreover, by Lemma~\ref{Waring}, we can write $z=x^3+y^3$ as a sum of two cubes in $\F_{4^k}$.
So we conclude that
\beq
\alpha=\tr(x^3+y^3)=\sum_{1\leq j\leq n}(\tr(a_jx)^3+\tr(a_jy)^3)
\eeq
is a sum of cubes in $\F_4$, in contradiction with Lemma~\ref{Waring} and the discussion following it.
\end{proof}

The more constraints we put on the structure of the algorithms, the smaller the set of such algorithms is, and hopefully the lower the complexity of search in this set should be.
This makes one wonder whether an adaptation of the methods of~\cite{BDEZ} could allow one to succesfully compute the exact values of $\mu_q^{\mathrm{tri}}(k)$ and $\mu_q^{\mathrm{nrm}}(k)$ for a not-so-small range of $k$,
and find the corresponding optimal algorithms.

\refstepcounter{section}
\section*{Appendix \thesection: On symmetric multilinearized polynomials}
\label{sect_multilinearized}

\entry
\label{actions_A}
Let $\cA$ be the algebra of all functions from $(\F_{q^r})^t$ to $\F_{q^r}$.
It is easily checked that any such function can be represented as a polynomial function, and moreover, since elements of $\F_{q^r}$ satisfy $x^{q^r}=x$,
we have a natural identification
\beq
\cA=\F_{q^r}[x_1,\dots,x_t]/(x_1^{q^r}-x_1,\dots,x_t^{q^r}-x_t).
\eeq
Often we will identify an element $f\in\cA$ with its (unique) representative of minimum degree in $\F_{q^r}[x_1,\dots,x_t]$. Likewise we identify $\Z/r\Z$ with $\{0,1,\dots,r-1\}$.

The symmetric group $\fS_t$ acts linearly on $\cA$, by permutation of the variables:
for $\sigma\in\fS_t$, $f\in\cA$, and $(u_1,\dots,u_t)\in(\F_{q^r})^t$,
we set
\beq
({}^\sigma\!f)(u_1,\dots,u_t)=f(u_{\sigma(1)},\dots,u_{\sigma(t)}).
\eeq
Also, the Frobenius $f\mapsto f^q$ defines an automorphism of $\cA$ over $\Fq$, of order $r$, hence an action of $\Z/r\Z$ on $\cA$,
where $j\in\Z/r\Z$ acts as $f\mapsto f^{q^j}$.

These two actions commute, so $\cA$ is equipped with an action of
\beq
G=\mathfrak{S}_t\times\Z/r\Z,
\eeq
the invariants of which are the symmetric functions on $(\F_{q^r})^t$ with values in $\Fq$.

Note that the action of $\mathfrak{S}_t$ is linear over $\F_{q^r}$, while the action of $\Z/r\Z$ (hence that of $G$)
is only linear over $\Fq$.

\begin{definition}
The $t$-multilinearized polynomials with coefficients in $\F_{q^r}$ over $\Fq$ are the polynomials of the form
\beq
\sum_{0\leq i_1,\dots,i_t\leq r-1} a_{i_1,\dots,i_r}x_1^{q^{i_1}}\cdots x_t^{q^{i_t}}
\eeq
with $a_{i_1,\dots,i_r}\in\F_{q^r}$.
\end{definition}

These $t$-multilinearized polynomials form a $\F_{q^r}$-linear subspace
\beq
\cB\subset\cA,
\eeq
of dimension $r^t$ over $\F_{q^r}$ (hence also of dimension $r^{t+1}$ over $\Fq$).

It is easily checked that $\cB$ coincides precisely with the space of functions from $(\F_{q^r})^t$ to $\F_{q^r}$ that are $t$-multilinear over $\Fq$.

For $t=1$, one retrieves the notion of linearized polynomials, which is an important tool in the theory of finite fields and in coding theory. 
For $t=2$, bilinearized polynomials have also been introduced to solve various problems in bilinear algebra,
as illustrated in~\cite{SL84} and~\cite{agis}.
Our aim here is to extend some of the results from~\cite{agis} to arbitrary $t$.

More precisely, in Theorem~\ref{th_poly_descr_StF} we construct a family of \emph{homogeneous symmetric} $t$-multilinearized polynomials
\beq
S_I:(\F_{q^r})^t\longto\F_{q^{r_I}}
\eeq
(the index $I$ ranges in a certain set $\cS$ and is essentially the multidegree of $S_I$)
taking values in an intermediate field $\F_{q^{r_I}}$, and satisfying the following \emph{universal} property:
for any $\Fq$-vector space $V$, and for any map
\beq
F:(\F_{q^r})^t\longto V
\eeq
symmetric $t$-multilinear over $\Fq$, there is a unique family of $\Fq$-linear maps
\beq
f_I: \F_{q^{r_I}}\longto V
\eeq
such that
\beq
F=\sum_{I\in\cS}f_I\circ S_I.
\eeq
Then in Theorem~\ref{borne_equidistr} (or more precisely in Corollary~\ref{borne_degSI}) we give an upper bound on the degrees of the $S_I$.

\mysubsection{Polynomial description of symmetric powers of an extension field}

\vspace{.5\baselineskip}

First we state (without proof) two easy results on group actions.

\begin{lemma}
\label{base_invariants}
Let $\Gamma$ be a finite group acting on a finite set $\cP$,
let $\cR\subset\cP$ be a set of representatives for the action, and for $I\in\cR$ let $o(I)\subset\cP$ be its orbit. 
Suppose also $\Gamma$ acts linearly on a vector space $V$ with basis $(b_I)_{I\in\cP}$,
such that
\beq
\gamma\cdot b_I=b_{\gamma I}
\eeq
for all $\gamma\in\Gamma$, $I\in\cP$.
Now for $I\in\cR$, set
\beq
s_I=\sum_{J\in o(I)}b_J.
\eeq
Then, the subspace of invariants $V^\Gamma$ admits the $(s_I)_{I\in\cR}$ as a basis.
\end{lemma}

\begin{lemma}
\label{action_produit}
Let the finite cyclic group $\Z/r\Z$ act on a finite set $\cR$, and let $\cS\subset\cR$ be a set of representatives for the action.
For each $I\in\cS$, let $r_I=\abs{o(I)}$ be the size of its orbit, so $r_I|r$ and $r_I\Z/r\Z\subset\Z/r\Z$ is the stabilizer subgroup of $I$.

Now set
\beq
\cP=\Z/r\Z\times\cR
\eeq
and let $\Z/r\Z$ act on this product set, on the first factor, by translation, and on the second factor, by the action we started with.
Then the action of $\Z/r\Z$ on $\cP$ is free, and it admits
\beq
\cT=\{(i,I)\,;\;0\leq i\leq r_I-1,\,I\in\cS\}\;\subset\cP
\eeq
as a set of representatives. Moreover we have
\beq
\abs{\cT}=\sum_{I\in\cS}r_I=\abs{\cR}.
\eeq
\end{lemma}

Note that there are other possible choices for a set of representatives, for instance a more obvious one would be $\{0\}\times\cR$.
However, $\cT$ is the choice that will make the proof of Theorem~\ref{th_poly_descr_StF} below work.

\entry
\label{action_St}
To each $I=(i_1,\dots,i_t)\in(\Z/r\Z)^t$ we associate a monomial 
\beq
M_I=x^{q^I}=x_1^{q^{i_1}}\cdots x_t^{q^{i_t}},
\eeq
of degree
\beq
D_I=q^{i_1}+\cdots+q^{i_t}
\eeq
so $t\leq D_I\leq tq^{r-1}$. These form a basis of $\cB$ over $\F_{q^r}$.

As we saw in \ref{actions_A}, the symmetric group $\mathfrak{S}_t$ acts on $\cA$, and we let it act also on $(\Z/r\Z)^t$ by permutation of coordinates.
The map $I\mapsto M_I$ is compatible with these actions, in that
\beq
{}^\sigma M_I=M_{\sigma(I)}
\eeq
for $\sigma\in\mathfrak{S}_t$, so $\cB$ is stable under $\mathfrak{S}_t$.

Let $\cR\subset(\Z/r\Z)^t$ be the set of nonincreasing $t$-tuples of elements of $\Z/r\Z\simeq\{0,1,\dots,r-1\}$.
It has cardinality
\beq
\abs{\cR}=\binom{r+t-1}{t}.
\eeq
Clearly $\cR$ is a set of representatives for the action of $\mathfrak{S}_t$ on $(\Z/r\Z)^t$, so we have a bijection
\beq
\begin{array}{ccc}
\cR & \overset{\sim}{\longto} & (\Z/r\Z)^t/\mathfrak{S}_t\\
I & \mapsto & o(I)
\end{array}
\eeq
where $o(I)\subset(\Z/r\Z)^t$ is the orbit of $I$ under $\mathfrak{S}_t$.
Now for $I\in\cR$ we set
\beq
S_I=\sum_{J\in o(I)}M_J,
\eeq
so $S_I$ is a symmetric homogeneous polynomial of degree $D_I$ in $t$ variables over $\F_{q^r}$, which is also symmetric $t$-multilinear over $\Fq$.
The number $\abs{o(I)}$ of monomials in $S_I$ is a divisor of $t!$ (it can be a strict divisor if $I$ has repeated elements).

Then, by Lemma~\ref{base_invariants}, the subspace of invariants
\beq
\cB^{\mathfrak{S}_t}=\Sym^t_{\Fq}(\F_{q^r};\F_{q^r}) 
\eeq
admits the $(S_I)_{I\in\cR}$ as a basis over $\F_{q^r}$.


\entry
\label{action_Frob}
As we saw in \ref{actions_A}, the cyclic group $\Z/r\Z$ acts on $\cA$ by Frobenius,
and we let it act also on $(\Z/r\Z)^t$ diagonally by translation, that is, we let $j\in\Z/r\Z$ act as $I=(i_1,\dots,i_t)\mapsto I+j=(i_1+j,\dots,i_t+j)$
where addition is modulo $r$.
The map $I\mapsto M_I$ is compatible with these actions, in that
\beq
(M_I)^{q^j}=M_{I+j}
\eeq
so $\cB$ is stable under $\Z/r\Z$.

This diagonal action of $\Z/r\Z$ on $(\Z/r\Z)^t$ commutes with that of $\mathfrak{S}_t$, so it defines an action of $\Z/r\Z$ on $\cR\simeq(\Z/r\Z)^t/\mathfrak{S}_t$,
which can be written as
\beq
\begin{array}{ccc}
\Z/r\Z\times\cR & \longto & \cR\\
(j,I) & \mapsto & I\boxplus j 
\end{array}
\eeq
where $I\boxplus j$ is the representative of $I+j$ in $\cR$.
More precisely, $I+j$ need not be nonincreasing since addition is modulo $r$, but there is a (cyclic) permutation that puts it back in nonincreasing order, the result of which is $I\boxplus j$.
(Example: $r=10$, $t=5$, $\,I=(8,7,4,2,2)$, $\,I+3=(1,0,7,5,5)$, $\,I\boxplus 3=(7,5,5,1,0)$.)

By construction we then have
\beq
S_I^{q^j}=S_{I\boxplus j}
\eeq
in $\cB^{\mathfrak{S}_t}$.

Choosing a set of representatives
\beq
\cS\subset\cR
\eeq
for this action $\boxplus$ of $\Z/r\Z$ on $\cR$,
we note that
\beq
\cS\simeq\cR/(\Z/r\Z)\simeq((\Z/r\Z)^t/\mathfrak{S}_t)/(\Z/r\Z)\simeq(\Z/r\Z)^t/G, 
\eeq
so $\cS$ is also a set of representatives for the action of $G$ on $(\Z/r\Z)^t$.

\entry
\label{th_poly_descr_StF}
For each $I\in\cS$, we let $r_I$ be the size of its orbit under $\Z/r\Z$ in $\cR$, so $r_I|r$ and $r_I\Z/r\Z\subset\Z/r\Z$ is the stabilizer subgroup of $I$.
We then have $S_I^{q^{r_I}}=S_I$, so in fact $S_I$ defines a map
\beq
S_I:(\F_{q^r})^t\longto\F_{q^{r_I}}
\eeq
whose image lies in the subfield $\F_{q^{r_I}}\subset\F_{q^r}$.

\begin{theorem*}
With these notations, the map
\beq
\Psi:(\F_{q^r})^t\longto\prod_{I\in\cS}\F_{q^{r_I}}
\eeq
whose components are the $S_I$ for $I\in\cS$, is symmetric $t$-multilinear over $\Fq$,
and moreover it is \emph{universal} for this property.

In particular, it induces an isomorphism
\beq
S_{\Fq}^t\F_{q^r}\simeq\prod_{I\in\cS}\F_{q^{r_I}}
\eeq
of $\Fq$-vector spaces.
\end{theorem*}
\begin{proof}
We have to show that, for a certain basis $B$ of the dual $\Fq$-vector space $(\prod_{I\in\cS}\F_{q^{r_I}})^\vee$,
the $(b\circ\Psi)_{b\in B}$ form a basis of $(S_{\Fq}^t\F_{q^r})^\vee=\Sym^t_{\Fq}((\F_{q^r})^\vee)$ over $\Fq$.
For this we would like, ultimately, to apply Lemma~\ref{base_invariants} to the action of $\Z/r\Z$ on $\cB^{\mathfrak{S}_t}$. Indeed, as already noted $\cB^{\mathfrak{S}_t}=\Sym^t_{\Fq}(\F_{q^r};\F_{q^r})$, so its subspace of invariants under Frobenius is
\beq
\cB^G=(\cB^{\mathfrak{S}_t})^{\Z/r\Z}=\Sym^t_{\Fq}(\F_{q^r};\Fq)=\Sym^t_{\Fq}((\F_{q^r})^\vee).
\eeq 
Now, since the action of $\Z/r\Z$ is only $\Fq$-linear, we need a basis of $\cB^{\mathfrak{S}_t}$ over $\Fq$, stable under the action.

First choose a $\gamma\in\F_{q^r}$ such that
\beq
\gamma,\gamma^q,\dots,\gamma^{q^{r-1}}
\eeq
is a (normal) basis of $\F_{q^r}$ over $\Fq$.
Then, given $I\in\cS$, we set
\beq
\beta_{i,I}=\sum_{j\in\Z/r\Z,\,j\equiv i\bmod r_I}\gamma^{q^j}
\eeq
for $0\leq i\leq r_I-1$, which happen to form a basis of $\F_{q^{r_I}}$ over $\Fq$.
This is easily checked directly, but can also be viewed as a consequence of Lemma~\ref{base_invariants} (with $\Gamma=r_I\Z/r\Z$, $\cP=\Z/r\Z$, $V=\F_{q^r}$, $V^\Gamma=\F_{q^{r_I}}$).

In~\ref{action_St} we saw that the $(S_I)_{I\in\cR}$ form a basis of $\cB^{\mathfrak{S}_t}$ over $\F_{q^r}$.
It then follows that the
\beq
(\gamma^{q^i}S_I)_{i\in\Z/r\Z,I\in\cR}
\eeq
form a basis of $\cB^{\mathfrak{S}_t}$ over $\Fq$.
This basis is stable under the action of $\Z/r\Z$ on $\cB^{\mathfrak{S}_t}$ by Frobenius, more precisely we have
\beq
(\gamma^{q^i}S_I)^{q^j}=\gamma^{q^{i+j}}S_{I\boxplus j}.
\eeq
This means, our basis is indexed by
\beq
\cP=\Z/r\Z\times\cR,
\eeq
and $\Z/r\Z$ acts on this product set, on the first factor, by translation, and on the second factor, by the action $\boxplus$.
Let
\beq
\cT=\{(i,I)\,;\;0\leq i\leq r_I-1,\,I\in\cS\}\;\subset\cP
\eeq
be the set of representatives given by Lemma~\ref{action_produit}.

Now we can apply Lemma~\ref{base_invariants}, which gives that the
\beq
F_{i,I}=\sum_{j\in\Z/rZ}(\gamma^{q^i}S_I)^{q^j},
\eeq
for $(i,I)\in\cT$, form a basis of $(\cB^{\mathfrak{S}_t})^{\Z/r\Z}=\Sym^t_{\Fq}((\F_{q^r})^\vee)$ over $\Fq$.
Using the invariance of $S_I$ under $r_I\Z/r\Z$ and grouping together the $j$ according to their class modulo $r_I$, these can also be written
\beq
F_{i,I}=\sum_{j\in\Z/r_IZ}(\beta_{i,I}S_I)^{q^j}=\phi_{i,I}\circ S_I=\phi_{i,I}\circ\pi_I\circ\Psi
\eeq
where
\beq
\begin{array}{cccc}
\phi_{i,I}: & \F_{q^{r_I}} & \longto & \Fq\\
& x & \mapsto & \tr_{\F_{q^{r_I}}/\Fq}(\beta_{i,I}x)
\end{array}
\eeq
is the trace linear form deduced from $\beta_{i,I}$, and
\beq
\pi_I:\prod_{I\in\cS}\F_{q^{r_I}}\surj\F_{q^{r_I}} 
\eeq
is projection on the $I$-th factor.

Now, for fixed $I$, the $\beta_{i,I}$ form a basis of $\F_{q^{r_I}}$ over $\Fq$, so the $\phi_{i,I}$ form a basis of $(\F_{q^{r_I}})^\vee$ over $\Fq$.
Hence as $i$ and $I$ vary, the $\phi_{i,I}\circ\pi_I$ form a basis of $(\prod_{I\in\cS}\F_{q^{r_I}})^\vee$ over $\Fq$.
This is the basis $B$ we were looking for at the beginning of the proof.
\end{proof}

As a double check, Lemma~\ref{action_produit} also gives directly
\beq
\dim_{\Fq}\prod_{I\in\cS}\F_{q^{r_I}}=\sum_{I\in\cS}r_I=\abs{\cR}=\binom{r+t-1}{t}=\dim_{\Fq}S_{\Fq}^t\F_{q^r}.
\eeq
Also we note that Burnside's lemma allows us to compute $\abs{\cS}=\abs{\cR/(\Z/r\Z)}=\frac{1}{r}\sum_{d|\gcd(r,t)}\Eulerphi(d)\binom{(r+t)/d-1}{t/d}$,
although this will not be needed in the sequel.

\mysubsection{Equidistributed beads on a necklace}

\entry
Recall from \ref{action_St}-\ref{action_Frob} we are interested in the set $\cR=\cR_{r,t}\subset(\Z/r\Z)^t$ of nonincreasing $t$-tuples of elements of $\Z/r\Z\simeq\{0,1,\dots,r-1\}$,
of cardinality $\abs{\cR_{r,t}}=\binom{r+t-1}{t}$, modulo the action $\boxplus$ of $\Z/r\Z$, inherited from the diagonal action of $\Z/r\Z$ on $(\Z/r\Z)^t$ by translation.

There are several ways to interpret this object. For instance, we can also view it as the set of multisets of cardinality $t$ of elements of $\Z/r\Z$,
or as the set of vectors in $\N^r$ that sum to $t$ (identify a multiset with its characteristic vector), with the natural action of $\Z/r\Z$ by cyclic permutation.
So, in a sense, the quotient set $\cR_{r,t}/(\Z/r\Z)$ describes all the possible arrangements of $r$ beads with weight in $\N$ into a circular necklace of total weight $t$.

We introduce a particular element $I_{r,t}\in\cR_{r,t}$, which corresponds to the weight being equidistributed on the necklace:
\beq
\begin{split}
I_{r,t}&=\left(\left\lfloor\frac{(t-1)r}{t}\right\rfloor,\,\left\lfloor\frac{(t-2)r}{t}\right\rfloor,\,\dots,\left\lfloor\frac{r}{t}\right\rfloor,\,0\right)\\
&=\left(r-\left\lceil\frac{r}{t}\right\rceil,\,r-\left\lceil\frac{2r}{t}\right\rceil,\,\dots,r-\left\lceil\frac{(t-1)r}{t}\right\rceil,\,0\right)
\end{split}
\eeq

Equip $\cR_{r,t}$ with the lexicographic order, so for $I=(i_1,\dots,i_t)$ and $J=(j_1,\dots,j_t)$ in $\cR_{r,t}$,
we set $I<J$ if and only if there exists an index $a$ such that $i_b=j_b$ for all $b<a$, and $i_a<j_a$.


\begin{theorem}
\label{borne_equidistr}
Each orbit in $\cR_{r,t}/(\Z/r\Z)$ admits a representative $I\leq I_{r,t}$.
\end{theorem}

Example: $r=10$, $t=7$, $\,I_{10,7}=(8,7,5,4,2,1,0)$.
Let $J=(9,8,7,6,4,3,1)$. Then in the orbit of $J$ we can find $J\boxplus4=(8,7,5,3,2,1,0)<I_{10,7}$,
and also $J\boxplus7=(8,6,5,4,3,1,0)<I_{10,7}$ (so there is not unicity in the Theorem).


Note that the Theorem is stated for multisets, but then, it applies \emph{a fortiori} to ordinary sets.
So it gives that, for any subset $J\subset\Z/r\Z$ of cardinality $\abs{J}=t$, there is a translate $I$ of $J$ in $\Z/r\Z$
whose largest elements are $\lfloor\frac{(t-1)r}{t}\rfloor,\dots,\lfloor\frac{(t-a+1)r}{t}\rfloor$,
but then the next one (if applicable) is smaller than $\lfloor\frac{(t-a)r}{t}\rfloor$.
Moreover, here we used the lexicographic order, but the very same method of proof
gives a similar result for the antilexicographic order: there is also a translate $I$ of $J$ in $\Z/r\Z$
whose smallest elements are $0,\lfloor\frac{r}{t}\rfloor,\dots,\lfloor\frac{(a-1)r}{t}\rfloor$,
but then the next one (if applicable) is smaller than $\lfloor\frac{ar}{t}\rfloor$.

The proof of the Theorem will require several intermediary results:

\begin{definition}
We say $J=(j_1,\dots,j_t)\in\cR_{r,t}$ is reduced if $j_t=0$.
We let $\cR_{r,t}^0\subset\cR_{r,t}$ be the set of reduced elements.
\end{definition}

For instance, $I_{r,t}\in\cR_{r,t}^0$ is reduced. Note also that forgetting the last coordinate gives $\cR_{r,t}^0\simeq\cR_{r,t-1}$, so $\abs{\cR_{r,t}^0}=\binom{r+t-2}{t-1}$.


\entry
Let
\beq
\cG_{r,t}\subset\N_{>0}\times\N^{t-1}
\eeq
be the set of $t$-tuples of integers $(g_1,\dots,g_t)$ with $g_1>0$ and sum
\beq
g_1+\cdots+g_t=r
\eeq
so $\abs{\cG_{r,t}}=\binom{r+t-2}{t-1}$.
Equip $\cG_{r,t}$ with the lexicographic order.

Given $J=(j_1,\dots,j_t)\in\cR_{r,t}^0$, so $j_t=0$, we define its $i$-th gap, $1\leq i\leq t$, as
\beq
g_i(J)=
\begin{cases}
r-j_1 & \textrm{for $i=1$}\\
j_{i-1}-j_i & \textrm{for $2\leq i\leq t$}
\end{cases}
\eeq
so in particular $g_t(J)=j_{t-1}$, and we let its gap sequence be
\beq
g(J)=(g_1(J),\dots,g_t(J))\in\cG_{r,t}.
\eeq
Then:

\begin{lemma}
\label{gap_bij}
This map $g:\cR_{r,t}^0\longto\cG_{r,t}$ is an order-reversing bijection.
\end{lemma}
\begin{proof}
Indeed, to $(g_1,\dots,g_t)\in\cG_{r,t}$, the inverse map associates the $t$-uple $(j_1,\dots,j_t)\in\cR_{r,t}^0$ given by
\beq
j_i=r-(g_1+\cdots+g_i),
\eeq
which is clearly order-reversing.
\end{proof}

\entry
We let $\Z/t\Z$ act on $\N^t$ by cyclic permutation. More precisely, we let $\sigma_0$ be the identity on $\N^t$,
and for $g=(g_1,\dots,g_t)\in\N^t$ and $1\leq a\leq t-1$ we set
\beq
\sigma_a(g)=(g_{a+1},g_{a+2},\dots,g_t,g_1,g_2,\dots,g_a).
\eeq
This action ``almost'' preserves $\cG_{r,t}$: more precisely, for $g\in\cG_{r,t}$, we have $\sigma_a(g)\in\cG_{r,t}$ if and only if $g_{a+1}>0$.

\entry
\label{cyclic_gaps}
Let $J=(j_1,\dots,j_t)\in\cR_{r,t}^0$
and let $j>0$ be such that $j=j_a$ for an index $a$; if there are several choices for such an $a$, choose it maximum, so $j_a>j_{a+1}$, hence $g_{a+1}>0$.
Then:
\begin{lemma*}
With these notations, we have
\beq
g(J\boxplus(r-j_a))=\sigma_a(g(J)).
\eeq
\end{lemma*}
\begin{proof}
Clear, from
\beq
\begin{split}
J\boxplus(r-j_a)=(j_{a+1}+r-j_a,\;j_{a+2}+r-j_a,\;\dots,&\;j_{t-1}+r-j_a,\;r-j_a,\\
&\!\!\!j_1-j_a,\;j_2-j_a,\;\dots,\;j_{a-1}-j_a,0).
\end{split}
\eeq
\end{proof}

\begin{lemma}[``large gap'']
\label{gros_trou}
Let $J=(j_1,\dots,j_t)\in\cR_{r,t}^0$ have gap sequence $g(J)=(g_1,\dots,g_t)$.
Suppose there is an index $a$ such that $g_a>\lceil\frac{r}{t}\rceil$.
Then there is $j\in\Z/r\Z$ such that $J\boxplus j<I_{r,t}$.
\end{lemma}
\begin{proof}
Thanks to Lemma~\ref{cyclic_gaps}, after possibly replacing $J$ with $J\boxplus(r-j_{a-1})$ if $a>1$, we can suppose $a=1$.
Then
\beq
j_1=r-g_1\;<\;r-\left\lceil\frac{r}{t}\right\rceil
\eeq
so $J<I_{r,t}$.
\end{proof}

\begin{lemma}[``small gap'']
\label{petit_trou}
Let $J=(j_1,\dots,j_t)\in\cR_{r,t}^0$ have gap sequence $g(J)=(g_1,\dots,g_t)$.
Suppose there is an index $a$ such that $g_a<\lfloor\frac{r}{t}\rfloor$.
Then there is $j\in\Z/r\Z$ such that $J\boxplus j\leq I_{r,t}$.
\end{lemma}
\begin{proof}
Choose $a$ such that $g_a$ is minimum; if there are several choices for such an $a$, choose it maximum, so $g_{a+1}>g_a$ if $a<t$.
Then, thanks to Lemma~\ref{cyclic_gaps}, after possibly replacing $J$ with $J\boxplus(r-j_a)$ if $a<t$, we can suppose $a=t$, so $j_{t-1}=g_t<\lfloor\frac{r}{t}\rfloor$,
hence, since these are integers,
\beq
j_{t-1}\leq\left\lfloor\frac{r}{t}\right\rfloor-1\leq\frac{r}{t}-1.
\eeq

Now we proceed by contradiction: suppose $J\boxplus j>I_{r,t}$ for all $j$.
We will construct a sequence of indices $1\leq a_1<a_2<\dots<a_k<t$ with the following properties:
\begin{enumerate}[(i)]
\item $r-j_{a_1}<\frac{a_1r}{t}$
\item $j_{a_i}-j_{a_{i+1}}<\frac{(a_{i+1}-a_i)r}{t}\quad$ for $2\leq i\leq k-1$
\item $j_{a_k}<\frac{(t-a_k)r}{t}$.
\end{enumerate}
Summing all these inequalities we find $r<r$, a contradiction.

The sequence is constructed as follows.
To start with, we have $J>I_{r,t}$ so there is an index $a<t$ such that $j_a>r-\lceil\frac{ar}{t}\rceil$.
If there are several choices for such an $a$, choose it maximum
(which implies $j_a>j_{a+1}$), and call it $a_1$.
Then $r-j_{a_1}<\lceil\frac{a_1r}{t}\rceil$ and moreover $r-j_{a_1}$ is an integer, so $r-j_{a_1}<\frac{a_1r}{t}$.

Now suppose we have already constructed $1\leq a_1<a_2<\dots<a_\ell<t$ satisfying (i) and (ii), and with $j_{a_\ell}>j_{a_\ell+1}$.
If $a_\ell=t-1$, then we are done: (iii) is more than satisfied (with $k=\ell$) since $j_{t-1}\leq\frac{r}{t}-1<\frac{r}{t}$.

If $a_\ell<t-1$, we use the fact that $J\boxplus(r-j_{a_\ell})>I_{r,t}$, so for some $b\geq1$, these sequences coincide on the first $b-1$ positions, while the $b$-th coefficient
of $J\boxplus(r-j_{a_\ell})$ (whose expression is given in the proof of Lemma~\ref{cyclic_gaps}) is larger than that of $I_{r,t}$. We distinguish two cases.

First case: $b<t-a_\ell$. Then $j_{a_\ell+b}+r-j_{a_\ell}>r-\lceil\frac{br}{t}\rceil$, with the left-hand side an integer, hence in fact $j_{a_\ell+b}+r-j_{a_\ell}>r-\frac{br}{t}$.
Thus there exists an index $a$ (namely here $a=a_\ell+b$ works) with $a_\ell<a<t$ and $j_{a_\ell}-j_a<\frac{(a_\ell-a)r}{t}$.
If there are several choices for such an $a$, choose it maximum
(which implies $j_a>j_{a+1}$), and call it $a_{\ell+1}$.

Second case: $b\geq t-a_\ell$, so the $(t-1-a_\ell)$-th coefficient of $J\boxplus(r-j_{a_\ell})$ is equal to that of $I_{r,t}$, 
that is, $j_{t-1}+r-j_{a_\ell}=r-\lceil\frac{(t-1-a_\ell)r}{t}\rceil$, or $j_{a_\ell}=j_{t-1}+\lceil\frac{(t-1-a_\ell)r}{t}\rceil<j_{t-1}+\frac{(t-1-a_\ell)r}{t}+1$.
But then we use $j_{t-1}\leq\frac{r}{t}-1$ to conclude that (iii) is satisfied with $k=\ell$.
%
%
\end{proof}

\begin{definition}
Let $J=(j_1,\dots,j_t)\in\cR_{r,t}^0$ have gap sequence $g(J)=(g_1,\dots,g_t)$.
We say $J$ is \emph{balanced} if $\abs{g_a-\frac{r}{t}}<1$ for all $a$ (\ie $g_a=\lfloor\frac{r}{t}\rfloor$ or $\lceil\frac{r}{t}\rceil$).
We let
\beq
\cR_{r,t}^\textrm{bal}\subset\cR_{r,t}^0
\eeq
be the set of such balanced sequences.
\end{definition}

\entry
\label{derived_sequence}
Suppose $t{\not|}\,r$, set $Q=\lceil\frac{r}{t}\rceil$, so $Q-1=\lfloor\frac{r}{t}\rfloor$, and write
\beq
r=tQ-u
\eeq
with $0<u<t$, which can also be written
\beq
r=u\left\lfloor\frac{r}{t}\right\rfloor+(t-u)\left\lceil\frac{r}{t}\right\rceil.
\eeq
Let $J=(j_1,\dots,j_t)\in\cR_{r,t}^\textrm{bal}$ be a balanced sequence, with gaps $g(J)=(g_1,\dots,g_t)$.
Since all $g_a=\lfloor\frac{r}{t}\rfloor$ or $\lceil\frac{r}{t}\rceil$ and they must sum to $r$, we deduce exactly $u$ of them are equal to $\lfloor\frac{r}{t}\rfloor$,
and the other $t-u$ are equal to $\lceil\frac{r}{t}\rceil$. So let $a_1<a_2<\dots<a_u$ be those indices $a$ with $g_a=\lfloor\frac{r}{t}\rfloor$, and then for $1\leq i\leq u$
set $b_i=r-a_i$. Note $1\leq a_i\leq t$ so $b_i\in\{0,\dots,t-1\}\simeq\Z/t\Z$, and the $b_i$ form a decreasing sequence, hence they define an element of $\cR_{t,u}$.
\begin{definition*}
With these notations, we call
\beq
\partial(J)=(b_1,\dots,b_u)
\eeq
the derived sequence of $J$.
\end{definition*}

\begin{proposition}
\label{derive_croit}
This map
\beq
\partial:\cR_{r,t}^\textrm{bal}\longto\cR_{t,u}
\eeq
is injective and order-preserving.
%
\end{proposition}
\begin{proof}
The map $\partial$ factorizes as
\beq
\cR_{r,t}^\textrm{bal}\xrightarrow{\,g_{|\cR_{r,t}^\textrm{bal}}\,}\cG_{r,t}^\textrm{bal}\xrightarrow{\;\;\phantom{{}_|}b\phantom{{}_|}\;\;}\cR_{t,u}.
\eeq
Here $\cG_{r,t}^\textrm{bal}=g(\cR_{r,t}^\textrm{bal})\subset\cG_{r,t}$ is the set of $t$-tuples of integers $(g_1,\dots,g_t)\in\N^t$
among which $u$ of them are equal to $\lfloor\frac{r}{t}\rfloor$,
and the other $t-u$ are equal to $\lceil\frac{r}{t}\rceil$, and with $g_1>0$.
Note that, if $t<r$, the last condition vanishes since $g_1\geq\lfloor\frac{r}{t}\rfloor>0$ automatically, while if $t>r$, then $g_1>0$ means $g_1=\lceil\frac{r}{t}\rceil$.
Then by Lemma~\ref{gap_bij}, the restriction $g_{|\cR_{r,t}^\textrm{bal}}$ is an order-reversing bijection.

The second map is $b(g_1,\dots,g_t)=(b_1,\dots,b_u)$, where $b_i=t-a_i$ and $a_1<a_2<\dots<a_u$ are those indices $a$ with $g_a=\lfloor\frac{r}{t}\rfloor$.
This map is clearly injective.

Now suppose $(g_1,\dots,g_t)<(g'_1,\dots,g'_t)$, so there is some index $k$ with $g_i=g'_i$ for $i<k$, and $g_k=\lfloor\frac{r}{t}\rfloor<g'_k=\lceil\frac{r}{t}\rceil$.
So $k=a_v<a'_v$ for some $v$, while $a_w=a'_w$ for $w<v$. This means $(a_1,\dots,a_u)<(a'_1,\dots,a'_u)$, or equivalently, $(b_1,\dots,b_u)>(b'_1,\dots,b'_u)$.
Hence the map $b$ is also order-reversing.

We conclude since composing two order-reversing maps gives an order-preserving map.
\end{proof}

\begin{proposition}
\label{derive_I}
Suppose $t{\not|}\,r$, and write $r=tQ-u$ with $Q=\lceil\frac{r}{t}\rceil$, $0<u<t$. Then $I_{r,t}\in\cR_{r,t}^\textrm{bal}$ and
\beq
\partial(I_{r,t})=I_{t,u}.
\eeq
\end{proposition}
\begin{proof}
We have $g(I_{r,t})=(g_1,\dots,g_t)$ with
\beq
g_a=\left\lceil\frac{ar}{t}\right\rceil-\left\lceil\frac{(a-1)r}{t}\right\rceil.
\eeq
However, writing $\frac{ar}{t}=aQ-\frac{au}{t}$ we find $\lceil\frac{ar}{t}\rceil=aQ-\lfloor\frac{au}{t}\rfloor$. Then, since $\frac{u}{t}<1$, there are only two possibilities:
\begin{itemize}
\item either $\lfloor\frac{au}{t}\rfloor=\lfloor\frac{(a-1)u}{t}\rfloor$, in which case $g_a=Q=\lceil\frac{r}{t}\rceil$,
\item or $\lfloor\frac{au}{t}\rfloor=\lfloor\frac{(a-1)u}{t}\rfloor+1$, and $g_a=Q-1=\lfloor\frac{r}{t}\rfloor$.
\end{itemize}
We are interested in the second case. It happens precisely when
\beq
\frac{(a-1)u}{t}<m\leq\frac{au}{t}
\eeq
for a certain integer $m$ (namely $m=\lfloor\frac{au}{t}\rfloor$), that is when $a-1<\frac{mt}{u}\leq a$, or
\beq
a=a_m=\left\lceil\frac{mt}{u}\right\rceil.
\eeq
Setting $b_m=t-a_m$, we find
\beq
\partial(I_{r,t})=(b_1,\dots,b_u)=I_{t,u}
\eeq
as claimed.
\end{proof}

\begin{lemma}
\label{bal_cyclic}
Suppose $t{\not|}\,r$. Let $J=(j_1,\dots,j_t)\in\cR_{r,t}^\textrm{bal}$, and let $v$ be an index with $j_v>j_{v+1}$.
Then
\beq
\partial(J\boxplus(r-j_v))=\partial(J)\boxplus(t-v)
\eeq
in $\cR_{t,u}$.
\end{lemma}
\begin{proof}
Consequence of Lemma~\ref{cyclic_gaps} and the definition of $\partial$.
\end{proof}

\begin{proof}[Proof of Theorem~\ref{borne_equidistr}]
We proceed by induction on $t$. The result is clear for $t=1$. Suppose it holds for all $t'<t$.

Let $J=(j_1,\dots,j_t)\in\cR_{r,t}$. After possibly replacing $J$ with $J\boxplus(r-j_t)=J-j_t$, we can suppose $J\in\cR_{r,t}^0$.
If $J$ has a gap larger than $\lceil\frac{r}{t}\rceil$ we conclude with Lemma~\ref{gros_trou};
and if $J$ has a gap smaller than $\lfloor\frac{r}{t}\rfloor$ we conclude with Lemma~\ref{petit_trou}.
So the only case remaining is $J\in\cR_{r,t}^\textrm{bal}$.
If $t|r$, this means all gaps are equal to $\frac{r}{t}$, and then $J=I_{r,t}$.
Now suppose $t{\not|}\,r$, set $Q=\lceil\frac{r}{t}\rceil$, and write $r=tQ-u$ with $0<u<t$.
Then $\partial(J)\in\cR_{t,u}$ and we can apply the induction hypothesis to find $v$ such that
\beq
\partial(J)\boxplus(t-v)\leq I_{t,u}.
\eeq
In particular the first coefficient of $\partial(J)\boxplus(t-v)$ is at most $t-\lceil\frac{t}{u}\rceil<t-1$. This means
that the first gap of $\partial(J)\boxplus(t-v)$, or equivalently the $(v+1)$-th gap of $J$,
is not equal to $\lfloor\frac{r}{t}\rfloor$, so it is $\lceil\frac{r}{t}\rceil\geq 1$,
which means in turn $j_v>j_{v+1}$.
We can then apply Lemma~\ref{bal_cyclic} to the left-hand side of our last inequality, and Proposition~\ref{derive_I} to the right-hand side,
to restate it as
\beq
\partial(J\boxplus(r-j_v))\leq\partial(I_{r,t}).
\eeq
But $\partial$ is order-preserving by Proposition~\ref{derive_croit}, so
\beq
J\boxplus(r-j_v)\leq I_{r,t}
\eeq
which finishes the proof.
\end{proof}

\entry
\label{borne_degSI}
Theorem~\ref{borne_equidistr} can be used to make Theorem~\ref{th_poly_descr_StF} more precise.
Recall there we are working in an extension $\F_{q^r}$ of a finite field $\Fq$, and we have constructed a universal family
of symmetric homogeneous $t$-multilinearized polynomials
\beq
S_I:(\F_{q^r})^t\longto\F_{q^{r_I}}
\eeq
for $I$ ranging in a set $\cS$ of representatives of $\cR_{r,t}$ modulo $\Z/r\Z$.
\begin{corollary*}
Suppose $t\leq q$, and set
\beq
T=q^{\left\lfloor\frac{(t-1)r}{t}\right\rfloor}+q^{\left\lfloor\frac{(t-2)r}{t}\right\rfloor}+\cdots+q^{\left\lfloor\frac{r}{t}\right\rfloor}+1.
\eeq
Then the set $\cS$ can be chosen so that the polynomials $S_I$ all have
\beq
\deg(S_I)\leq T.
\eeq
\end{corollary*}
\begin{proof}
Recall for $I\in\cR_{r,t}$ we defined $D_I=q^{i_1}+\cdots+q^{i_t}=\deg(M_I)=\deg(S_I)$.
In particular we have $T=D_{I_{r,t}}$.
However, since $t\leq q$, it is easily seen that $D_I\leq D_{I'}$ if and only if $I\leq I'$ for the lexicographic order.
To construct $\cS$, in each orbit of $\cR_{r,t}/(\Z/r\Z)$ we choose the representative $I$ that is minimum for the lexicographic order,
so $I\leq I_{r,t}$ by Theorem~\ref{borne_equidistr}, and we conclude.
\end{proof}

(Note that for $t>q$, it can be that $D_I>D_{I'}$ although $I<I'$. This happens for instance, for $t=3$, $q=2$, $I=(2,2,2)$, $I'=(3,0,0)$.
In such a case, instead of choosing in each orbit the representative $I$ minimum for the lexicographic order,
perhaps it is better to choose the one that gives the smallest $D_I$.)

\begin{example}
For $t=2$, Theorem~\ref{borne_equidistr} specializes to the results of~\cite{agis}. Here we have
\beq
\cS=\{(0,0),(1,0),(2,0),\dots,(\lfloor r/2\rfloor,0)\}
\eeq
and the $S_I$ for $I\in\cS$ are the maps $m_0(x,y)=xy$ and $m_i(x,y)=x^{q^i}y+xy^{q^i}$ for $1\leq i\leq\lfloor r/2\rfloor$. 
The maximum degree is $q^{\lfloor r/2\rfloor}+1$, reaching the bound in Corollary~\ref{borne_degSI}.

If $r$ is odd, then
\beq
(\F_{q^r})^{\frac{r+1}{2}}
\eeq
can be seen as a $\Fq$-vector space of dimension $\frac{r(r+1)}{2}$.
On the other hand, if $r$ is even, by abuse of notation we set
\beq
(\F_{q^r})^{\frac{r+1}{2}}=(\F_{q^r})^{r/2}\times\F_{q^{r/2}}
\eeq
and again this can be seen as a $\Fq$-vector space of dimension $\frac{r(r+1)}{2}$; note then that $m_{r/2}$ takes values in $\F_{q^{r/2}}$.
In any case, the map $\Psi=(m_0,m_1,\dots,m_{\lfloor r/2\rfloor})$ induces an isomorphism
\beq
S_{\Fq}^2\F_{q^r}\simeq(\F_{q^r})^{\frac{r+1}{2}} 
\eeq
of $\Fq$-vector spaces.

This can be viewed as a symmetric variant of the isomorphism
\beq
\F_{q^r}\tens_{\Fq}\F_{q^r}\simeq(\F_{q^r})^r
\eeq
induced by the maps $(x,y)\mapsto x^{q^i}y$ for $i\in\Z/r\Z$,
although the latter has the additional property that it is in fact an isomorphism of $\Fq$-algebras.
\end{example}

\begin{example}
For $q=2$, $r=2$, $t=3$, we can take
\beq
\cS=\{(0,0,0),(1,0,0)\}
\eeq
with associated maps $m(x,y,z)=xyz$ and $m'(x,y,z)=x^2yz+xy^2z+xyz^2$ on $(\F_4)^3$. Then $(m,m')$ induces an isomorphism
\beq
S_{\F_2}^3\F_4\simeq(\F_4)^2 
\eeq
of $\F_2$-vector spaces, of maximum degree $4$.

For $q=2$, $r=3$, $t=3$, we can take
\beq
\cS=\{(0,0,0),(1,0,0),(1,1,0),(2,1,0)\}
\eeq
with associated maps $\psi_1(x,y,z)=xyz$, $\psi_2(x,y,z)=x^2yz+xy^2z+xyz^2$, $\psi_3(x,y,z)=x^2y^2z+x^2yz^2+xy^2z^2$,
and $\psi_4(x,y,z)=x^4y^2z+x^4yz^2+x^2y^4z+x^2yz^4+xy^4z^2+xy^2z^4$ on $(\F_8)^3$.
Note that $\psi_4$ is invariant under Frobenius, so it takes values in $\F_2$, and then $(\psi_1,\psi_2,\psi_3,\psi_4)$ induces an isomorphism
\beq
S_{\F_2}^3\F_8\simeq(\F_8)^3\times\F_2 
\eeq
of $\F_2$-vector spaces, of maximum degree $7$.
\end{example}

\refstepcounter{section}
\section*{Appendix \thesection: Review of open questions}
\label{sect_open_quesions}

Here is a subjective selection of (hopefully) interesting open problems related to the topic of products of codes, some of which were already mentioned in the text.

\entry
Study arithmetic in the $+,*,\subset\,$-ordered semiring of linear codes of length $n$ over $\Fq$ (either for $n$ fixed, or for $n\to\infty$).

One could devise an almost infinite number of questions, but perhaps the most natural ones to start with concern the distribution of squares \cite{agis}:
\begin{itemize}
\item Among all linear codes of length $n$ over $\Fq$, how many of them are squares?
\item What is the maximum number of square roots a code can admit?
\item If a code is not a square, what are the largest squares it contains?
or the smallest squares it is contained in?
\end{itemize}
This leads to some approximation problems. Consider the metric $\dist(C,C')=\dim(C+C')-\dim(C\cap C')$. Then:
\begin{itemize}
\item How far can a code be from the set of squares?
\end{itemize}

\entry
These questions can also be made more algorithmic:
\begin{itemize}
\item Is there an efficient algorithm to decide if a code if a square?
\item If so, is there an efficient algorithm to compute one of its square roots? or to compute all of them?
\end{itemize}

\entry
In~\ref{open_pb_sym} we pointed out three open problems concerning symmetries and automorphisms of powers of codes:
\begin{itemize}
\item Compare the $\Aut(C\deux[t])$ as $t$ varies, where $\Aut(C)\subset\Aut(\F^n)$ is the group of monomial transformations preserving $C$;
\eg for all $C$ and $t$, does it hold $\Aut(C\deux[t])\,\widehat{\subset}\,\Aut(C\deux[t+1])$?
(From \ref{aut_divise} and \ref{aut_stable} we know $\Aut(C\deux[t])\,\widehat{\subset}\,\Aut(C\deux[t'])$ only for $t|t'$ or $t'\geq r(C)$.)
\item Instead of $\Aut(C)$, compare the $\Aut^{\textrm{in}}(C\deux[t])$ as $t$ varies, where $\Aut^{\textrm{in}}(C)$ is the group of invertible linear endomorphisms of $C$
(seen as an abstract vector space) that preserve the Hamming metric.
\item Do the same thing for semilinear automorphisms.
\end{itemize}

\entry
Study how change of field operations interact with product of codes.
Here is what is known:

\begin{table}[h]
\vspace{-.2\baselineskip}
\begin{tabular}{|c|c|}
\hline
extension of scalars & obvious (\ref{compatibilites_extension}\eqref{compatibilites_extension_+*}, \ref{parametres_extension})\\
\hline
concatenation & partial results, highly dependent on the inner code (\ref{prop_concat})\\
\hline
trace code & not yet understood\\
\hline
subfield subcode & not yet understood\\
\hline
\end{tabular}
\vspace{-.8\baselineskip}
\end{table}

As noted in~\ref{appl_McEliece}, results on products of subfield subcodes would be useful in the analysis of the original McEliece cryptosystem.
Since the trace code operation is somehow dual to the subfield subcode operation, one could guess both to be equally difficult to understand.
With some abuse of language one can consider trace codes as a very specific class of concatenated codes (with a noninjective symbol mapping).
This raises the question of whether some of the techniques introduced for the study of products of concatenated codes
(\eg polynomial representation of $\Fq$-multilinear maps) could also be applied to trace codes.

\entry
Improve the bounds on the possible joint parameters of a family of codes and their product,
or of a code and its powers, given in section \ref{sect_pure_bounds}.
In particular:
\begin{itemize}
\item Beside the Singleton bound, try to generalize the other classical bounds of coding theory in the context of products of codes.
\item Fill the large gap between \ref{cor_Singleton} and \ref{borne_inf_alpha2}:
\beq
0.001872-0.5294\,\delta\leq\alpha_2^{\langle 2\rangle}(\delta)\leq 0.5-0.5\,\delta.
\eeq
\item
Is $\tau(2)>2$? That means, does there exist an asymptotically good family of binary linear codes $C$ whose cubes $C\deux[3]$ also form an asymptotically good family?
\item
Is $\tau(2)$ infinite?
That means, for any $t$ (arbitrarily large), does there exist an asymptotically good family of binary linear codes $C$ whose $t$-th powers $C\deux[t]$ also form an asymptotically good family?
\end{itemize}

\entry
Does there exist an asymptotically good family of binary linear codes $C$, whose squares $C\deux$, and also whose dual codes $C^\perp$, form an asymptotically good family?

(Motivation comes from the theory of multi-party computation, as mentioned after~\ref{fin_multilin_algo}.)

If instead of binary codes, one is interested in $q$-ary codes, then AG codes are easily seen to provide a positive answer, at least for $q$ large.
When $q$ becomes smaller, AG codes still work, except perhaps for $q=2,3,4,5,7,11,13$:
this is shown in~\cite{CCX}, using a careful analysis of the $2$-torsion in the class group of a certain tower of curves.

However, for $q=2$, curves simply do not have enough points, so there is no hope that bare AG codes work in this case.
Probably one should combine AG codes with another tool.
Note that concatenation as in \cite{agis} does not work, since it destroys the dual distance.

\entry
Our bound \ref{bornes_concat} on $t$-th powers of concatenated codes, for $t\leq q$, relied on Appendix~\ref{sect_multilinearized},
in which we gave a polynomial description of the symmetric power $S_{\Fq}^t\F_{q^r}$, using homogeneous $t$-multilinearized polynomials of controlled degree.

To extend this bound for $t>q$, it would be nice to have a similar result also for the Frobenius symmetric power $S_{Frob,\Fq}^t\F_{q^r}$, as defined in~\ref{Frobenius_symmetric_algebra}.
That means: construct a universal family of Frobenius symmetric $t$-multilinearized polynomials, of controlled degree.

Such a construction seems highly unlikely if one keeps the homogeneity condition; so we might drop this condition, to allow more flexibility.
Indeed, usually this will not be a problem for applications. For instance,
Proposition~\ref{prop_concat} can deal with non-homogeneous polynomials,
provided the external code $C$ contains the all-$1$ word $1_{[n]}$.
It is often so in practice, \eg when $C$ is an AG code.

\entry
In~\ref{longueurs} we defined $n_i=\dim \langle x\in C^{\perp}\,;\;w(x)\leq i\rangle^{\perp}$,
and noted that $n_0$ is the length of $C$, while $n_1$ is its support length, and $n_2$ its projective length. Is there such a nice interpretation for the subsequent values $n_i$, $i\geq3$?
Or conversely, is there another ``natural'' sequence of which $n_0,n_1,n_2$ are the first terms?


\begin{thebibliography}{1}

\bibitem{BR}
S.~Ballet \& R.~Rolland.
``On the bilinear complexity of the multiplication in finite fields.''
In: Y.~Aubry \& G.~Lachaud, Eds.
\emph{Arithmetic, Geometry and Coding Theory (AGCT 2003).}
S\'em. Congr. Vol.~11,
Soci\'et\'e Math\'ematique de France, 2005,
pp.~179-188.

\bibitem{BDEZ}
R.~Barbulescu, J.~Detrey, N.~Estibals \& P.~Zimmermann.
``Finding optimal formulae for bilinear maps.''
Intern. Workshop on Arithmetics of Finite Fields (WAIFI 2012),
Bochum, Germany, July 16-19, 2012.
Online version: \url{http://hal.inria.fr/hal-00640165}

\bibitem{BS}
E.~Barnes \& N.~Sloane.
``New lattice packings of spheres.''
Canad. J.~Math.~\textbf{35} (1983) 117-130.

\bibitem{BGW}
M.~Ben-Or, S.~Goldwasser \& A.~Wigderson.
``Completeness theorems for non-cryptographic fault-tolerant distributed computation.''
Proc. 20th Ann. ACM Symp. on Theory of Computing (STOC '88), 1988, pp.~1-10. 

\bibitem{BMT}
E.~Berlekamp, R.~McEliece \& H. van Tilborg.
``On the inherent intractability of certain coding problems.''
IEEE Trans. Inform. Theory~\textbf{24} (1978) 203-207.

\bibitem{BD}
R.~Brockett \& D.~Dobkin.
``On the optimal evaluation of a set of bilinear forms.''
Proc. 5th Ann. ACM Symp. on Theory of Computing (STOC '73), 1973, pp.~88-95. 

\bibitem{Bshouty}
N.~Bshouty.
``Multilinear complexity is equivalent to optimal tester size.''
Electron. Colloq. Comput. Complexity, Report No.~TR13-011 (2013).



\bibitem{CCX}
I.~Cascudo, R.~Cramer \& C.~Xing.
``The torsion-limit for algebraic function fields and its application to arithmetic secret sharing.''
In: P.~Rogaway, Ed.
\emph{Advances in cryptology --- CRYPTO 2011.}
Lecture Notes in Comp. Science Vol.~6841,
Springer-Verlag, Berlin, 2011, pp.~685-705.

\bibitem{CCD}
D.~Chaum, C.~Cr\'epeau \& I.~Damg\r{a}rd.
``Multiparty unconditionally secure protocols.''
Proc. 20th Ann. ACM Symp. on Theory of Computing (STOC '88), 1988, pp.~11-19. 

\bibitem{CC}
H.~Chen \& R.~Cramer.
``Algebraic geometric secret sharing schemes and secure multi-party computations over small fields.''
In: C.~Dwork, Ed.
\emph{Advances in cryptology --- CRYPTO 2006.}
Lecture Notes in Comp. Science Vol.~4117,
Springer-Verlag, Berlin, 2006, pp.~521-536.

\bibitem{ChCh}
D.~Chudnovsky \& G.~Chudnovsky.
``Algebraic complexities and algebraic curves over finite fields.''
J. Complexity~\textbf{4} (1988) 285-316.

\bibitem{CL}
G.~Cohen and A.~Lempel,
``Linear intersecting codes.''
Discr. Math.~\textbf{56} (1985) 35-43.

\bibitem{CGLM}
P.~Comon, G.~Golub, L.-H.~Lim \& B.~Mourrain.
``Symmetric tensors and symmetric tensor rank.''
SIAM J. Matrix Anal. Appl.~\textbf{30} (2008) 1254-1279.


\bibitem{CGGOT}
A.~Couvreur, P.~Gaborit, V.~Gauthier, A.~Otmani \& J.-P.~Tillich.
``Distinguisher-based attacks on public-key cryptosystems using Reed-Solomon codes.''
Presented at WCC 2013, to appear in Des. Codes Crypto.
Preprint: \url{http://arxiv.org/abs/1307.6458}

\bibitem{CDM}
R.~Cramer, I.~Damg\r{a}rd \& U.~Maurer.
``General secure multi-party computation from any linear secret-sharing scheme.''
In: B.~Preneel, Ed.
\emph{Advances in cryptology --- EUROCRYPT 2000.}
Lecture Notes in Comp. Science Vol.~1807,
Springer-Verlag, Berlin, 2000, pp.~316-334.

\bibitem{CK}
C.~Cr\'epeau \& J.~Kilian.
``Achieving oblivious transfer using weakened security assumptions.''
Proc. 29th IEEE Symp. on Found. of Computer Sci. (FOCS '88), 1988, pp.~42-52.

\bibitem{DV}
V.~Drinfeld \& S.~Vladut.
``Number of points of an algebraic curve.''
Funct. Anal.~\textbf{17} (1983) 53-54.

\bibitem{DP}
I.~Duursma \& R.~Pellikaan.
``A symmetric Roos bound for linear codes.''
J.~Combin. Theory Ser.~A \textbf{113} (2006) 1677-1688. 

\bibitem{Eisenbud}
D.~Eisenbud.
\emph{The geometry of syzygies. A second course in commutative algebra and algebraic geometry.}
Graduate Texts in Math. Vol.~229,
Springer-Verlag, New York, 2005.

\bibitem{EP}
D.~Eisenbud \& S.~Popescu.
``Gale duality and free resolutions of ideals of points.''
Invent. Math.~\textbf{136} (1999) 419-449.

\bibitem{Forney88-1}
G.~Forney.
``Coset codes I: Introduction and geometrical classification.''
IEEE Trans. Inform. Theory~\textbf{34} (1988) 1123-1151.

\bibitem{GSBB}
A.~Garcia, H.~Stichtenoth, A.~Bassa \& P.~Beelen.
``Towers of function fields over non-prime finite fields.''
To appear.
Preprint: \url{http://arxiv.org/abs/1202.5922}

\bibitem{EGA2}
A.~Grothendieck \& J.~Dieudonn\'e.
``\'El\'ements de g\'eom\'etrie alg\'ebrique. II. \'Etude globale \'el\'ementaire de quelques classes de morphismes.''
Inst. Hautes \'Etudes Sci. Publ. Math.~\textbf{8} (1961).

\bibitem{Hartshorne}
R.~Hartshorne.
\emph{Algebraic geometry.}
Graduate Texts in Math. Vol.~52,
Springer-Verlag, New York-Heidelberg, 1977.

\bibitem{Ihara}
Y.~Ihara.
``Some remarks on the number of rational points of algebraic curves over finite fields.''
\emph{J.~Fac. Sci. Univ. Tokyo Sect.~IA Math.}~\textbf{28} (1981) 721-724.

\bibitem{IKOPSW}
Y.~Ishai, E.~Kushilevitz, R.~Ostrovsky, M.~Prabhakaran, A.~Sahai \& J.~Wullschleger.
``Constant-rate oblivious transfer from noisy channels.''
In: P.~Rogaway, Ed.
\emph{Advances in cryptology --- CRYPTO 2011.}
Lecture Notes in Comp. Science Vol.~6841,
Springer-Verlag, Berlin, 2011, pp.~667-684.


\bibitem{JMV}
D.~Jungnickel, A.~Menezes \& S.~Vanstone
``On the number of self-dual bases of $GF(q^m)$ over $GF(q)$.''
Proc. AMS \textbf{109} (1990) 23-29.

\bibitem{KO}
W.~Kositwattanarerk \& F.~Oggier. 
``On construction D and related constructions of lattices from linear codes.''
Presented at WCC 2013, to appear in Des. Codes Crypto.
Preprint: \url{http://arxiv.org/abs/1308.6175}

\bibitem{Kot}
R.~K\"otter.
``A unified description of an error locating procedure for linear codes.''
Proc. Int. Workshop on Algebraic and Comb. Coding Theory, Voneshta Voda, Bulgaria, June 22-28, 1992.

\bibitem{Laurent}
M.~Laurent.
``Hauteur de matrices d'interpolation.''
In: \emph{Approximations diophantiennes et nombres transcendants (Luminy, 1990)}, de Gruyter, Berlin, 1992, pp.~215-238.


\bibitem{vLW}
J.~van Lint \& R.~Wilson.
``On the minimum distance of cyclic codes.''
IEEE Trans. Inform. Theory~\textbf{32} (1986) 23-40.

\bibitem{McEliece}
R.~McEliece.
``A public-key system based on algebraic coding theory.''
Deep Space Network Progress Report~\textbf{44} (1978) 114-116.

\bibitem{McWilliams}
F.~J.~McWilliams.
\emph{Combinatorial properties of elementary abelian groups.}
Ph.D. dissertation, Harvard University, Cambridge, Mass., 1962.

\bibitem{MR}
D.-J.~Mercier \& R.~Rolland.
``Polyn\^omes homog\`enes qui s'annulent sur l'espace projectif $\PP^m(\F_q)$.''
J.~Pure Appl. Algebra~\textbf{124} (1998) 227-240.

\bibitem{Mirandola}
D.~Mirandola.
\emph{Schur products of linear codes: a study of parameters.}
Master Thesis (under the supervision of G.~Z\'emor),
Univ. Bordeaux~1 \& Stellenbosch Univ., July 2012.
Online version: \url{http://www.algant.eu/documents/theses/mirandola.pdf}

\bibitem{Mumford}
D.~Mumford.
\emph{Lectures on curves on an algebraic surface.}
Annals of Math. Studies Vol.~59, Princeton University Press, Princeton, N.J., 1966.

\bibitem{Mumford_quad}
D.~Mumford.
``Varieties defined by quadratic equations.''
In: \emph{Questions on Algebraic Varieties (C.I.M.E., III Ciclo, Varenna, 1969)},
Ed. Cremonese, Rome, 1970, pp.~29-100.

\bibitem{OZ}
F.~Oggier \& G.~Z\'emor.
``Coding constructions for efficient oblivious transfer from noisy channels.''
In preparation.

\bibitem{P92}
R.~Pellikaan.
``On decoding by error location and dependent sets of error positions.''
Discrete Math.~\textbf{106/107} (1992) 369-381.

\bibitem{P96}
R.~Pellikaan.
``On the existence of error-correcting pairs.''
J.~Statist.~Plann.~Inference~\textbf{51} (1996) 229-242.

\bibitem{HSSdim0}
H.~Randriambololona.
``Hauteurs des sous-sch\'emas de dimension nulle de l'espace projectif.''
Ann. Inst. Fourier (Grenoble)~\textbf{53} (2003) 2155-2224.

\bibitem{ChCh+}
H.~Randriambololona.
``Bilinear complexity of algebras and the Chudnovsky-Chudnovsky interpolation method.''
J.~Complexity~\textbf{28} (2012) 489-517.

\bibitem{agis}
H.~Randriambololona.
``Asymptotically good binary linear codes with asymptotically good self-intersection spans.''
IEEE Trans. Inform. Theory~\textbf{59} (2013) 3038-3045.

\bibitem{Singleton}
H.~Randriambololona.
``An upper bound of Singleton type for componentwise products of linear codes.''
To appear in IEEE Trans. Inform. Theory.

\bibitem{21sep}
H.~Randriambololona,
``$(2,1)$-separating systems beyond the probabilistic bound.''
To appear in Israel J.~Math.

\bibitem{Roos}
C.~Roos.
``A new lower bound for the minimum distance of a cyclic code.''
IEEE Trans. Inform. Theory~\textbf{29} (1983) 330-332.

\bibitem{SSB}
G.~Schmidt, V.~Sidorenko \& M.~Bossert.
``Syndrome decoding of Reed-Solomon codes beyond half the minimum distance based on shift-register synthesis.''
IEEE Trans. Inform. Theory~\textbf{56} (2010) 5245-5252.

\bibitem{SL80}
G.~Seroussi \& A.~Lempel.
``Factorization of symmetric matrices and trace-orthogonal bases in finite fields.''
SIAM J.~Comput.~\textbf{9} (1980) 758-767.

\bibitem{SL84}
G.~Seroussi \& A.~Lempel.
``On symmetric algorithms for bilinear forms over finite fields.''
J. Algorithms~\textbf{5} (1984) 327-344.

\bibitem{SerreCL}
J.-P.~Serre.
\emph{Corps locaux.}
Actualit\'es Sci. Indust. No.~1296,
Hermann, Paris, 1968.

\bibitem{Serre82}
J.-P.~Serre.
``Nombres de points des courbes alg\'ebriques sur $\Fq$.''
S\'em. th\'eorie des nombres Bordeaux~\textbf{12} (1982/1983).

\bibitem{STV}
I.~Shparlinski, M.~Tsfasman \& S.~Vladut.
``Curves with many points and multiplication in finite fields.''
In:
H.~Stichtenoth \& M.~Tsfasman, Eds.
\emph{Coding theory and algebraic geometry (Luminy, 1991).}
Lecture Notes in Math. Vol.~1518,
Springer-Verlag, Berlin, 1992, pp.~145-169.

\bibitem{Singh}
S.~Singh.
``Analysis of each integer as sum of two cubes in a finite integral domain.''
Indian J. Pure Appl. Math.~\textbf{6} (1975) 29-35.

\bibitem{Slepian}
D.~Slepian.
``Some further theory of group codes.''
Bell Syst. Tech.~J.~\textbf{39} (1960) 1219-1252.

\bibitem{Sorensen}
A.~S{\o}rensen.
``Projective Reed-Muller codes.''
IEEE Trans. Inform. Theory~\textbf{37} (1991) 1567-1576.

\bibitem{Stichtenoth}
H.~Stichtenoth.
\emph{Algebraic function fields and codes.}
Graduate Texts in Math. Vol.~254, Springer-Verlag, Berlin, 2009.

\bibitem{TVbook}
M.~Tsfasman \& S.~Vladut.
\emph{Algebraic-geometric codes.}
Math. and its Appl. (Soviet Series) Vol.~58,
Kluwer Acad. Publishers Group, Dordrecht, 1991.

\bibitem{TVieee}
M.~Tsfasman \& S.~Vladut.
``Geometric approach to higher weights.''
IEEE Trans. Inform. Theory~\textbf{41} (1995) 1564-1588. 

\bibitem{Wei}
V.~Wei.
``Generalized Hamming weights for linear codes.''
IEEE Trans. Inform. Theory~\textbf{37} (1991) 1412-1418. 

\bibitem{Wieschebrink}
C.~Wieschebrink.
``Cryptanalysis of the Niederreiter public key scheme based on GRS subcodes.''
In: N.~Sendrier, Ed.
\emph{Post-quantum cryptography.}
Lecture Notes in Comp. Science Vol.~6061,
Springer-Verlag, Berlin, 2010, pp.~61-72.


\end{thebibliography}
\end{document}